\documentclass[11pt]{article}

\usepackage[english]{babel}
\usepackage{amsmath,amsthm}

\usepackage{latexsym} %\usepackage{epsf}
\usepackage{amsfonts}
\usepackage{amsmath}
\usepackage{amsthm}
\usepackage{amssymb}
\renewcommand{\theequation}{\thesection.\arabic{equation}}
\newcounter{subequation}[equation]
\makeatletter

\expandafter\let\expandafter\reset@font\csname reset@font\endcsname

\def\subeqnarray{\arraycolsep1pt
	\def\@eqnnum\stepcounter##1{\stepcounter{subequation}%
		{\reset@font\rm(\theequation\alph{subequation})}}
	\jot5mm     \eqnarray}

\makeatother

\def\str{\mathop{\hbox{\rm str}}\nolimits}

\def\be{\begin{equation}}
\def\ee{\end{equation}}

\def\lb{\label}

\def\bea{\begin{eqnarray}}
\def\eea{\end{eqnarray}}
\def\ba{\begin{array}}
	\def\ea{\end{array}}

\def\half{\frac{1}{2}}

\def\ket#1{\left|#1\right>}
\def\bra#1{\left<#1\right|}

\def\one#1{#1^{\raise5pt\hbox{$\scriptstyle\!\!\!\!1$}}\,{}}
\def\two#1{#1^{\raise5pt\hbox{$\scriptstyle\!\!\!\!2$}}\,{}}

\def\tilde{\widetilde}
\def\II{\hbox{{1}\kern-.25em\hbox{l}}}

\def\qed{\rule{5pt}{5pt}}
\makeatletter
\def\binrel@#1{\begingroup
	\setboxz@h{\thinmuskip0mu
		\medmuskip\m@ne mu\thickmuskip\@ne mu
		\setbox\tw@\hbox{$#1\m@th$}\kern-\wd\tw@
		${}#1{}\m@th$}%
	\edef\@tempa{\endgroup\let\noexpand\binrel@@
		\ifdim\wdz@<\z@ \mathbin
		\else\ifdim\wdz@>\z@ \mathrel
		\else \relax\fi\fi}%
	\@tempa
}
\let\binrel@@\relax
\def\overset#1#2{\binrel@{#2}%
	\binrel@@{\mathop{\kern\z@#2}\limits^{#1}}}
\def\underset#1#2{\binrel@{#2}%
	\binrel@@{\mathop{\kern\z@#2}\limits_{#1}}}
\makeatother
\newfont{\bbd}{msbm10 scaled\magstep1}

\newtheorem{proposition}{Proposition}

\begin{document}

	\begin{center}
		{\LARGE {Yangians and Yang-Baxter $R$-operators \\
for ortho-symplectic superalgebras}}
 %relations }}

 \vspace{0.3cm}

\large \sf J. Fuksa$^{a,d}$\footnote{\sc e-mail: fuksa@theor.jinr.ru},
A.P. Isaev$^{a}$\footnote{\sc e-mail: isaevap@theor.jinr.ru}, \\
D. Karakhanyan$^{b}$\footnote{\sc e-mail: karakhan@yerphi.am},
R. Kirschner$^c$\footnote{\sc e-mail: Roland.Kirschner@itp.uni-leipzig.de} \\

\vspace{0.5cm}

\begin{itemize}
\item[$^a$]
{\it Bogoliubov Laboratory of Theoretical Physics, JINR, Dubna, Russia}
\item[$^b$]
{\it Yerevan Physics Institute,
 2 Alikhanyan br., 0036 Yerevan, Armenia }
\item[$^c$]
{\it Institut f\"ur Theoretische
Physik, Universit\"at Leipzig, \\
PF 100 920, D-04009 Leipzig, Germany}
\item[$^d$]
{\it Faculty of Nuclear Sciences and Physical Engineering, \\
Czech Technical University in Prague, Czech Republic}
\end{itemize}
\end{center}
\vspace{0.5cm}
\begin{abstract}
\noindent
Yang-Baxter relations symmetric with respect to the ortho-symplectic
superalgebras are studied. We start from the formulation of graded
algebras
and the linear superspace carrying the vector (fundamental) representation of the
ortho-symplectic supergroup. On this basis we study the analogy of the
Yang-Baxter operators considered earlier for the cases of orthogonal and
symplectic symmetries: the vector (fundamental) $R$ matrix, the $L$ operator defining
the Yangian algebra and its first and second order evaluations.
We investigate the condition for $L(u)$ in the case of the truncated expansion in
inverse powers of $u$ and give examples of Lie algebra representations
obeying these conditions. We construct the $R$ operator intertwining two
super-spinor representations and study the fusion of $L$ operators
involving the tensor product of such representations.
\end{abstract}

\section{Introduction}

The orthogonal and the symplectic groups have similarities, which can be
traced back to the existence of an invariant bilinear form in the
fundamental representation space. The similarities allow a unified
treatment not only of the groups and Lie algebras but also of the
Yang-Baxter relations with such symmetries. The orthogonal and symplectic
groups are embedded in the ortho-symplectic ($OSp$) super-group.

In the present paper we study Yang-Baxter relations with ortho-symplectic
($osp$)
super-symmetry. In a recent paper \cite{IsKarKir15} Yang-Baxter relations
symmetric with respect to orthogonal or symplectic groups have been
considered in a unified formulation, continuing preceeding studies
\cite{Witten},\cite{KarT},\cite{ZamL},\cite{Re},\cite{Resh}
 and especially \cite{CDI,CDI2} where  orthogonal 
 and partially symplectic symmetries have been considered.
We rely on the similarities, proceed to some extend parallel to
\cite{IsKarKir15} and we follow similar motivations. The Yang-Baxter
operators, in particular their dependence on the spectral parameter,
are generally more involved as compared to the case of general linear
symmetry. Compact forms
of $L$ operators appear only in cases of particular representations.
Such cases are of interest in physical applications,
in particular for the theory of quantum integrable systems where the $L$ operators
are used as building blocks. The viewpoint of the Yangian algebra and its
finite order evaluations is appropriate to understand the cases of
 %simple
explicit solutions for $L$ operators.

The extension to the $osp$ supersymmetric case is of interest
in the studies of integrable spin chains of $osp$-type \cite{Ragoucy}
and exact factorizing $S$ matrices \cite{Zam} for two-dimensional
field theories with ortho-symplectic supersymmetries.
It is also relevant to investigations of gauge field theories, in particular because
of the conformal transformations
 %represent the orthogonal algebra.
 form the (pseudo)orthogonal groups.
 %The viewpoint of
 From the mathematical point of view it would
 be interesting to extend the Drinfeld's classification theorem
 for the finite-dimensional irreducible representations
 of the Yangians of the classical type
  (see \cite{MolRag} and references therein) to the case
 of the $osp$-type Yangians.
 %is useful in the
 %representation theory of super Lie algebras.

 %If one would not care about details,
 The generalization of $so$ and $sp$ cases
to the supersymmetric $osp$ case
 %could seem
 is not straightforward. %However,
 We shall see that it takes much attention and carefully prepared formulations
to obtain in a systematic way the correct forms of all the relations starting
from the  theory of super-matrices of $OSp$ and its super-algebra
 %relations
 up to the
Yang-Baxter relations and the forms of the $L$ operators related to the
first and second order evaluations of the Yangian $Y(osp)$.
 %Typically,
 Interesting and intriguing details are related to sign factors. For deriving
a particular relation one could choose the approach of starting from an
ansatz up to sign and then invent a test to fix the sign. For our purpose
that approach would be not sufficient, not only because of lack of elegance but mostly
because  it would cost more efforts then the derivations
of all formulas from the first principles.
 %in the frame of a complete formulation.
The complete formulation of Lie supergroups Osp, their Lie superalgebras
and Yangians $Y(osp)$ given in this paper provides a convenient basis for
further studies and applications.

We outline the contents emphasizing the main points of the paper.

We present a systematic and detailed super-symmetric
formulation of the Yang-Baxter relations and the involved operators.
To set up the formulation we start in Sect. 2 with the superspace where the
ortho-symplectic super-matrices act. Recalling the notion of invariant
tensors we prepare the formulation of the fundamental $R$-matrix
of $osp$-type. The
involved invariant operators acting in the twofold tensor product of the superspace
are shown to represent the Brauer centralizer algebra \cite{Brauer}
(see also \cite{Wenz},\cite{Nazar}).
 In Section 2 we explain a lot of
 %There are  intriguing
 details and useful notations important in our formulation.
 %The super-metric defining the invariant scalar product allows to rise and
 %lower indices. But unlike the conventional case the metric tensor with both
 %indices rised up does not coincide with the elements of the inverse of the matrix
 %composed from the tensor components with lower indices.
In particular we introduce the notion of the sign operator
 %turns out useful in many relations.
 which is useful in many relations.
 %In connection with operators acting
 The action of a super-group element $U$ on a higher tensor product of the
 fundamental superspaces is described by the tensor product of operators
 $U$ dressed by these sign operators.
 %is relevant.

Whereas the braid form of the fundamental Yang-Baxter relation
coincides for the
supersymmetric case with the case without supersymmetry,
factors of sign operators make the
difference when changing to the $R$ operators including super-permutation.
This is explained in Sect. 3. The modifications by factors of sign operators
appear in the $RLL$ relations as well as discussed in Sect. 4.

The graded $RLL$ relation involving the $osp$ supersymmetric fundamental $R$
matrix serves as the defining relation of the $L$ operators. It defines the
$osp$ Yangian  algebra following the known scheme described in Sect. 4:
By expanding $L(u)$ in inverse powers of the spectral parameter $u$ one
obtains the Yangian algebra generators $L^{(k)}_{ab}$. Substituting this
expansion in the $RLL$ relation one obtains all Yangian algebra relations.
The first non-trivial term in the $u^{-1}$ expansion involves the $osp$ Lie
superalgebra generators $L^{(1)}_{ab}$.
We formulate their %index
 symmetry relation and the Lie algebra commutation relations as
derived from the $RLL$ relation. For the super traceless part $G$ of the
 %generator
matrix $L^{(1)}$ the commutation relations are to be compared with the ones obtained
for the infinitesimal $OSp$ supermatrices formulated earlier in Sect. 2.
We show that the corresponding basis sets are transformed into each other
by the sign operator.

In Sect. 5 we show that
in the linear evaluation of $\mathcal{Y}(osp)$ the matrix $L^{(1)}$
of generators of $osp$ algebra has to
obey the additional constraint in form of a quadratic characteristic
identity.
 %We identify the super spinor representation of
 %$osp$ obeying this constraint and generalizing the
 We construct the infinite dimensional super spinor representation of $osp$ obeying
 this constraint as a generalization of the
spinor representation of $so$ and the metaplectic representation of $sp$.
We formulate %the generators of
 the super spinor representation in terms of
an algebra of super oscillators which generalizes the Clifford algebra
 and the standard multidimensional oscillator algebra.

In Sect. 6 we construct the super spinorial $R$ operator
which acts in the tensor product of two super spinor representations.
 Our systematic
supersymmetric formulation allows to extend to the
 $osp$ case the approach developed
for the $so$ case in \cite{CDI},\cite{CDI2}
and for the $sp$ case in \cite{IsKarKir15}.

The graded symmetrisation of products of super oscillator operators is
treated by the generating function techniques using auxiliary variables obeying
a graded multiplication law.

 %\marginpar{\it RK, shifted to Discussion}

The fusion procedure with superspinor $L$ opertors is considered in Sect. 7.
The projection of the tensor product of superspinor representations to the
fundamental one
applied to the product of such $L$ operator matrices
reproduces the fundamental $R$ matrix.

In Sect. 8 we investigate the quadratic evaluation of the Yangian. We
formulate the conditions implied by the $RLL$ relation with $L$ depending
quadratically on the spectral parameter. The sign operator helps to
formulate these conditions in a form similar to the ones in the cases
of orthogonal and symplectic symmetry. We construct the solution where the
second non-trivial term in $L$ is the quadratic polynomial in the generator
matrix appearing in the constraint for the first order evaluation. The
conditions imply constraints on the Lie superalgebra representation
which are fulfilled by the Jordan-Schwinger ansatz for the generators
in terms of graded Heisenberg pairs. In this case the generator matrix
obeys the super anticommutator constraint which appeared in Sect. 6 in the
superspinorial $RLL$ relation, and as a consequence, the generator matrix
obeys a cubic characteristic relation.

\section{The ortho-symplectic supergroup}
\label{supgr}
\setcounter{equation}{0}

\subsection{The superspace, the super-group OSp and its superalgebra osp}

To fix notations we need to formulate the notion of
 the ortho-symplectic supergroups and their Lie superalgebras
 (see, e.g., \cite{Ber} as an introduction to the super analysis and
 %\marginpar{\bf \large Is}
 to the theory of Lie supergroups and superalgebras).

We denote by ${\rm grad}(A)$ the  grading of the algebraic object $A$.
 Let ${\cal V}_{(N|M)}$ be a superspace and $z^a$ ($a=1,\dots,N+M$)  the
 graded coordinates in ${\cal V}_{(N|M)}$. We distinguish
 $N$ even and $M$ odd coordinates $z^a$ , denote the grading
 ${\rm grad}(z^a)$ of the coordinate~$z^a$ as
 $[a]=0,1 \, ({\rm mod}2)$
 and call $[a]$  the degree of the index $a$. If the coordinate $z^a$ is
  even then $[a]=0 \, ({\rm mod}2)$, and if the coordinate $z^a$ is
  odd then $[a]= 1 \, ({\rm mod}2)$. The coordinates $z^a$ and $w^a$ of two supervectors
 $z,w \in {\cal V}_{(N|M)}$ commute as
\begin{equation}
 \label{commzw}
 z^a \, w^b = (-1)^{[a]  [b]}  \, w^b \, z^a \; .
 \end{equation}
 Thus the coordinates of the  vectors
 $z,w, \dots \, \in {\cal V}_{(N|M)}$ can be considered as generators
of a graded algebra.
For two homogeneous elements $A,B$ of this algebra we have
 $$
 {\rm grad}(A \cdot B) = {\rm grad}(A) + {\rm grad}(B) \; , \;\;\;
 A \cdot B = (-1)^{{\rm grad}(A){\rm grad}(B)} \,
 B \cdot A \; .
 $$
Further we endow the superspace ${\cal V}_{(N|M)}$
 with the  bilinear form
\begin{equation}
\label{bf}
(z\cdot w) \equiv \varepsilon_{ab} z^a  w^b = z^a w_a =
z_b w_a \bar{\varepsilon}^{ab} \; ,
\end{equation}
where the super-metric $\varepsilon_{ab}$ has the property
\begin{equation} \label{SuperMetric1}
\varepsilon_{ab} = \epsilon (-1)^{[a][b]} \varepsilon_{ba} \;\; \Leftrightarrow \;\;
\bar{\varepsilon}^{ab} = \epsilon (-1)^{[a][b]} \bar{\varepsilon}^{ba} \; ,
\end{equation}
and the matrix $||\bar{\varepsilon}^{ab}||$ is inverse to the matrix
 $||\varepsilon_{ab}||$. Here $\epsilon=\pm 1$.
 Moreover, we require that
 the super-metric $\varepsilon_{ab}$ is even in the sense
that $\varepsilon_{ab}\neq 0$ iff $[a]+[b]=0 \; ({\rm mod}2)$. This  means that
\begin{equation}\label{SuperMet2}
\varepsilon_{ab}= (-1)^{[a]+[b]} \varepsilon_{ab} ,
\end{equation}
 and therefore the properties (\ref{SuperMetric1}) can be written as
 $\varepsilon_{ab} = \epsilon (-1)^{[a]} \varepsilon_{ba} =
 \epsilon (-1)^{[b]} \varepsilon_{ba}$ (the same for $\bar{\varepsilon}^{ab}$).
Taking into account (\ref{commzw})
 we obtain $(z\cdot w) = \epsilon \, (w \cdot z)$, i.e., the bilinear form (\ref{bf})
 is symmetric for $\epsilon = +1$ and skew-symmetric for $\epsilon = -1$.
Further we adopt
 the following rule for  rising and lowering indices,
 \begin{equation}
\label{agree}
z_a = \varepsilon_{ab} \, z^b \; , \;\;\; z^a =  \bar{\varepsilon}^{ab} \, z_b \; .
\end{equation}
 This rule will be applied also for any tensors of higher ranks.
According to this rule we have
 \begin{equation}
\label{agree1}
 \varepsilon^{ab} = \bar{\varepsilon}^{ac}  \bar{\varepsilon}^{bd} \varepsilon_{cd} =
 \bar{\varepsilon}^{ba} =  \epsilon (-1)^{[a][b]} \bar{\varepsilon}^{ab}  \; ,
 \end{equation}
 We see that the metric tensor with upper indices $\varepsilon^{ab}$
 does not coincide with the inverse matrix $\bar{\varepsilon}^{ab}$.
 Further we shall use only the inverse matrix $\bar{\varepsilon}^{ab}$ and never
 the metric tensor $\varepsilon^{ab}$.
 So to simplify formulas  below
 we omit the bar in the notation $\bar{\varepsilon}^{ab}$ and write
 simply $\varepsilon^{ab}$ keeping in mind that this is the inverse
 matrix for the metric $\varepsilon_{ab}$.\footnote{An alternative rule
  for lifting and lowering indices
 (instead of (\ref{agree})) is
 $z_a = \varepsilon_{ab} \, z^b$, $z^a =  z_b \,
 \varepsilon^{ba}$, where $\varepsilon^{ab}$ is the metric
 tensor (\ref{agree1})
 with upper indices. So, to lower the index we
 act by the matrix $\varepsilon_{ab}$ from the left, while to
  lift the index we act by the matrix $\varepsilon^{ba}$
 from the right. In this case one should remember the unusual
  conditions of inversion:
 $\varepsilon_{ab}\varepsilon^{ac}= \delta^c_b$.
 Note that is not the rule adopted below}

Consider a linear transformation in ${\cal V}_{(N|M)}$
\begin{equation}
\label{transU}
z^a \rightarrow z'^a=U^a_{\,\ b} \, z^b \; ,
\end{equation}
which preserves the
grading of the coordinates ${\rm grad}(z'^a) =  {\rm grad}(z^a)$. For the
elements $U^a_{\ b}$ of the supermatrix $U$
from (\ref{transU}) we have ${\rm grad}(U^a_{\ b}) =  [a] +[b]$.
The ortho-symplectic group $OSp$ is defined as the set of
supermatrices $U$ which
preserve the bilinear form (\ref{bf}) with respect to the
transformations (\ref{transU}):
$$
\varepsilon_{ab}\; z'^a w'^b=\varepsilon_{ab}   U^a_{\ c} z^c   U^b_{\ d} w^d = (-1)^{[c]([b]+[d])} \varepsilon_{ab}  U^a_{\ c}   U^b_{\ d} z^c w^d =\varepsilon_{cd}  z^c w^d
\;\;\; \Rightarrow
$$
\begin{equation} \label{eq:OSpDef}
(-1)^{[c]([b]+[d])} \varepsilon_{ab} U^a_{\ c} U^b_{\ d} =  \varepsilon_{cd}
\;\; \Leftrightarrow \;\;
(-1)^{[c]([b]+[d])}  U^a_{\ c} U^b_{\ d} \varepsilon^{cd} =  \varepsilon^{ab} \; .
\end{equation}
We have to mention that the
first equation in \eqref{eq:OSpDef} can be also rewritten
by using the properties of the supermetric as
$$
(-1)^{[c]([a]+[c])} \varepsilon_{ab} U^a_{\ c} U^b_{\ d} =  \varepsilon_{cd}.
$$
We  represent \eqref{eq:OSpDef} in coordinate-free form  as
\begin{equation} \label{eq:OSpDef1}
\varepsilon_{\langle 12} \; U_1(-)^{12} U_2 (-)^{12} = \varepsilon_{\langle 12}
\;\; \Leftrightarrow \;\;
U_1(-)^{12} U_2 (-)^{12} \; \varepsilon^{12 \rangle} = \varepsilon^{12 \rangle} \; ,
\end{equation}
 where we have used the concise matrix notations
 \begin{equation}
 \label{min12}
 \begin{array}{c}
 \varepsilon^{12 \rangle} \in {\cal V}_{(N|M)} \otimes {\cal V}_{(N|M)} \; , \;\;\;
  U_1 = U \otimes I \; , \;\;\;
  U_2 = I \otimes U  \; , \\ [0.3cm]
  ((-)^{12})^{a_1 a_2}_{\;\;\; b_1 b_2} = (-1)^{[a_1][a_2]}
  \delta^{a_1}_{\;\; b_1} \delta^{a_2}_{\;\; b_2} \; , \;\;\;
  (-)^{12} \in {\rm End}({\cal V}_{(N|M)} \otimes {\cal V}_{(N|M)}) \; .
  \end{array}
  \end{equation}
  Here $\otimes$ denotes the graded tensor product:
  \be \label{gradten}
  (I \otimes B)(A \otimes I) = (-1)^{[A] \, [B]} \, (A \otimes B) \; ,
  \ee
  where $[A] := {\rm grad}(A)$ and $[B] := {\rm grad}(B)$.
   The sign operator $(-)^{12}$ is an extremely
  %\marginpar{\bf \large Is, \it RK}
   useful tool for the $R$-matrix formulation \cite{FRT}
  of quantum supergroups and their particular
  cases as super-Yangians. This operator was  used first in \cite{Isae}.
For more details about  the graded tensor product and
its consequences for the record of equations,
 especially the Yang-Baxter equation,
see  appendix {\bf \ref{sec:GTP}}.\footnote{Let us remark  shortly  that we use
the convention that the gradation is carried by the coordinates.
 There is the opposite convention in which the gradation is carried by the
basis vectors. The relation of expressions in both conventions is explained
in appendix \ref{sec:GTP}.
Our convention is of great importance for the formulation of the
quantum hyperplane for Lie super-algebras and their quantum deformations.}
 We also describe further details of this formalism in the next subsection.

The set of supermatrices $U$ which satisfy
 %\marginpar{\bf \large Is}
equations \eqref{eq:OSpDef} form the supergroup $OSp(N|M)$
for $\epsilon = +1$ and the supergroup $OSp(M|N)$
for $\epsilon = -1$ with respect to the usual matrix
multiplication.
Let us consider the defining relations \eqref{eq:OSpDef} for
 $OSp$ groups elements close to the unit element
$U  = I + A + \dots$. As a result we obtain the conditions
for the elements $A$ of the Lie superalgebra $osp$ of the
supergroup $OSp$:
\begin{align}
 \label{osp}
(-1)^{[c]([b]+[d])} \varepsilon_{ab} ( \delta^a_{\ c} A^b_{\ d}  + A^a_{\ c} \delta^b_d ) =
\left( (-1)^{[c]+[c][d]} \varepsilon_{cb} A^b_{\ d}
 + \varepsilon_{ad}  \, A^a_{\ c} \right) =0 \; ,
\end{align}
or equivalently
\begin{equation} \label{OspAlg}
A_{cd}=-\epsilon(-1)^{[c][d]+[c]+[d]} A_{dc} \; .
\end{equation}
The coordinate free form of (\ref{osp}), (\ref{OspAlg})
can be obtained directly from (\ref{eq:OSpDef1}):
 \begin{equation} \label{OspAlg1}
 \varepsilon_{\langle 12} (A_1 + (-)^{12}A_2 (-)^{12})
 = 0 \;\; \Leftrightarrow \;\;
  (A_1 + (-)^{12}A_2 (-)^{12})  \varepsilon^{12 \rangle} = 0 \; .
\end{equation}
 {\bf Remark.} The set of super-matrices $A$,
 which satisfy (\ref{osp}), (\ref{OspAlg1}), forms a vector space over $\mathbb{C}$
  denoted as $osp$. One can check that for two super-matrices
 $A,B \in osp$ the commutator
 \begin{equation}
 \label{osp05}
 [A,B]=AB - BA \; ,
 \end{equation}
 also obeys (\ref{osp}), (\ref{OspAlg1})
 and thus belongs to the vector space $osp$. It means that $osp$ is an algebra.
 Any matrix $A$ which satisfies (\ref{osp}), (\ref{OspAlg1}) can be represented as
\begin{equation}
 \label{osp02}
A^a_{\;\; c} = E^a_{\;\; c} -
 (-1)^{[c]+[c][d]} \varepsilon_{cb} E^b_{\;\; d} \varepsilon^{da}
 \end{equation}
where $||E^a_{\;\; c}||$ is an arbitrary  matrix.
Let $\{e^{\;\; f}_{g} \}$ be the matrix units,
$(e^{\;\; f}_{g})^b_{\;\; d}= \delta^f_d \delta^b_g$.
 If we choose $E = e^{f}_{\;\; g} = \varepsilon^{f g'}
 \varepsilon_{g f'} e^{\;\;\; f'}_{g'}$  in (\ref{osp02}) then
 we obtain the basis $\{\tilde{G}^f_{\;\; g}\}$ in the space $osp$
of matrices (\ref{OspAlg1}):
 \begin{equation}
 \label{osp03}
 \begin{array}{c}
 (\tilde{G}^f_{\;\; g})^a_{\;\; c} \equiv (e^f_{\;\; g})^a_{\;\; c} -
 (-1)^{[c]+[c][d]} \varepsilon_{cb} (e^f_{\;\; g})^b_{\;\; d} \varepsilon^{da}
 = \varepsilon^{f a} \varepsilon_{g c} - \epsilon (-1)^{[c][a]}
 \delta^f_c \delta^a_g \; .
 \end{array}
\end{equation}
 Now any super-matrix $A \in osp$ which satisfy (\ref{osp}), (\ref{OspAlg1})
can be expanded  in this basis  (\ref{osp03})
 \begin{equation}
 \label{osp11}
 A^a_{\;\; c} = a^g_{\;\; f} (\tilde{G}^f_{\;\; g})^a_{\;\; c} \; ,
 \end{equation}
 where $a^g_{\;\; f}$ are super-coefficients.
Since the elements $(\tilde{G}^f_{\;\; g})^a_{\;\; c}$ are even, i.e.,
$(\tilde{G}^f_{\;\; g})^a_{\;\; c} \neq 0$ iff $[f]+[g]+[a]+[c] = 0 \; ({\rm mod}2)$,
then from ${\rm grad}(A^a_{\;\; c}) =[a]+[c]$
we obtain that ${\rm grad}(a^g_{\;\; f}) =[g]+[f]$. This means that the
ordinary commutator
(\ref{osp05}) can rewritten as
 $$
 [A,B]^a_{\;\; c} = [ a^g_{\;\; f} (\tilde{G}^f_{\;\; g}) , \;
  b^n_{\;\; k} (\tilde{G}^k_{\;\; n})]^a_{\;\; c} =
 a^g_{\;\; f} b^n_{\;\; k}  \;
 \bigl([\tilde{G}^f_{\;\; g} , \; \tilde{G}^k_{\;\; n}]_{\pm} \bigr)^a_{\;\; c}  \; ,
 $$
where we define the super-commutator
 \begin{equation}
 \label{osp06b}
 \bigl( [\tilde{G}^{a_1}_{\;\; b_1} , \; \tilde{G}^{a_2}_{\;\; b_2}]_{\pm}
 \bigr)^{a_3}_{\;\; c_3}  \equiv
 (\tilde{G}^{a_1}_{\;\; b_1})^{a_3}_{\;\; b_3} \,
 (\tilde{G}^{a_2}_{\;\; b_2})^{b_3}_{\;\; c_3}
  - (-1)^{([a_1]+[b_1])([a_2]+[b_2])}
  (\tilde{G}^{a_2}_{\;\; b_2})^{a_3}_{\;\; b_3}
  (\tilde{G}^{a_1}_{\;\; b_1})^{b_3}_{\;\; c_3} \; .
  \end{equation}
  By using the explicit representation (\ref{osp03}) for
  $\tilde{G}^{a}_{\;\; b}$
  one can evaluate the supercommutator (\ref{osp06b})
  and obtain the defining relations for the Lie superalgebra $osp$:
 \begin{equation}
 \label{osp06}
 \begin{array}{c}
  [\tilde{G}^{a_1}_{\;\; b_1} , \; \tilde{G}^{a_2}_{\;\; b_2}]_{\pm}
  = -(-1)^{[a_1][a_2]+[b_1][a_2]} \varepsilon^{a_1a_2} \tilde{G}_{b_1b_2} +
  \epsilon (-1)^{[b_1][a_2]} \delta^{a_2}_{b_1} \, \tilde{G}^{a_1}_{\;\;\; b_2} + \\ [0.3cm]
  + (-1)^{[a_1][a_2]+[b_1][a_2]} \varepsilon_{b_1b_2} \tilde{G}^{a_2a_1} -
  \epsilon (-1)^{[a_1]([b_1]+[a_2])+[b_1][a_2]} \delta^{a_1}_{b_2} \,
  \tilde{G}^{a_2}_{\;\;\; b_1} \; .
 \end{array}
\end{equation}
 In component free form we write (\ref{osp06})
 as following
  \begin{equation}
 \label{osp06a}
 [(-)^{12} \tilde{G}_{13} (-)^{12} \; , \; \tilde{G}_{23}] =
 [\epsilon {\cal P}_{12} - {\cal K}_{12}   \; , \; \tilde{G}_{23}] \; ,
\end{equation}
 where we have introduced matrices the ${\cal K},{\cal P} \in
 {\rm End}({\cal V}_{(N|M)}^{\otimes 2})$:
\begin{equation}
 \label{osp07}
 {\cal K}^{a_1 a_2}_{b_1 b_2} = \varepsilon^{a_1 a_2} \varepsilon_{b_1 b_2} \; ,\; \;\;\;
 {\cal P}^{a_1 a_2}_{b_1 b_2} =
 (-1)^{[a_1][a_2]} \delta^{a_1}_{b_2} \delta^{a_2}_{b_1} \; .
\end{equation}
The matrix ${\cal P}$ is called the superpermutation operator
 since it permutes super-spaces. For example,
using this matrix one can write (\ref{commzw}) as
${\cal P}^{ab}_{cd}w^c z^d = z^a w^b$.
The properties of the operators ${\cal P}$ and ${\cal K}$
 which will be used below are listed in
the appendix {\bf \ref{PKprop}}.

 Note that the $osp$ conditions (\ref{OspAlg1}) for the basis matrices
$(\tilde{G}^{f}_{\;\; g})^{a}_{\;\; c}$ are represented in the form
 \begin{equation}
 \label{osp08b}
 {\cal K}_{12} (\tilde{G}_{31} + (-)^{12} \tilde{G}_{32} (-)^{12}) = 0 \; , \;\;\;
(\tilde{G}_{31} + (-)^{12} \tilde{G}_{32} (-)^{12})  {\cal K}_{12}  = 0 \; ,
 \end{equation}
  Taking into account their property
 $(\tilde{G}^{a_3}_{\;\; c_3})^{a_1}_{\;\; c_1}=
(-1)^{[a_1]}(\tilde{G}^{a_1}_{\;\; c_1})^{a_3}_{\;\; c_3}
 (-1)^{[c_3]}$ evident from
  (\ref{osp03}) one can rewrite (\ref{osp08b}) as
  \begin{equation}
 \label{osp08}
 {\cal K}_{12} ((-)^{12} \tilde{G}_{13}(-)^{12}  + \tilde{G}_{23} ) = 0 \; , \;\;\;
((-)^{12} \tilde{G}_{13}(-)^{12} + \tilde{G}_{23})
 {\cal K}_{12}  = 0 \; .
 \end{equation}
 Moreover one can check the identities (see appendix
 {\bf \ref{PKprop}})
 \begin{equation}
 \label{osp09}
 {\cal P}_{12} (-)^{12} \tilde{G}_{13} (-)^{12} =
  \tilde{G}_{23} {\cal P}_{12}  \; , \;\;\;
(-)^{12} \tilde{G}_{13}  (-)^{12}  {\cal P}_{12} =
 {\cal P}_{12} \tilde{G}_{23}  \; .
 \end{equation}
 By using (\ref{osp08}) and (\ref{osp09})
we can rewrite relations (\ref{osp06a}) as
 \begin{equation}
 \label{osp10}
  [(-)^{12} \tilde{G}_{13} (-)^{12} \;  , \; \tilde{G}_{23}]
   = [\epsilon {\cal P}_{12} - {\cal K}_{12}  \; , \;
   (-)^{12} \tilde{G}_{13} (-)^{12}] \; .
 \end{equation}
This shows that the $osp$ defining relations (\ref{osp06})
 can be written in several  equivalent forms.

\subsection{$OSp$ super-group invariants \label{invar}}

Here we give the formulation of invariance
 %\marginpar{\bf \large Is}
of tensors with respect to
actions of a supergroup. This formulation will be important for our
discussion below of the invariance of the $osp$-type $R$-matrices which are
solutions of Yang-Baxter equations.

First we recall the case of an ordinary group.
Let $G$ be
a matrix group which acts in the vector space ${\cal V}$.
 We say that the tensor $\mathcal{O}^{a_1\cdots a_k}_{b_1\cdots b_j}$
with $k$ upper and $j$ lower indices is invariant w.r.t. a group $G$,
 if for all $U\in G$
\begin{equation} \label{eq:GroupInvariance}
U^{a_1}_{\ c_1} \cdots U^{a_k}_{\ c_k} \mathcal{O}^{c_1\cdots c_k}_{b_1\cdots b_j} =\mathcal{O}^{a_1\cdots a_k}_{d_1\cdots d_j} U^{d_1}_{\ b_1} \cdots U^{d_j}_{\ b_j}
\end{equation}
This is represented at the level of its Lie algebra $\mathcal{G}$ as:
  For all $A\in\mathcal{G}$ holds that
\begin{equation}
\sum_{i=1}^k A^{a_i}_{\ c_i} \mathcal{O}^{a_1\cdots c_i \cdots a_k}_{b_1\cdots b_j} - \sum_{i=1}^j \mathcal{O}^{a_1 \cdots a_k}_{b_1\cdots d_i \cdots b_j} A^{d_i}_{\ b_i}=0.
\end{equation}
If the tensor $\mathcal{O}$ is an operator in ${\cal V}^{\otimes k}$,
 we have $j=k$ and both conditions can be written in the concise form as
\begin{equation}
U_1\cdots U_k \mathcal{O}_{1\dots k} =  \mathcal{O}_{1\dots k} U_1\cdots U_k, \qquad
\left[ \mathcal{O}, A_1+ \cdots +A_k \right] = 0 ,
\end{equation}
where $1,\dots,k$ are labels of super-vector spaces.

Now we extend the above formulation to the case of a supergroup $G$ which acts
in a superspace ${\cal V}$.
Let us start with looking for a compact expression describing an action
of a supergroup $G$ on  graded tensor products
$x^{1 \rangle} y^{2\rangle}\equiv x\otimes y$, where $x,y \in {\cal V}$.
First of all,
let us understand the meaning  of the  action on it by the tensor product of
two superoperators $A_1 B_2\equiv A\otimes B$
\begin{equation}
\label{ABxy}
Ax\otimes By = (-1)^{[B][x]}(A\otimes B)(x\otimes y) \quad \Leftrightarrow
\quad A_1 x^{1 \rangle} \; B_2 y^{2 \rangle} =
A_1(-)^{12} B_2 (-)^{12} \; x^{1 \rangle} y^{2 \rangle}.
\end{equation}
Here we introduce the operator $(-)^{12}$
acting on the tensor product of two homogeneous vectors as
\begin{equation} \label{eq:GradOp}
(-)^{12} \; x^{1 \rangle} y^{2 \rangle} \equiv (-1)^{[x][y]} \; x^{1 \rangle} y^{2 \rangle}.
\end{equation}
In coordinates the operator $(-)^{12}$ has been given above in
(\ref{min12}).
We call it {\it the sign operator}. It satisfies
\begin{equation}
\label{siden}
(-)^{12}(-)^{23}=(-)^{23}(-)^{12} \; , \;\;\;
(-)^{21}= (-)^{12} \; , \;\;\; \left( (-)^{12} \right)^2=\mathbf{1} \; .
\end{equation}

 The formula (\ref{ABxy}) can be  generalized obviously for higher
 numbers of vector spaces in the tensor product ${\cal V}^{\otimes k}$.
 Let an operator $C$ act the  $j$-th super-space in
 ${\cal V}^{\otimes k}$, then
$$
x^{1\rangle} y^{2\rangle} \cdots C_j z^{j\rangle} \cdots w^{k\rangle} =
C_{\{1\dots j\}}  \  x^{1\rangle} y^{2\rangle} \cdots z^{j\rangle} \cdots w^{k\rangle} \; ,
 $$
 where we define the operator dressed by sign operators as
 \begin{equation}
\label{signop}
C_{\{1\dots j\}}
 \equiv (-)^{j-1,j} \cdots (-)^{2j} \cdot (-)^{1j} \cdot C_j \cdot
  (-)^{1j} \cdot (-)^{2j} \cdots (-)^{j-1,j} \; ,
 \end{equation}
 and $(-)^{\ell,j} = (-)^{\ell j}$ denotes the sign operator which acts nontrivially
 only in $\ell$-th and $j$-th factors of the tensor product ${\cal V}^{\otimes k}$.
 Then the formula (\ref{ABxy}) is generalized as
 \begin{equation}
 \label{s-act}
 A_1 x^{1\rangle} \,  B_2 y^{2\rangle}  \cdots C_j z^{j\rangle}  \cdots D_k w^{k \rangle}  =
 \left( A_1 \cdot B_{\{12\}} \cdots C_{\{1\dots j\}} \cdots D_{\{1 \dots k\}} \right)
 x^{1\rangle} y^{2\rangle} \cdots z^{j\rangle} \cdots w^{k\rangle} \; .
 \end{equation}

We specify now the super-operators to the elements of the super-group $U \in G$.
Taking into account (\ref{ABxy})
the action of $U$ in ${\cal V}^{\otimes 2}$
 can be expressed using the sign operators $(-)^{12}$
\begin{eqnarray}
&x'^a\otimes y'^c=U^a_{\ b}x^b \otimes U^c_{\ d}y^d=(-1)^{[b]([c]+[d])} \,
(U^{a}_{\ b}\otimes U^c_{\ d})(x^b\otimes y^d) = \notag\\ [0.3cm]
 & = U^a_{\ r} \; ((-)^{12})^{rc}_{bs} \; U^s_{\ t} \; ((-)^{12})^{bt}_{ud} \; x^u y^d
\;\; \Leftrightarrow \;\;
 x^{\prime \, 1 \rangle} y^{\prime \, 2 \rangle}= U_1 (-)^{12} U_2 (-)^{12} \,
 x^{1 \rangle} y^{2 \rangle} \; .
\end{eqnarray}

Let $\hat{R}_{12}$ be an operator acting in $ {\cal V}^{\otimes 2}$.
 The invariance of $\hat{R}_{12}$ w.r.t. the action of  $G$ is expressed as
\begin{equation}
\label{Rinvar}
\hat{R}_{12} U_1 (-)^{12} U_2 (-)^{12} = U_1 (-)^{12} U_2 (-)^{12} \hat{R}_{12} \; .
\end{equation}
 Note that for matrix $R_{12} = {\cal P}_{12} \hat{R}_{12}$,
 where ${\cal P}_{12}$ is the superpermutation  (\ref{osp07}), the invariance condition modifies to
 \begin{equation}
\label{Rinvar2}
R_{12} U_1 (-)^{12} U_2 (-)^{12} = (-)^{12} U_2 (-)^{12} U_1  R_{12} \; .
\end{equation}
The generalization  of \eqref{Rinvar}  for an  operator $\mathcal{O}_{1\dots k}$ acting %in arbitrary number $k$ of spaces in tensor product
 in ${\cal V}^{\otimes k}$ for arbitrary $k$ is obtained straightforwardly:
 %\begin{align}
 \begin{equation}
\label{SGrInvOp}
\mathcal{O}_{1\dots k} \Bigl( \prod_{j=1}^k \; U_{\{1 \dots j\}} \Bigl)
% (-)^{1j} (-)^{2j}\cdots (-)^{j-1,j} U_j (-)^{1j} (-)^{2j}\cdots (-)^{j-1,j}
= \Bigl( \prod_{j=1}^k U_{\{1 \dots j\}} \Bigr)
%(-)^{1j} (-)^{2j}\cdots (-)^{j-1,j} U_j (-)^{1j} (-)^{2j}\cdots (-)^{j-1,j}
 \; \mathcal{O}_{1\dots k} \; , \;\;\;\;
 \prod_{j=1}^k  \; U_{\{1 \dots j\}} \equiv U_1 \cdot U_{\{12\}} \cdots U_{\{1\dots k\}} \; .
 %\end{align}
 \end{equation}
 Then, for corresponding Lie super-algebra $\mathcal{G}$ we obtain $(\forall A\in \mathcal{G})$:
\begin{equation}
 \label{AOAO}
\Bigl[ O_{1\dots k}\ , \ \sum_{j=1}^k A_{\{1 \dots j\}}
%(-)^{1j} (-)^{2j}\cdots (-)^{j-1,j} A_j (-)^{1j} (-)^{2j}\cdots (-)^{j-1,j}
\Bigr] = 0 \; .
\end{equation}

The considerations above were done  for operators only.
They can be generalized to general tensors.
With slight abuse of the notation  a general tensor
$T\in {\cal V}^{\otimes k} \otimes \overline{\cal V}^{\otimes j}$
with $k$ upper and $j$ lower indices can be written as
$T^{1\dots k \rangle}_{\;\; \langle 1\dots j}$.
Here $\overline{\cal V}$ denotes the superspace  dual to ${\cal V}$.
In this notation the operator $O_{1\dots k}$ is written as
 $O^{1\dots k \rangle}_{\;\; \langle 1\dots k}$.
 %, the operator $(-)^{12}$ should be written
 %as $(-)^{12}_{12}$ etc. On the other hand, the super-metric
 %$\varepsilon_{12}$ is indexed correctly.
Then the invariance of the tensor $T$ w.r.t. the super-group $G$ is expressed as
\begin{equation} \label{eq:SGI}
T^{1\dots k \rangle}_{\;\; \langle 1\dots j} \;
\Bigl( \prod_{i=1}^j \; U_{\{1\dots i\}} \Bigr)
 %(-)^{1i} (-)^{2i}\cdots (-)^{i-1,i} U_i (-)^{1i} (-)^{2i}\cdots (-)^{i-1,i}
 = \Bigl( \prod_{i=1}^k U_{\{1\dots i\}} \Bigr) \;
 %(-)^{1i} (-)^{2i}\cdots (-)^{i-1,i} U_i (-)^{1i} (-)^{2i}\cdots (-)^{i-1,i} \;
 T^{1\dots k \rangle}_{\;\; \langle 1\dots j}.
\end{equation}
The infinitesimal form of (\ref{eq:SGI})
which generalizes (\ref{AOAO}) is
\begin{equation}
 \label{TATA}
T^{1\dots k \rangle}_{\;\; \langle 1\dots j}  \Bigl( \sum_{i=1}^j A_{\{1 \dots i\}} \Bigr)=
 \Bigl( \sum_{i=1}^k \; A_{\{1 \dots i\}} \Bigr) T^{1\dots k \rangle}_{\;\; \langle 1\dots j}
 \; .
\end{equation}
We see that the relations of the super-metric invariance
\eqref{eq:OSpDef1} and \eqref{OspAlg1}
 are just special cases of the general formulas (\ref{eq:SGI})
and  (\ref{TATA}) for $j=2,k=0$ and for $j=0,k=2$.

\section{The fundamental R-matrix and the graded Yang--Baxter equation\label{sec2}}
\setcounter{equation}{0}

 %We construct the R-matrix out of
 We have  three  $OSp$ invariant operators in ${\cal V}_{(N|M)}^{\otimes 2}$:
 the identity operator $\mathbf{1}$, the super-permutation operator $\mathcal{P}$
and the tensor $\mathcal{K}$.
 %which is composed from two metric tensors $\varepsilon^{a_1a_2}$
 %and $\varepsilon_{b_1b_2}$.
 The super-permutation $\mathcal{P}_{12}$ is a product of the usual
 permutation $P_{12}$ and the sign operator $(-)^{12}$
\begin{equation}
 \label{PP12}
\mathcal{P}_{12} = (-)^{12} P_{12}
\qquad \text{and in coordinates} \qquad
\mathcal{P}^{a_1a_2}_{b_1b_2}=(-1)^{[a_1][a_2]}\delta^{a_1}_{b_2} \delta^{a_2}_{b_1}.
\end{equation}
The operator $\mathcal{K}_{12}$ is defined as
 \begin{equation}
 \label{KK12}
\mathcal{K}_{12} = \varepsilon^{12 \rangle} \varepsilon_{\langle 12} \qquad \text{and in coordinates} \qquad\mathcal{K}^{a_1a_2}_{b_1b_2}=\varepsilon^{a_1a_2}
 \varepsilon_{b_1b_2}.
 \end{equation}
 In coordinates these operators $\mathcal{P},\mathcal{K}$ have been  introduced
 above in (\ref{osp07}).
Their invariance w.r.t. $OSp$ can be proved directly by
applying the results of above subsection {\bf \ref{invar}}.

Using the operators $\mathcal{P},\mathcal{K}$ one can construct
 the set of operators $\{s_i,e_i|{i=1,\dots,n-1}\}$ in
${\cal V}_{(N|M)}^{\otimes n}$:
\begin{equation}
 \label{Brauer}
s_i=\epsilon \mathcal{P}_{i,i+1},\qquad e_i=\mathcal{K}_{i,i+1},\qquad i=1,\dots,n-1,
\end{equation}
which generate the Brauer algebra $B_n(\omega)$ \cite{Brauer} with the parameter
\begin{equation}
\label{oMN}
\omega=\varepsilon^{cd}\varepsilon_{cd}
=\epsilon (N-M) \; .
%=\epsilon(\dim(\text{bosons})-\dim(\text{fermions})).
\end{equation}
Here $N$ and $M$ are the numbers of even and odd coordinates, respectively.
 Indeed, one can check directly
 (see appendix {\bf \ref{PKprop}})
 that the operators (\ref{Brauer}) satisfy
 the defining relations for generators of $B_n(\omega)$ (see, e.g.,
 \cite{Wenz}, \cite{Nazar}, \cite{IsMol10} and references therein)
 \begin{equation}
 \label{defBrauer1}
 \begin{array}{c}
s^2_i = 1 \; , \;\;\;  e^2_i = \omega e_i \; , \;\;\;
s_i \, e_i = e_i \, s_i = e_i\; , \;\;\; i = 1, . . . , n -1 , \\
s_i s_j = s_j s_i \; , \;\;\; e_i e_j = e_j e_i \; , \;\;\; s_i e_j = e_j s_i
\; , \;\;\;  |i - j| > 1,
 \end{array}
\end{equation}
\begin{equation}
 \label{defBrauer2}
 \begin{array}{c}
s_i \, s_{i+1} \, s_i = s_{i+1} \, s_i \, s_{i+1} \; , \;\;\; e_i \, e_{i+1} \, e_i = e_i
\; , \;\;\;  e_{i+1} \, e_i \, e_{i+1} = e_{i+1} \; , \\
s_i \, e_{i+1} \, e_i =
 s_{i+1} \, e_i \; , \;\;\; e_{i+1}\, e_i\, s_{i+1} = e_{i+1} \, s_i \; , \;\;\;
  i = 1, . . . , n-2\; .
\end{array}
\end{equation}
Thus, one can consider eqs. (\ref{Brauer}) as the matrix representation $T$ of the
 generators of the Brauer algebra $B_n(\omega)$
  in the space ${\cal V}_{(N|M)}^{\otimes n}$. We note that
  $T$ (\ref{Brauer}) is the special reducible representation
  of $B_n(\omega)$. Irreducible representations of Brauer algebra
  were investigated in many papers (see, e.g.
  \cite{Nazar}, \cite{IsMol10}, \cite{IsMolOg} and references therein).

Let us consider the following linear combination of the generators
$s_i,e_i \in B_n(\omega)$
 \begin{equation}
 \label{RBrauer}
\check{\rho}_i(u) =u(u+\beta) \, s_i- (u+\beta) \mathbf{1}+ u \, e_i \;\; \in \;\;
 B_n(\omega) \; ,
\end{equation}
where $u$ is a spectral parameter and
\begin{equation}
\beta=1-\frac{\omega}{2} \; .
\end{equation}
 %to link our results to  results in \cite{IsKarKir15}.. By using defining relations (\ref{defBrauer})
 %One can prove directly that
\noindent
\begin{proposition} \label{Prop11}
The element (\ref{RBrauer}) satisfies the Yang-Baxter equation
 \begin{equation}
 \label{YBEbr}
\check{\rho}_i(u) \check{\rho}_{i+1}(u + v) \check{\rho}_i(v) =
\check{\rho}_{i+1}(v) \check{\rho}_{i}(u + v) \check{\rho}_{i+1}(u)  \; ,
\end{equation}
and the unitarity condition
 \begin{equation}
 \label{unitar}
\check{\rho}_i(u) \check{\rho}_{i}(-u)  =
(u^2-1)(u^2 - \beta^2) \mathbf{1}  \; .
\end{equation}
\end{proposition}
\noindent
{\bf Proof.} Substitute (\ref{RBrauer}) into (\ref{YBEbr}), (\ref{unitar}).
 One obtains 27 terms in both sides of (\ref{YBEbr}).
 The terms which do not
 contain the elements $e_i,e_{i+1}$ are cancelled due
 to the first identity in (\ref{defBrauer2}) (the remaining
 terms just cancel each  other identically).
 Other terms which include the
 elements $e_i,e_{i+1}$ are cancelled due to the identities (\ref{defBrauer1}), (\ref{defBrauer2}).
 In particular to prove (\ref{YBEbr})  identities like
 $e_i s_{i \pm 1} e_i = e_i$ and $s_i e_{i + 1} s_i = s_{i+1} e_i s_{i+1}$
 which follow from (\ref{defBrauer1}),
 (\ref{defBrauer2}) are useful.  \hfill \qed

\noindent
One can consider equation (\ref{YBEbr})
as the nontrivial identity in the algebra $B_n(\omega)$.
 \vspace{0.2cm}

The matrix representation
 $T$ (\ref{Brauer}) of the element (\ref{RBrauer}) is
 \begin{equation}
\label{RBrauer01}
\check{R}(u) \equiv \epsilon \, T(\hat{\rho}(u)) =
u(u+\beta)\mathcal{P}-\epsilon(u+\beta)\mathbf{1}+\epsilon u\mathcal{K} \; .
\end{equation}
Here we suppress the index $i$ for simplicity.
 (\ref{YBEbr}) implies (see appendix {\bf \ref{PKprop}})
that $\hat{R}(u)$ satisfies the braid version of the Yang--Baxter equation
\begin{equation} \label{eq:YBEbraid}
\check{R}_{12}(u-v) \check{R}_{23}(u) \check{R}_{12}(v) =
\check{R}_{23}(v) \check{R}_{12}(u) \check{R}_{23}(u-v).
\end{equation}
 Further we  use the  $R$-matrix
\begin{align}
 \label{eq:Rmatrix}
R(u)&={\cal P} \check{R}(u)=
(u-\frac{\omega}{2}+1)(u\mathbf{1}-\epsilon\mathcal{P})+u\mathcal{K} \notag\\
&= u(u+\beta) \mathbf{1} -\epsilon (u+\beta) \mathcal{P} +u \mathcal{K} \; ,
\end{align}
 %where $\mathcal{P}$ is the super-permutation matrix.
 which includes the super permutation and is the image of the element \cite{IsMol10}:
 $$
 \rho_i(u) =
  u(u+\beta) \mathbf{1} - (u+\beta) \, s_i + u \,  e_i \;\; \in \;\;
 B_n(\omega) \; .
 $$

The braid version (\ref{eq:YBEbraid}) of the Yang--Baxter equation
has the same form in both the supersymmetric
the non supersymmetric cases.
However
 the standard matrix $R(u)=\mathcal{P}\check{R}(u)$
 (\ref{eq:Rmatrix})
 satisfies the graded version  of the Yang--Baxter equation \cite{KulSkl}
involving extra sign factors,
\begin{equation} \label{eq:YBEgraded}
R_{12}(u-v)(-)^{12} R_{13}(u)(-)^{12} R_{23}(v) = R_{23}(v) (-)^{12} R_{13}(u) (-)^{12} R_{12}(u-v).
\end{equation}
%To prove \eqref{eq:YBEgraded}, we have to realize
Indeed, after substituting  $\check{R}_{ij}(u)=\mathcal{P}_{ij}R_{ij}(u) =
(-)^{ij} P_{ij}R_{ij}(u)$ into
(\ref{eq:YBEbraid}) and moving all standard permutations $P_{ij}$ to the left we
write  (\ref{eq:YBEbraid}) in the form
 \begin{equation} \label{YBEgr01}
R_{23}(u-v)(-)^{13} R_{13}(u)(-)^{12} R_{12}(v) = R_{12}(v) (-)^{13} R_{13}(u) (-)^{23} R_{23}(u-v).
\end{equation}
If $R \in {\rm End}({\cal V}_{(N|M)}^{\; \otimes 2})$ is an even matrix,
then the following condition holds
\begin{equation} \label{Reven1}
R^{i_1i_2}_{j_1j_2} \neq 0 \qquad\; \text{iff}\qquad\;
 [i_1]+[i_2]+[j_1]+[j_2]=0 \; ({\rm mod} 2) \; .
\end{equation}
In particular one can easily check this property for the
 matrices $\mathbf{1},\mathcal{P},\mathcal{K}$
 out of which the operator $R(u)$ is composed.
 %which compose the operator $R(u)$.
Therefore,
\begin{equation} \label{eq:Rmat-}
R_{ij} I_k (-)^{ik}(-)^{jk} = (-)^{ik}(-)^{jk} R_{ij} I_k  \; .
\end{equation}
Using this property and the identities (\ref{siden})
we convert the right hand side of (\ref{YBEgr01}) to
\begin{equation}
 \label{tranRRR}
 \begin{array}{c}
 R_{12}(-)^{13} R_{13} (-)^{23} (-)^{12} (-)^{12}  R_{23}
 = R_{12} (-)^{13} (-)^{23} (-)^{12} R_{13}  (-)^{12}  R_{23} = \\ [0.3cm]
 = (-)^{13} (-)^{23}  R_{12}(-)^{12} R_{13} (-)^{12}  R_{23}  \; ,
 \end{array}
\end{equation}
 and doing the analogous transformations for the left hand side
 we represent (\ref{YBEgr01}) in the form
\begin{equation}
\label{RRR}
\begin{array}{c}
R_{23}(u-v) (-)^{12} R_{13}(u) (-)^{12} R_{12}(v) %= \\ [0.3cm]
 =  R_{12}(v) (-)^{12} R_{13}(u) (-)^{12} R_{23}(u-v) \; .
\end{array}
\end{equation}
After the exchange of the spectral
 parameters $v \to u-v$, $u \to u$ in (\ref{RRR}) we obtain the graded version
 of Yang--Baxter equation \eqref{eq:YBEgraded}.
Finally we stress that $(-)^{12}$  in \eqref{eq:YBEgraded}
 can be exchanged with the sign operator $(-)^{23}$ by means of similar manipulations
 as in (\ref{tranRRR}). Moreover, if $R_{12}(u)$ solves the Yang-Baxter
 equation \eqref{eq:YBEgraded}, then
 the twisted $R$-matrix
 \begin{equation}
 \label{twiR}
\tilde{R}_{12}(u) :=  (-)^{12}R_{12}(u)(-)^{12} \; ,
\end{equation}
is also a solution.

\vspace{0.2cm}

\noindent
{\bf Remark.} Eqs. (\ref{RBrauer01}), (\ref{eq:Rmatrix})
 express uniformly the  $R$-matrices which
are invariant under the action of $SO$, $Sp$ or $OSp$ groups. Recall that
for the $SO$ case the $R$-matrix
  was found in \cite{Zam}, for the $Sp$ case it was constructed
  in \cite{Karow},\cite{Re}. For the $OSp$ case such $R$-matrices were considered in
 many papers (see, e.g., \cite{Kulish}, \cite{Ragoucy}, \cite{Isae}).
Note that the $OSp$ type R-matrix ${\sf R}(u)$ proposed in \cite{Isae}
 coincides with (\ref{eq:Rmatrix})
 in the case $\epsilon=+1$ in the following way
\begin{equation}
 {\sf R}(u)= -R^{\epsilon=+1}(-u).
\end{equation}
The R-matrix \eqref{eq:Rmatrix} is related to $SO$ and
$Sp$ type R-matrix presented in \cite{IsKarKir15} just by
rescaling of the spectral parameter $u\rightarrow -\epsilon u$.
The parameter $\beta=1-\frac{\omega}{2}$ is then rescaled to
$-\epsilon \beta=\frac{\epsilon \omega}{2}-\epsilon$.

\section{The graded RLL-relation and the Yangian ${\cal Y}(osp)$}
\setcounter{equation}{0}

We start with the graded form of the RLL-relation
(see, e.g., \cite{KulSkl}, \cite{Isae})
\begin{equation} \label{eq:RLL}
R_{12}(u-v)L_1(u)(-)^{12} L_2(v) (-)^{12} =
(-)^{12} L_2(v) (-)^{12} L_1(u)  R_{12}(u-v) \; ,
\end{equation}
(the different, but equivalent,
  %\marginpar{\bf \large Is}
 choice of the RLL-relation will be discussed
at the end of this Section, in Remark 3).
The R-matrix is of the form \eqref{eq:Rmatrix} and the elements
 of the super-matrix $||L^a_{\; b}(u)||_{a,b=1}^{N+M}$
involve the  generators of an associative algebra which we denote by ${\cal Y}(osp)$.
We specify this algebra below.
 Note that the sign operators in (\ref{eq:RLL}) are fixed according to
the invariance condition (\ref{Rinvar2}). Consider the product
$L_1(-)^{12} L_2(-)^{23}(-)^{13} L_3$ of three $L$-operators and reorder it
  with the help of (\ref{eq:RLL}) as
 $$
 L_1(u)(-)^{12} L_2(v)(-)^{23}(-)^{13} L_3(w)  \;\; \to \;\;
 L_3(w)(-)^{23} L_2(v)(-)^{13} (-)^{12} L_1(u)  \; ,
 $$
 in two different ways in accordance with the arrangement of brackets
 \be
 \lb{lexy}
 \Bigl( L_1(u)(-)^{12} L_2(v) \Bigr) (-)^{23}(-)^{13} L_3(w) =
 L_1(u)(-)^{12} (-)^{13} \Bigl( L_2(v)(-)^{23} L_3(w)  \Bigr) \; .
 \ee
  As a result we obtain  (using the properties (\ref{siden}),
  (\ref{eq:Rmat-}) of the sign operators) the
 associativity condition for the algebra (\ref{eq:RLL})
 in the form of the graded Yang-Baxter equation (\ref{eq:YBEgraded}).

Expand the L-operator  in the spectral parameter $u$ as
\begin{equation} \label{eq:Lexp}
L^a_{\; b}(u) = \mathbf{1} \delta^a_b+\sum_{k=1}^\infty
 \frac{(L^{(k)})^a_{\; b}}{u^k} \; ,
\end{equation}
where $\mathbf{1}$ denotes the unit element in ${\cal Y}(osp)$.
We multiply the R-matrix by $u^{-2}v^{-2}$
\begin{equation}
\frac{R(u-v)}{u^2v^2}  = \! \left(\frac{1}{v^2}-\frac{2}{uv}+\frac{1}{u^2}+
\frac{\beta}{uv^2}-\frac{\beta}{u^2v}\right)\!\mathbf{1}
-\left( \frac{\epsilon }{uv^2} -\frac{\epsilon }{u^2v}+\frac{\epsilon \beta}{u^2v^2} \right) \!\mathcal{P} +
\left( \frac{1}{uv^2} - \frac{1}{u^2v} \right) \!\mathcal{K},
\end{equation}
and expand the RLL-relation \eqref{eq:RLL} in the spectral parameters $u^{-1}$ and $v^{-1}$.
The coefficient
 at $u^{-k} v^{-j}$ in (\ref{eq:RLL}) gives the defining
 relations for the infinite dimensional associative algebra $\mathcal{Y}(osp)$
 which is called the Yangian of the $osp$-type:
 \begin{equation}
  \label{eq:Yangian}
\begin{array}{c}
  \Bigl\{ L^{(k)}_1 (-)^{12} L^{(j-2)}_2 - 2 L^{(k-1)}_1 (-)^{12} L^{(j-1)}_2 + L^{(k-2)}_1 (-)^{12} L^{(j)}_2 \\
 +\beta L^{(k-1)}_1 (-)^{12} L^{(j-2)}_2 -\beta L^{(k-2)}_1 (-)^{12} L^{(j-1)}_2 \Bigr\} (-)^{12} \\
 -(-)^{12} \Bigl\{ L_2^{(j-2)} (-)^{12} L_1^{(k)} -2 L_2^{(j-1)} (-)^{12} L_1^{(k-1)} + L_2^{(j)} (-)^{12} L_1^{(k-2)} \\
 +\beta L_2^{(j-2)} (-)^{12} L_1^{(k-1)} - \beta L_2^{(j-1)} (-)^{12} L_1^{(k-2)}\Bigr\} \\
-\epsilon \mathcal{P}_{12} \Bigl\{ L_1^{(k-1)}(-)^{12} L_2^{(j-2)} - L_1^{(k-2)}(-)^{12} L_2^{(j-1)} +\beta L_1^{(k-2)}(-)^{12} L_2^{(j-2)} \Bigr\} (-)^{12}  \\
 +\epsilon (-)^{12} \Bigl\{ L_2^{(j-2)}(-)^{12} L_1^{(k-1)} - L_2^{(j-1)}(-)^{12} L_1^{(k-2)} +\beta L_2^{(j-2)}(-)^{12} L_1^{(k-2)} \Bigr\} \mathcal{P}_{12} \\
 +\mathcal{K}_{12}  \Bigl\{ L_1^{(k-1)}(-)^{12} L_2^{(j-2)} - L_1^{(k-2)}(-)^{12} L_2^{(j-1)}  \Bigr\} (-)^{12} \\
 - (-)^{12} \Bigl\{ L_2^{(j-2)}(-)^{12} L_1^{(k-1)} - L_2^{(j-1)}(-)^{12} L_1^{(k-2)} \Bigr\}  \mathcal{K}_{12} = 0 \; .
\end{array}
\end{equation}
The super-commutator is defined as
\begin{equation} \label{SupComm}
[A ,B]_\pm \equiv AB- (-1)^{[A][B]} BA.
\end{equation}
Notice that for elements of two
super-matrices $||A^a_{\; b}||,||B^a_{\; b}||$ the super-commutator has
the following representation
$$
\bigl( A_1 (-)^{12} B_2 (-)^{12}  -  (-)^{12} B_2 (-)^{12} A_1 \bigr)^{a_1a_2}_{c_1c_2} = (-1)^{[c_1]([a_2]+[c_2])} [A^{a_1}_{\ c_1},B^{a_2}_{\ c_2}]_\pm.
$$
This implies that the relation \eqref{eq:Yangian} for the Yangian $\mathcal{Y}(osp)$
has the following coordinate form:
\begin{equation}
\label{supyang}
\begin{array}{c}
(-1)^{[c_1]([a_2]+[c_2])} \Bigl\{ \left[(L^{(k)})^{a_1}_{c_1}, (L^{(j-2)})^{a_2}_{c_2}\right]_\pm -2  \left[(L^{(k-1)})^{a_1}_{c_1}, (L^{(j-1)})^{a_2}_{c_2}\right]_\pm   \\
+  \left[(L^{(k-2)})^{a_1}_{c_1}, (L^{(j)})^{a_2}_{c_2}\right]_\pm + \beta \left[(L^{(k-1)})^{a_1}_{c_1}, (L^{(j-2)})^{a_2}_{c_2}\right]_\pm - \beta  \left[(L^{(k-2)})^{a_1}_{c_1}, (L^{(j-1)})^{a_2}_{c_2}\right]_\pm \Bigr\}  \\
-\epsilon (-1)^{[a_1][a_2]+[c_1][c_2]+[a_1][c_1]} \Bigl\{ (L^{(k-1)})^{a_2}_{c_1} (L^{(j-2)})^{a_1}_{c_2}   - (L^{(j-2)})^{a_2}_{c_1} (L^{(k-1)})^{a_1}_{c_2} \\
- (L^{(k-2)})^{a_2}_{c_1} (L^{(j-1)})^{a_1}_{c_2}+ (L^{(j-1)})^{a_2}_{c_1} (L^{(k-2)})^{a_1}_{c_2}\\
 + \beta (L^{(k-2)})^{a_2}_{c_1} (L^{(j-2)})^{a_1}_{c_2}   - \beta (L^{(j-2)})^{a_2}_{c_1} (L^{(k-2)})^{a_1}_{c_2} \Bigr\}  \\
 + (-1)^{[c_1]([b_2]+[c_2])} \mathcal{K}^{a_1 a_2}_{b_1 b_2}  \Bigl\{ (L^{(k-1)})^{b_1}_{c_1} (L^{(j-2)})^{b_2}_{c_2} - (L^{(k-2)})^{b_1}_{c_1} (L^{(j-1)})^{b_2}_{c_2} \Bigr\}  \\
- (-1)^{[a_1]([a_2]+[b_2])}  \Bigl\{ (L^{(j-2)})^{a_2}_{b_2} (L^{(k-1)})^{a_1}_{b_1} - (L^{(j-1)})^{a_2}_{b_2} (L^{(k-2)})^{a_1}_{b_1}  \Bigr\} \mathcal{K}^{b_1b_2}_{c_1c_2} = 0.
 \end{array}
\end{equation}

{\bf Remark 1.} We have noticed in section {\bf \ref{sec2}} that the relation between the
 $OSp$ type R-matrix (\ref{eq:Rmatrix}) and the R-matrix in \cite{IsKarKir15}
is given just by rescaling $u\rightarrow -\epsilon u$.
Doing the same rescaling for the L-operator we obtain
$L^{(k)}\rightarrow -\epsilon L^{(k)}$ for $k$ odd and
$L^{(k)}\rightarrow L^{(k)}$ for $k$ even, respectively.
After that the Yangian relations in \cite{IsKarKir15} appear just as a special case
of the relations (\ref{supyang}).

\vspace{0.3cm}

Choosing $k=1,j=3$ we obtain the defining relation for the Lie
superalgebra $osp$
\begin{equation} \label{eq:OSpComm}
L_1^{(1)} (-)^{12} L_2^{(1)}(-)^{12}  -  (-)^{12} L_2^{(1)} (-)^{12} L_1^{(1)}  = [\epsilon\mathcal{P}_{12}-\mathcal{K}_{12},(-)^{12} L_2^{(1)} (-)^{12}] \; .
\end{equation}
 The enveloping algebra of $osp$ is thus a subalgebra in the Yangian  ${\cal Y}(osp)$.
Permuting $(1\leftrightarrow 2)$ in \eqref{eq:OSpComm} (we multiply it by super-permutation from both sides) and adding the original and permuted equation
we obtain the consistency condition
\begin{equation}
\label{Ccond}
\mathcal{K}_{12} \Bigl\{L_1^{(1)} + (-)^{12} L_2 ^{(1)} (-)^{12}\Bigr\} = \Bigl\{L_1^{(1)} + (-)^{12} L_2 ^{(1)} (-)^{12} \Bigr\} \mathcal{K}_{12} \; .
\end{equation}
Using this condition and the identities
  \begin{equation}
 \label{PLL}
 {\cal P}_{12} (-)^{12} L_2 ^{(1)} (-)^{12} = L_1^{(1)} {\cal P}_{12} \; , \;\;\;
  (-)^{12} L_2 ^{(1)} (-)^{12} {\cal P}_{12}  = {\cal P}_{12}  L_1^{(1)}  \; ,
 \end{equation}
 one can simplify defining relation (\ref{eq:OSpComm}) of $osp$ as
\begin{equation} \label{osp1}
L_1^{(1)} (-)^{12} L_2^{(1)}(-)^{12}  -  (-)^{12} L_2^{(1)} (-)^{12} L_1^{(1)}  = [\mathcal{K}_{12} - \epsilon\mathcal{P}_{12}, \, L_1^{(1)} ] \; ,
\end{equation}
which has the following component form
\begin{equation} \label{osp2}
\begin{array}{c}
(-1)^{[c_1]([a_2]+[c_2])} \left[(L^{(1)})^{a_1}_{\ c_1}, (L^{(1)})^{a_2}_{\ c_2}\right]_{\pm}  =   \epsilon (-1)^{[c_1][c_2]}\delta^{a_2}_{c_1} (L^{(1)})^{a_1}_{\ c_2} - \\ [0.3cm] -
\epsilon (-1)^{[a_1][a_2]} \delta^{a_1}_{c_2} (L^{(1)})^{a_2}_{\ c_1}
+\epsilon (-1)^{[c_2]} \varepsilon^{a_1a_2} (L^{(1)})_{c_2c_1}
-\epsilon (-1)^{[a_2]} \varepsilon_{c_1c_2} (L^{(1)})^{a_1a_2} .
\end{array}
\end{equation}
Note that  (\ref{osp1}) can be directly obtained
 from (\ref{eq:Yangian}) with the choice $k=3$, $j=1$.

Multiplying both sides of (\ref{Ccond}) by $\mathcal{K}_{12}$ from the left (or
from the right) and using the identities
\begin{equation} \label{eq:Identities}
\mathcal{K}_{12}^2= \omega \mathcal{K}_{12},\quad \mathcal{K}_{12} A_1 \mathcal{K}_{12} = \epsilon \ \mathrm{str}(A)\ \mathcal{K}_{12},\quad  \mathcal{K}_{12} (-)^{12} A_2 (-)^{12} \mathcal{K}_{12} = \epsilon \ \mathrm{str}(A)\ \mathcal{K}_{12},
\end{equation}
we obtain
\begin{equation} \label{Consist}
\mathcal{K}_{12} \Bigl\{L_1^{(1)} + (-)^{12} L_2 ^{(1)} (-)^{12}\Bigr\} = \frac{2\epsilon}{\omega}\ \mathrm{str} (L^{(1)}) \ \mathcal{K}_{12}= \Bigl\{L_1^{(1)} + (-)^{12} L_2 ^{(1)} (-)^{12} \Bigr\} \mathcal{K}_{12} \; ,
\end{equation}
where $\mathrm{str}(A)\equiv (-1)^{[a]}
 A^{a}_{\ a}=\epsilon\varepsilon^{ba} A_{ba}$ is the supertrace of the super-matrix $A$.
  One can check
 that the element $\mathrm{str} (L^{(1)})$  belongs to the center of the Yangian ${\cal Y}(osp)$
 and therefore to the center of
  the Lie super-algebra $osp$.

  The Yangian ${\cal Y}(osp)$ (as well as the
  $g\ell$-type Yangian \cite{Khor} and the $so$- and $sp$ type Yangians
  \cite{IsKarKir15}) possess the set of automorphisms which
 %\marginpar{\bf \large Is}
  are defined by the assignments
  \begin{equation}
\label{auto}
 L^a_{\; b}(u) \;\; \mapsto \;\; f(u) \; L^a_{\; b}(u + b_0) \; ,
\end{equation}
  where $f(u) = 1 + b_1/u + b_2/u + ...$ is a scalar function and $b_k$ are
  parameters (in general $b_k$ are central elements
  in ${\cal Y}(osp)$). The transformations (\ref{auto}) are implied by the
  form of the defining relations (\ref{eq:RLL}).
  %\marginpar{\it RK}
  As shown in \cite{IsKarKir15} one can use the
  automorphisms (\ref{auto}) to fix $L^{(1)}$ such that str$(L^{(1)})=0$.
  In view of this we define
  the traceless generators
  \begin{equation}
\label{DefG}
 G^a_{\; b} \equiv (L^{(1)})^a_{\; b} -
 \frac{\epsilon}{\omega} \; \mathrm{str} (L^{(1)}) \delta^a_{b} \; , \;\;\;\;\;
 \mathrm{str} (G)=0 \; .
\end{equation}
 These generators satisfy (since we have automorphisms (\ref{auto})
 and $\mathrm{str} (L^{(1)})$ is the central element)
  the same commutation relations (\ref{eq:OSpComm})
  which we write as:
 \begin{equation} \label{CommG}
(-)^{12} G_1 (-)^{12} G_2  -  G_2 (-)^{12} G_1 (-)^{12}  =
[\epsilon\mathcal{P}_{12}-\tilde{\mathcal{K}}_{12}, \; G_2 ] \; ,
\end{equation}
where $\tilde{\cal K}_{12}= (-)^{12} {\cal K}_{12} (-)^{12}$
($\tilde{\cal K}^{a_1 a_2}_{\;\; b_1 b_2}=
 \varepsilon^{a_2 a_1}\varepsilon_{b_2 b_1}$) and we have used
 $(-)^{12} {\cal P}_{12} (-)^{12} ={\cal P}_{12}$. This is to be compared
with (\ref{osp06a}), (\ref{osp10}) if the algebra of elements $G^a_{\; b}$ is
represented in the space $\cal{V_{(N|M)}}$. Note that the commutation relations
 %\marginpar{\bf \large Is, it RK}
(\ref{CommG}) transform to the commutation relations (\ref{osp06a})
if we  redefine  the supermetric as $\varepsilon_{a b} \to
\epsilon (-1)^{[a]} \varepsilon_{a b} = \varepsilon_{ba}$ (see also
discussion in Remark 3 below).

Further, for the traceless generators (\ref{DefG}) $G^a_{\; b}$
from (\ref{Consist}) we have (cf. (\ref{osp08}))
\begin{equation}
\label{SymCond}
\mathcal{K}_{12} \Bigl\{G_1+(-)^{12} G_2  (-)^{12}\Bigr\}= 0 = \Bigl\{G_1 + (-)^{12} G_2 (-)^{12} \Bigr\} \mathcal{K}_{12} \; .
\end{equation}
 In components this reads as the condition (\ref{OspAlg})
for the generators of the Lie superalgebra $osp$:
\begin{equation} \label{SymmCond1}
G_{ab} +\epsilon(-1)^{[a][b]+[a]+[b]} G_{ba}=0.
\end{equation}
{\bf Remark 2.} Comparing the $RLL$-relations (\ref{eq:RLL}) and the graded
Yang-Baxter equation (\ref{eq:YBEgraded}) one finds that
the $L$-operator is represented (as an operator
 in ${\cal V}_{(N|M)} \otimes {\cal V}_{(N|M)}$)
 in the form of the twisted solution (\ref{twiR})
 of the Yang-Baxter equation:
 \begin{equation} \label{Lfunda}
 L(u) = \frac{1}{u^2} (-)^{12} R_{12}(u) (-)^{12} =
{\bf 1} + \frac{1}{u} \Bigl({\bf 1} \beta + (\tilde{\cal K} - \epsilon {\cal P}) \Bigr)
 - \frac{\epsilon \beta}{u^2} {\cal P} \; .
\end{equation}
The operators $L_1(u)$ and $L_2(v)$ in (\ref{eq:RLL}) should be understood
 as $\frac{1}{u^2} (-)^{13} R_{13}(u) (-)^{13}$ and
  $\frac{1}{u^2} (-)^{23} R_{23}(u) (-)^{23}$, respectively.
According to the above consideration of the Yangian ${\cal Y}(osp)$ the coefficient
$\Bigl({\bf 1} \beta + (\tilde{\cal K} - \epsilon {\cal P}) \Bigr)$ of $u^{-1}$
in (\ref{Lfunda}) is a representation of the element $L^{(1)}$ in (\ref{eq:Lexp})
 \begin{equation} \label{TL1}
 T^{a_2}_{\; c_2}((L^{(1)})^{a_1}_{\; c_1}) =
  \beta  \; \delta^{a_1}_{\; c_1} \delta^{a_2}_{\; c_2} +
 (\tilde{\cal K}^{a_1 a_2}_{\;\; c_1 c_2} -
 \epsilon {\cal P}^{a_1 a_2}_{\;\; c_1 c_2}) \; .
 \end{equation}
Since the $L$-operator (\ref{Lfunda}) satisfies the $RLL$-elations (\ref{eq:RLL})
the operator (\ref{TL1}) should obey commutation relations (\ref{eq:OSpComm})
 and conditions (\ref{Consist}). Taking into account (\ref{TL1}) we represent
 the traceless part (\ref{DefG}) as
 \begin{equation} \label{TG}
 G^{a_1 a_2}_{\; c_1 c_2} \equiv
  T^{a_2}_{\; c_2}(G^{a_1}_{\; c_1})  =
  (\tilde{\cal K}^{a_1 a_2}_{\;\; c_1 c_2} - \epsilon {\cal P}^{a_1 a_2}_{\;\; c_1 c_2}) =
  (-1)^{[a_1]+ [c_1]} \varepsilon^{a_1 a_2} \varepsilon_{c_1 c_2} -
  \epsilon (-1)^{[a_1][a_2]} \delta^{a_1}_{\; c_2} \delta^{a_2}_{\; c_1} \; .
 \end{equation}
This formula defines the fundamental representation $T$ of the $osp$-generators
 $G^a_{\; b}$ which satisfy
 commutation relations (\ref{CommG}) and conditions (\ref{SymCond}),
 (\ref{SymmCond1}).

 \vspace{0.2cm}

 \noindent
 {\bf Remark 3.}
 We note that the choice of the basis of $osp$ in (\ref{TG})
 differs from the choice of the basis of $osp$ in (\ref{osp03}) by sign factors
 $$
 T^{a_2}_{\; c_2}(G^{a_1}_{\; c_1})  =
 (-1)^{[a_1][a_2] + [c_1][c_2]} (\tilde{G}^{a_1}_{\; c_1})^{a_2}_{\; c_2} \; .
 $$
 So we have
 $$
 \tilde{G}_{12} = G_{21} = (-)^{12} G_{12} (-)^{12}  \; ,
 $$
  where $G_{12}$ and $\tilde{G}_{12}$ are defined in
  (\ref{TG}) and (\ref{osp03}) (compare eqs. (\ref{osp06a}), (\ref{osp08})
  with (\ref{eq:OSpComm}), (\ref{SymCond})). However
  one can start from the different
  form of the graded RLL-relation (cf. (\ref{eq:RLL}))
\begin{equation} \label{RLL2}
R_{12}(u-v)(-)^{12} \tilde{L}_1(u)(-)^{12} \tilde{L}_2(v)  =
 \tilde{L}_2(v) (-)^{12} \tilde{L}_1(u) (-)^{12} R_{12}(u-v) \; ,
\end{equation}
which yields the equivalent definition of the Yangian ${\cal Y}(osp)$.
In this case the $R$-matrix (fundamental) representation of the Yangian (\ref{RLL2})
 is given by  the formula (cf. (\ref{Lfunda}))
\begin{equation} \label{Lfun2}
 \tilde{L}(u) = \frac{1}{u^2} \, R_{12}(u)  =
{\bf 1} + \frac{1}{u} \Bigl({\bf 1} \beta + ({\cal K} - \epsilon {\cal P}) \Bigr)
 - \frac{\epsilon \beta}{u^2} {\cal P} \; ,
\end{equation}
which leads to the following fundamental representations of the
Yangian generators $\tilde{L}^{(1)}$
and their traceless $osp$ generators $\tilde{G}$
 (cf. (\ref{TL1}), (\ref{TG})):
  \begin{equation} \label{TL2}
 T^{a_2}_{\; c_2}((\tilde{L}^{(1)})^{a_1}_{\; c_1}) =
  \beta  \; \delta^{a_1}_{\; c_1} \delta^{a_2}_{\; c_2} +
 ({\cal K}^{a_1 a_2}_{\;\; c_1 c_2} -
 \epsilon {\cal P}^{a_1 a_2}_{\;\; c_1 c_2}) \; ,
 \end{equation}
 \begin{equation} \label{TG2}
 \tilde{G}^{a_1 a_2}_{\; c_1 c_2} \equiv
  T^{a_2}_{\; c_2}(\tilde{G}^{a_1}_{\; c_1})  =
  ({\cal K}^{a_1 a_2}_{\;\; c_1 c_2} - \epsilon {\cal P}^{a_1 a_2}_{\;\; c_1 c_2}) =
  \varepsilon^{a_1 a_2} \varepsilon_{c_1 c_2} -
  \epsilon (-1)^{[a_1][a_2]} \delta^{a_1}_{\; c_2} \delta^{a_2}_{\; c_1} \; .
 \end{equation}
 We see that the representation of the basis elements
 of $osp$ in (\ref{TG2})
 coincides with the basis of $osp$ proposed in (\ref{osp03}) and, thus,
 the consideration of Subsection 2.1 is relevant to the definition of the Yangian
 ${\cal Y}(osp)$ given in (\ref{RLL2}). The next point which we
 would like to stress here is that $RLL$ relations (\ref{eq:RLL}) can be
 rewritten in the form of (\ref{RLL2}):
 \begin{equation} \label{RLL3}
\tilde{R}_{12}(u-v) (-)^{12} L_1(u)(-)^{12} L_2(v)  =
 L_2(v) (-)^{12} L_1(u) (-)^{12} \tilde{R}_{12}(u-v) \; ,
\end{equation}
where $\tilde{R}_{12}(u)$ is the twisted solution
(\ref{twiR}) of the Yang-Baxter equation
(\ref{eq:YBEgraded}). The twisted matrix $\tilde{R}_{12}(u)$
can be obtained from $R$-matrix (\ref{eq:Rmatrix})
(compare (\ref{Lfunda}) and (\ref{Lfun2})) by the substitution
$\varepsilon_{ab} \to \epsilon (-1)^a\varepsilon_{ab}= \varepsilon_{ba}$. It means that
all formulas which we obtain below for the Yangian (\ref{eq:RLL})
can be easily transformed to the formulas for the Yangian (\ref{RLL2})
by the simple transformation of the supermetric $\varepsilon_{ab} \to \varepsilon_{ba}$.

\section{The linear evaluation of the Yangian $\mathcal{Y}(osp)$}
\label{sec:LinEval}
\setcounter{equation}{0}

\subsection{The conditions for linear evaluation }

Let us suppose that the L-operator expansion (\ref{eq:Lexp})
 terminates after the first term, i.e.,
\begin{equation}
 \label{L01}
L(u) = u \mathbf{1} +  L^{(1)}.
\end{equation}
Writing the RLL-relation \eqref{eq:RLL} with this form of the L-operator
and expanding it in $u$ and $v$ we obtain a set of conditions imposed on $L^{(1)}$.

All terms in the RLL-relation \eqref{eq:RLL}
 proportional to $u^k v^\ell$ for $(k+\ell) \geq 3$
 give trivial conditions which are automatically satisfied.
The coefficients at $u^2$, $v^2$ and $uv$ give the defining relations (\ref{osp1})
for generators $L^{(1)} \in osp$ and the condition (\ref{Ccond}).
 The condition appearing at first powers of $u$ and $v$ is
\begin{equation}\label{eq:RepSpec}
\mathcal{K}_{12} \left(  L_1^{(1)} (-)^{12} L_2^{(1)} (-)^{12} + \beta L_1 ^{(1)}  \right) = \left( (-)^{12} L_2^{(1)} (-)^{12} L_1^{(1)}   + \beta L_1^{(1)}  \right) \mathcal{K}_{12} \; .
\end{equation}
 Multiplying it by ${\cal P}_{12}$ from both sides
 and using (\ref{PLL}), (\ref{Ccond}) one represents it in the equivalent form
 \begin{equation}\label{new01}
\mathcal{K}_{12} \left( (-)^{12} L_2^{(1)} (-)^{12} L_1^{(1)}  - \beta L_1 ^{(1)}  \right) = \left( L_1^{(1)} (-)^{12} L_2^{(1)} (-)^{12}  - \beta L_1^{(1)}  \right) \mathcal{K}_{12} \; .
\end{equation}
These two equations can be obtained directly from \eqref{eq:Yangian} taking $j=2,k=3$
or $j=3,k=2$. Finally the condition at zero power of $u$ and $v$ is trivial
in view of identities (\ref{PLL}).

 Thus we come to the following statement.
 
 \noindent
\begin{proposition}\label{prop22} The $L$-operator
$$
L^a_{\;\; b}(u) = (u + \alpha) \, {\bf 1} \delta^a_b
 + (L^{(1)})^a_{\;\; b} \; ,
$$
where $\alpha$ is an arbitrary constant,
 solves the RLL-relation (\ref{eq:RLL}) iff
 the elements $(L^{(1)})^a_{\;\; b}$
 generate the  $osp$-algebra with the
 defining relations (\ref{osp1}) and satisfy
  the conditions (\ref{Ccond}), (\ref{new01}).
\end{proposition}

In the case of linear evaluation of $L$-operator (\ref{L01}),
in addition to the defining relations (\ref{osp1}) and the condition (\ref{Ccond})
we obtain only one  non-trivial constraint \eqref{new01}
on the generators $L^{(1)}$ of $osp$. Using \eqref{Consist} we write
 \eqref{new01} as the quadratic relation in $L^{(1)}$:
\begin{equation}
\left[ \mathcal{K}_{12}, \left(L_1^{(1)}\right)^2  + \beta' \, L_1^{(1)}  \right]=0
\end{equation}
where $\beta'=\beta - \frac{2\epsilon}{\omega} \mathrm{str}(L^{(1)})$.
Multiplying it by $\mathcal{K}_{12}$ from one side and using \eqref{eq:Identities} we arrive at the quadratic characteristic equation imposed on $L^{(1)}$
\begin{equation} \label{eq:LinEvalCond}
(L^{(1)})^2 + \beta' L^{(1)} -\frac{\epsilon}{\omega} \left\{ \mathrm{str}\left((L^{(1)})^2\right) + \beta' \mathrm{str}\left(L^{(1)}\right)\right\} \mathbf{1} =0.
\end{equation}
 %If we redefine $L^{(1)}$ to have vanishing supertrace, the condition simplifies
For the generators $G^a_{\; b}$  defined in (\ref{DefG})
 with vanishing supertrace, $\mathrm{str}(G)=0$,
this condition  simplifies to
\begin{equation} \label{Cond1}
G^2 + \beta \, G -\frac{\epsilon}{\omega} \,
 \mathrm{str}\left(G^2\right) \mathbf{1} =0 \; ,
\end{equation}
 where  $\beta = 1-\omega/2$.
We arrive  at the following statement.
\noindent
\begin{proposition}\label{prop2a} The linear evaluation (\ref{L01})
of the $L$-operator
$$
L^a_{\;\; b}(u) = (u + \alpha) \, {\bf 1} \delta^a_b
 + G^a_{\;\; b} \; ,
$$
where $\alpha$ is an arbitrary constant,
 solves the RLL-relation (\ref{eq:RLL}) if $G^a_{\;\; b}$
 is a traceless matrix of generators of $osp$, which satisfy
 eqs. (\ref{CommG}), (\ref{SymCond}),
  and in addition obeys the quadratic characteristic identity
  (\ref{Cond1}).
\end{proposition}

\subsection{The super-spinor representation }

In this subsection we intend to construct an explicit representation
of $\mathcal{Y}(osp)$
where the generators of $osp \subset \mathcal{Y}(osp)$
satisfy the quadratic characteristic equation
\eqref{Cond1} required for the linear evaluation (\ref{L01}).

We look for a generalization of the metaplectic or spinor representations
of the $Sp(n)$ or $SO(n)$ groups which can be formulated
 according to \cite{IsKarKir15} based on algebras of bosonic and fermionic
 oscillators with the defining relation
 invariant under the group action.
 We  introduce
 the algebra ${\cal A}$
 of super-oscillators involving both bosonic and fermionic
 oscillators.

 Consider the super-oscillators $c^a$ $(a=1,2,\dots,N+M)$ as generators of
an associative algebra ${\cal A}$ with the defining relation
\begin{equation}
 \label{suposc}
[c^a,c^b]_\epsilon \equiv c^a c^b +\epsilon (-1)^{[a][b]} c^b c^a =\varepsilon^{ab},
\end{equation}
where $\varepsilon^{ab}$ is the super-metric defined in \eqref{SuperMetric1} and
\eqref{SuperMet2}. The super-oscillators  $c^a$ with $[a]=0\; ({\rm mod}2)$
are bosonic and with $[a]=1 \; ({\rm mod}2)$ are fermionic.
It is important that the defining relations (\ref{suposc})
are invariant under the action $c^a\rightarrow c'^a=U^a_{\ c} c^c$
of the super-group $OSp$ as can be easily shown:
\begin{eqnarray}
& [c'^a,c'^b]_\epsilon = [U^a_{\ c} c^c,  U^b_{\ d} c^d]_\epsilon = U^a_{\ c} c^c U^b_{\ d} c^d + \epsilon(-1)^{[a][b]} U^b_{\ d} c^d U^a_{\ c} c^c & \notag \\
& = (-1)^{[c][b]+[c][d]} U^a_{\ c} U^b_{\ d} \left( c^c c^d +\epsilon (-1)^{[c][d]}  c^d c^c \right) = (-1)^{[c][b]+[c][d]} U^a_{\ c} U^b_{\ d} \varepsilon^{cd} =
\varepsilon^{ab}  , &
\end{eqnarray}
where we have used the condition (\ref{eq:OSpDef}) for the elements $U \in Osp$.

With the help of the convention (\ref{agree})
for lowering indices  one can write the relations (\ref{suposc})
in the equivalent forms
\begin{equation}
[c_a,c_b]_\epsilon \equiv c_a c_b +\epsilon (-1)^{[a][b]} c_b c_a =\varepsilon_{ba}
\;\; \Leftrightarrow  \;\;
c_a c^b + \epsilon(-1)^{[a][b]} c^b c_a =\delta^b_a \; .
\end{equation}
 The super-oscillators satisfy the following contraction identity:
\begin{align} \label{Contract}
c^ac_a = \epsilon(-1)^{[a]} c_a c^a =
 %\frac{1}{2} (c^ac_a + \epsilon(-1)^{[a]} c_a c^a) =
 \frac{1}{2} \varepsilon^{ab} (c_b c_a + \epsilon(-1)^{[a]} c_a c_b) = \frac{1}{2} \varepsilon^{ab} \varepsilon_{ab} =\frac{\omega}{2}.
\end{align}
Further we need the super-symmetrised product of two super-oscillators:
\begin{equation}
\label{symcacb}
c^{(a}c^{b)} := \frac{1}{2} \bigl( c^a c^b -\epsilon (-1)^{[a][b]} c^b c^a \bigr)
 \; \in \; {\cal A} \; ,
\end{equation}
and define the operators
\begin{equation}
\label{defF}
F^{ab} \equiv \epsilon (-1)^{b} c^{(a}c^{b)}
 %= \epsilon (-1)^{b} \frac{1}{2}
 %\left( c^a c^b -\epsilon (-1)^{[a][b]} c^b c^a \right)
 \; , \;\;\;\;\  F^{a}_{\ b}=\varepsilon_{bc}F^{ac} =
 \epsilon (-1)^{b} c^a c_b - \frac{1}{2} \delta^a_b \; .
\end{equation}

 \begin{proposition}\label{propF}
The operators $F^{ab} \in {\cal A}$
are traceless, $\str(F)=(-1)^{[a]}F^{a}_{\ a}=0$, and possess the symmetry property
(cf. (\ref{SymmCond1}))
\begin{equation} \label{t1}
F^{ab} = -\epsilon (-1)^{[a][b]+[a]+[b]} F^{ba}.
\end{equation}
They satisfy the supercommutation relations (\ref{osp2}) for generators of $osp$
\begin{align} \label{t2}
[F^{a_1}_{\;\;\, b_1},F^{a_2}_{\;\; b_2}]_\pm  =&  -\epsilon
 (-1)^{[a_1][a_2]+[b_1]([b_2]+[a_2])} \delta^{a_1}_{b_2} F^{a_2}_{\ b_1} - \epsilon (-1)^{[a_2]+[b_1]([b_2]+[a_2])} \varepsilon_{b_1 b_2} F^{a_1a_2} \notag \\
& +\epsilon (-1)^{[b_2]+[b_1]([b_2]+[a_2])} \varepsilon^{a_1a_2} F_{b_2 b_1} + \epsilon (-1)^{[b_1]} \delta^{a_2}_{b_1} F^{a_1}_{\;\;\; b_2} \; ,
\end{align}
and obey the quadratic characteristic identity
(\ref{Cond1}):
\begin{equation} \label{t3}
F^a_{\ b}F^{b}_{\ c}+\beta F^a_{\ c} -\frac{\epsilon}{\omega} \str(F^2) \delta^a_c =0,
\end{equation}
where $\beta = 1 - \omega/2$.
\end{proposition}
\noindent
{\bf Proof.}
The property \eqref{t1} follows from the definition
(\ref{defF}) of $F^{ab}$. The traceless property follows from
the identity \eqref{Contract}. To prove (\ref{t2})
we need the following relation:
\begin{align}
[c^ac^b,c^ec^f]_\pm = & -\epsilon (-1)^{[b][e]+[a][e]+[b][f]} \varepsilon^{af} c^e c^b \notag \\
&+  (-1)^{[b][e]+[a][e]} \varepsilon^{bf} c^ec^a -\epsilon (-1)^{[b][e]} \varepsilon^{ae} c^b c^f + \varepsilon^{be} c^a c^f \; ,
\end{align}
which implies for the supercommutator of two supersymmetrized
 quadratic operators $c^{(a}c^{b)}$:
\begin{eqnarray}
& [c^{(a_1} c^{c_1)},c^{(a_2} c^{c_2)}]_\pm = (-1)^{[a_2][c_2]+[c_1][c_2]} \varepsilon^{a_1c_2} c^{(c_1} c^{a_2)} + (-1)^{[a_2][c_1]+[a_1][a_2]} \varepsilon^{c_1c_2} c^{(a_2} c^{a_1)} & \notag \\
& -\epsilon (-1)^{[c_1][a_2]} \varepsilon^{a_1a_2} c^{(c_1} c^{c_2)} +  \varepsilon^{c_1a_2} c^{(a_1} c^{c_2)}.
\end{eqnarray}
Then using the definition
(\ref{defF}) of $F^{ab}$ we obtain
\begin{align}
 [F^{a_1}_{\ b_1},F^{a_2}_{\ b_2}]_\pm = & \epsilon^2 (-1)^{[b_1]+[b_2]} \varepsilon_{b_1c_1} \varepsilon_{b_2 c_2} [c^{(a_1} c^{c_1)},c^{(a_2} c^{c_2)}]_\pm \notag\\
  =& (-1)^{[a_2][b_2]+[b_1][b_2]+[b_1]+[a_2]} \delta^{a_1}_{b_2} F^{\;\; a_2}_{b_1} + \epsilon (-1)^{([b_1]+[a_1])[a_2]+[a_1]}  \varepsilon_{b_2b_1} F^{a_2a_1}  \notag\\
& - (-1)^{[b_1][a_2]+[b_1]} \varepsilon^{a_1a_2} F_{b_1b_2} +\epsilon (-1)^{[b_1]} \delta^{a_2}_{b_1} F^{a_1}_{\;\;\; b_2} \; .
\end{align}
Applying the properties of the supermetric $\varepsilon^{ab}$
and using the symmetry \eqref{t1}
one can show that this relation is equivalent to \eqref{t2}.
From the contraction identity \eqref{Contract} we obtain that
\begin{eqnarray}
&& F^a_{\ b} F^b_{\ d} = (\frac{\omega}{2}-1)F^a_{\ d} +\frac{1}{4} \delta^a_d = -\beta F^a_{\ d} +\frac{1}{4} \delta^a_d, \\
&& \mathrm{str} (F^2) = (-1)^a F^a_{\ b} F^b_{\ a} = \frac{1}{4}  (-1)^a \delta^a_a =  \frac{\epsilon\omega}{4},
\end{eqnarray}
which proves \eqref{t3}. \hfill \qed

\vspace{0.2cm}

 %We see in the light of this theorem that
 Thus the elements $F^{ab} \in {\cal A}$ form a set of traceless generators
 of  $osp$. Indeed, the elements
 $F^{ab}$ satisfy the supercommutation relation \eqref{osp2} and the symmetry condition \eqref{SymmCond1}. Moreover they satisfy the quadratic characteristic
 identity \eqref{Cond1} for the linear evaluation representation (\ref{L01}).
  It means (see Proposition {\bf \ref{prop2a}}) that the $L$-operator
 which solves RLL-equation (\ref{eq:RLL}) has the form
  \begin{equation}
\label{solF2}
  L^a_{\; b}(u + \alpha) = (u + \alpha) \delta^a_b +  F^a_{\;\; b} \; ,
  \end{equation}
  where $F^a_{\;\; b}$ is defined in (\ref{defF})
  and $\alpha$ is an arbitrary constant. Note that the
  appearance of the parameter $\alpha$
  in the solution (\ref{solF2}) is explained by the
  invariance of the $RLL$ equations (\ref{eq:RLL}) under the shift
  of the spectral parameters $u \to u+\alpha$, $v \to v+\alpha$.

  \vspace{0.3cm}

\noindent
{\bf Remark.} At the end of this subsection we note that for every super-matrix $||A_{ba}||$
we have
\begin{equation}
\label{AFc}
\left[\frac{1}{2}A_{ba}F^{ab},c^d \right] = \frac{\epsilon}{2} (-1)^{[b]} A_{ba} \left[c^{(a}c^{b)},c^d\right]_\pm = A^d_{\ b} c^b,
\end{equation}
where we applied the supercommutation relations between the symmetrized
quadratic product $c^{(a}c^{b)}$ and the super-oscillator $c^d$:
\begin{equation}
[c^{(a}c^{b)},c^d]_\pm = -\epsilon (-1)^{[b][d]} \varepsilon^{ad} c^b + \varepsilon^{bd} c^a +(-1)^{[a][b]+[a][d]} \varepsilon^{bd} c^a -\epsilon (-1)^{[a][b]} \varepsilon^{ad}c^b.
\end{equation}
Equation (\ref{AFc}) demonstrates that the operators
$F^{ab}$ generate any linear transformation of generators $c^a \in {\cal A}$
 under the adjoint action.
Let us consider the graded tensor product ${\cal A} \otimes {\cal A}$ of two
algebras of the super-oscillators. It is useful to denote the
generators of ${\cal A} \otimes {\cal A}$ as $c^a \otimes e =
c^a_1$ and $e \otimes c^a = c^a_2$ where $e$ is
the unit element of $\mathcal{A}$. Then formula (\ref{AFc}) for
the adjoint action  is generalized to the case of ${\cal A} \otimes {\cal A}$  as
follows.
\begin{equation}
\label{AFc2}
\left[\frac{1}{2}A_{ba}(F_1^{ab} + F_2^{ab}), \; c_1^{1 \dots r \rangle}
c_2^{r+1, \dots k \rangle}  \right]  = \sum_{i=1}^k \;
 A_{\{1\dots i\}} \;
c_1^{1 \dots r \rangle}  c_2^{r+1, \dots k \rangle}   \; ,
\end{equation}
where $c_\ell^{1 \dots r \rangle} = c^{a_1}_\ell \cdots c^{a_r}_\ell$ for $\ell=1,2$,
$F_\ell^{ab} = \epsilon (-1)^{[b]} c^{(a}_\ell c^{b)}_\ell$
and the dressed supermatrices $A_{\{1\dots k\}} $ have been defined in (\ref{signop}).
Comparing this formula with (\ref{TATA}) we find that the invariance condition
for any function $f(c^a_1,c^b_2) \in  {\cal A} \otimes {\cal A}$ is written in
the form
 $$\left[\frac{1}{2}A_{ba}(F_1^{ab} + F_2^{ab}), \; f(c^a_1,c^b_2) \right]  = 0.$$
 If ${\rm grad}(f)=0$, this invariance condition is
 equivalent to
 \begin{equation}
\label{AFc3}
\left[ (F_1^{ab} + F_2^{ab}), \; f(c^a_1,c^b_2) \right]  = 0 \; .
\end{equation}
We shall use this condition in Section {\bf \ref{RFRF}}.

\section{The super-spinorial R-operator  \label{RFRF}}
\setcounter{equation}{0}

We shall construct the $R$ operator intertwining in the $RLL$ relation two
super-spinor representations formulated in terms of super-oscillators.
We  follow here the approach  developed for the $so$-case in \cite{CDI},\cite{CDI2}
and then extended for the $sp$-case in \cite{IsKarKir15}.
We define the L-operator as
\begin{eqnarray}
 \label{L02}
L(u) \equiv u\mathbf{1} - \frac{1}{2} F^{ab} \otimes G_{ba} \;\; \in \;\;
{\cal A} \otimes {\cal Y}(osp) \; ,
\end{eqnarray}
where $G_{ba}$ are  generators of $osp$ and
 $F^{ab}=(-1)^b \epsilon c^{(a}c^{b)} \equiv \tau(G^{ab}) \in {\cal A}$ are
  elements $G^{ab}$ in the
 super-spinor representation $\tau$ (see proposition {\bf \ref{propF}}).
 We shall construct the  R-operator
 $\check{\mathcal{R}}_{12}(u) \in {\cal A} \otimes {\cal A}$ intertwining
the $L$-operators (\ref{L02}) via the following RLL-relation
\begin{equation} \label{eq:RLL1}
\check{\mathcal{R}}_{12}(u) L_1(u+v) L_2(v) =
L_1(v) L_2(u+v) \check{\mathcal{R}}_{12}(u) \qquad \in {\cal A} \otimes {\cal A} \otimes {\cal Y}(osp).
\end{equation}
The operator $\check{\mathcal{R}}_{12}(u)$ acts trivially on the factor ${\cal Y}(osp)$, whereas
 $$
L_1(u) \equiv u\cdot e\otimes e\otimes\mathbf{1} -  \frac{1}{2} F^{ab}_1 \otimes G_{ba}
 , \;\;\;
L_2(v) \equiv v \cdot e\otimes e\otimes\mathbf{1} -  \frac{1}{2}  F^{ab}_2 \otimes G_{ba} .
 $$
 $e$ is the unit element of ${\cal A}$
 and as before we denote $F^{ab}_1 = F^{ab} \otimes e$,
 $F^{ab}_2 = e \otimes F^{ab}$.
 We  consider the case when
 ${\rm grad}(\check{\mathcal{R}}_{12}(u)) = 0$.

\subsection{The defining conditions}
\label{sec:defcondR}

The conditions restricting the R-operator are obtained
from the expansion of the RLL-relation \eqref{eq:RLL1} in the parameter $v$.
The condition at $v^2$ is trivial. At $v^1$ we
 obtain the invariancy condition  (\ref{AFc3}) w.r.t. the adjoint action of $osp$
\begin{equation} \label{eq:Rcond1}
\left[ \check{\mathcal{R}}(u), F_1^{ab} +F_2^{ab} \right] = 0.
 %\left[ \hat{\mathcal{R}}(u), F_1^{ab} +F_2^{ab} \right]_\pm = 0.
\end{equation}
The condition appearing at $v^0$ is
\begin{equation} \label{eq:Rcond2}
\begin{array}{c}
u\left[ \check{\mathcal{R}}(u)F_2^{ab} - F_1^{ab} \check{\mathcal{R}}_{12}(u) \right]
\otimes G_{ba} - \\ [0.3cm]
 -\frac{1}{2} (-1)^{([b]+[c])([d]+[a])}
\left[ \check{\mathcal{R}}(u)F_1^{cb} F_2^{ad} - F_1^{cb} F_2^{ad} \check{\mathcal{R}}_{12}(u) \right] \otimes G_{bc} G_{da} = 0.
\end{array}
\end{equation}
The product of two generators can be rewritten via the supercommutator \eqref{SupComm} and the superanticommutator
\begin{equation}
G_{bc} G_{da} = \frac{1}{2}\left\{ [G_{bc}, G_{da}]_\pm + \{ G_{bc}, G_{da}\}_\mp \right\}.
\end{equation}
The superanticommutator is defined as
\begin{equation}
\{ A,B \}_\mp \equiv A \, B+(-1)^{[A][B]} \, B \, A  \; .
\end{equation}
We introduce the following notation
\begin{equation} \label{Xcbad}
X^{(cb)(ad)} \equiv (-1)^{([b]+[c])([d]+[a])}
\left( \check{\mathcal{R}}(u)F_1^{cb} F_2^{ad} - F_1^{cb} F_2^{ad} \check{\mathcal{R}}_{12}(u)\right)
\end{equation}
and use the supercommutation relations for $osp$ \eqref{eq:OSpComm} to write \eqref{eq:Rcond2} as
\begin{equation}
\label{eq:Rcond2a}
\left\{ u\left[ \check{\mathcal{R}}(u)F_2^{ab} - F_1^{ab} \check{\mathcal{R}}_{12}(u) \right]  -  \varepsilon_{cd} X^{(cb)(ad)} \right\} \otimes G_{ba} = \frac{1}{4}  X^{(cb)(ad)} \otimes
 \{G_{bc},G_{da}\}_\mp.
\end{equation}
This condition is fulfilled only if both sides vanish separately.
This becomes the  key point of the construction of the R-operator
$\hat{ \mathcal{R}}(u)$.
\begin{equation}
\label{eq:Rcond2al}
\left\{ u\left[ \check{\mathcal{R}}(u)F_2^{ab} - F_1^{ab} \check{\mathcal{R}}_{12}(u)
\right]  -  \varepsilon_{cd} X^{(cb)(ad)} \right\} \otimes G_{ba} = 0,
\end{equation}
\begin{equation}
\label{eq:Rcond2ar}
  X^{(cb)(ad)} \otimes \{G_{bc},G_{da}\}_\mp = 0.
\end{equation}

\subsection{Auxiliary variables}
\label{sec:auxvar}

We have to deal with the supersymmetrization of the product of
super-oscillators generalizing (\ref{symcacb}),
\begin{equation} \label{GradSym}
c^{(a_1} c^{a_2} \dots c^{a_k)} \equiv \frac{1}{k!}\sum_{\sigma\in S_k}
(-\epsilon)^{p(\sigma)}
(-1)^{\hat{\sigma}} c^{a_{\sigma(1)}} \cdots c^{a_{\sigma(k)}}
\end{equation}
where $p(\sigma)$ denotes the parity of the permutation $\sigma$. Let us
explain what we mean by ${\hat{\sigma}}$.
We denote the basic transposition as $\sigma_j\equiv \sigma_{j, j+1}$
permuting the $j$-th and $(j+1)$-st site.
For the basic transposition we define $\hat{\sigma_j}=[a_j][a_{j+1}]$.
For a general permutation $\sigma$ with a given decomposition into the basic
transpositions
$\sigma=\sigma_{j_1}\sigma_{j_2}\dots \sigma_{j_{k-1}}\sigma_{j_k}$, we
define
\begin{equation}
\hat{\sigma}=[a_{j_k}][a_{j_k+1}]+[a_{\sigma_{j_k}(j_{k-1})}]
[a_{\sigma_{j_k}(j_{k-1}+1)}]+
\cdots+[a_{\sigma_{j_2}\cdots\sigma_{j_k}(j_1)}]
[a_{\sigma_{j_2}\cdots\sigma_{j_k}(j_1+1)}].
\end{equation}
Thus, the factor $(-1)^{\hat{\sigma}}$ in (\ref{GradSym}) is needed
 %\marginpar{\bf \large Is}
to take into account the graded properties of the super-oscillators $c^a$
(the example of (\ref{GradSym}) for $k=2$ is given in (\ref{symcacb})).

It is useful to introduce a set of auxiliary variables $\kappa,\kappa'$
with the following properties
\be \label{kappa}
\kappa_a=\varepsilon_{ab} \kappa^b, \ee
$$
\kappa^a\kappa^b=-\epsilon (-1)^{[a][b]} \kappa^b\kappa^a, \ \
\kappa^a\kappa'^b=-\epsilon (-1)^{[a][b]} \kappa'^b\kappa^a, $$ $$
\kappa^a c^b=-\epsilon (-1)^{[a][b]} c^b\kappa^a, \ \
\kappa'^a c^b=-\epsilon (-1)^{[a][b]} c^b\kappa'^a. $$

with the derivatives $\partial^a\equiv \frac{\partial}{\partial \kappa_a}$
satisfying (cf. (\ref{suposc}))
$$
  [\partial^a,\kappa^b]_\epsilon =  \partial^a \kappa^b +
 \epsilon (-1)^{[a][b]} \kappa^b \partial^a = \varepsilon^{ba}
 \; ,
 $$
 \begin{equation}
\label{Kder}
\partial^a \partial^b +
 \epsilon (-1)^{[a][b]} \partial^b \partial^a = 0 \; , \;\;\; \partial^a c^b +
 \epsilon (-1)^{[a][b]} c^b \partial^a = 0 .
\end{equation}

The scalar product is defined by the supermetric $(\kappa\cdot\kappa')\equiv \varepsilon_{ba} \kappa^a \kappa^{\prime b}=\kappa_b \kappa^{\prime b}$. This product is skew-symmetric $(\kappa\cdot\kappa')=-(\kappa'\cdot\kappa)$.  It is easy to show  that
\begin{equation}
[\partial^a,(\kappa\cdot c)]=c^a
\end{equation}
and using this property we deduce that
\begin{equation}
\frac{1}{k!}\partial^{a_1}\cdots \partial^{a_k} (\kappa\cdot c)^k =  c^{(a_1}\cdots c^{a_k)}
\end{equation}
or equivalently
\begin{equation}
\label{best}
\partial^{a_1}\cdots \partial^{a_k} e^{(\kappa\cdot c)}\Big|_{\kappa=0} = c^{(a_1}\cdots c^{a_k)}.
\end{equation}
 Thus we represent the
 supersymmetrized product $c^{(a_1}\cdots c^{a_k)}$  of super-oscillators
(\ref{GradSym})
 with nontrivial commutation relations (\ref{suposc}) by the ordinary product
 $\partial^{a_1}\cdots \partial^{a_k}$ of $\kappa$-derivatives which obey
 homogeneous commutation relations (\ref{Kder}).
The derivative $\partial'^a$ w.r.t. the variable $\kappa'_a$
commutes with the product $(\kappa\cdot c)$. We can also show that
\begin{equation}
[(\kappa\cdot c),(\kappa'\cdot c)] =
- (\kappa\cdot \kappa')_1 = (\kappa'\cdot \kappa)_1
\end{equation}
where we introduce the second type of a scalar product
\begin{equation}
(\kappa\cdot \kappa')_1 \equiv \kappa^a \kappa'_a =
\varepsilon_{ab} \kappa^a \kappa'^b = \epsilon (-1)^{[b]} \kappa_b \kappa'^b.
\end{equation}

Using the Baker-Campbell-Hausdorff formula we calculate the product of two
symmetrized factors
$$
\begin{array}{l}
c^{(a_1}\cdots c^{a_k)} c^{(a}c^{b)} = \partial^{a_1}\cdots \partial^{a_k} e^{(\kappa\cdot c)}\ \partial'^a \partial'^b e^{(\kappa'\cdot c)} \Big|_{\kappa,\kappa'=0} =  \\ [0.2cm]
 = \partial^{a_1}\cdots \partial^{a_k} \ \partial'^a \partial'^b e^{(\kappa\cdot c)} e^{(\kappa'\cdot c)} \Big|_{\kappa,\kappa'=0}
=\partial^{a_1}\cdots \partial^{a_k} \ \partial'^a \partial'^b e^{((\kappa+\kappa')\cdot c)+\frac{1}{2}(\kappa'\cdot \kappa)_1} \Big|_{\kappa,\kappa'=0} =
\end{array}
$$
We continue by the
 change of the variables $\kappa,\kappa'\rightarrow \bar{\kappa}=
\kappa+\kappa',\kappa'$, and then omit the bar over $\bar{\kappa}$,
\begin{align}
&= \partial^{a_1}\cdots \partial^{a_k} \ (\partial^a+\partial'^a)
(\partial^b+\partial'^b)  e^{(\kappa\cdot c)+\frac{1}{2}(\kappa'\cdot \kappa)_1} \Big|_{\kappa,\kappa'=0}= & \notag\\
\intertext{and use $[(\kappa\cdot c),(\kappa'\cdot \kappa)_1]=0$ and
$[\partial'^a,(\kappa'\cdot \kappa)_1]=\epsilon (-1)^a \kappa^a$ to obtain}
 &=\partial^{a_1}\cdots \partial^{a_k} \  \Bigl[ \partial^a \partial^b +\frac{1}{2} \partial^a
 \epsilon (-1)^{[b]} \kappa^b -\frac{1}{2}(-1)^{[a][b]+[a]}\partial^b
 \kappa^a + \frac{1}{4}(-1)^{[a]+[b]}\kappa^a \kappa^b\Bigr] e^{(\kappa \cdot c)}
 \Big|_{\kappa=0} & \notag\\
&=\partial^{a_1}\cdots \partial^{a_k} \
[+]^{ab} \ e^{(\kappa \cdot c)} \Big|_{\kappa=0},&
 \label{ccF}
\end{align}
Here we introduce the  concise notation
\begin{equation}
\label{plus}
[+]^{ab}=\partial^a \partial^b + \frac{1}{2} \left( \epsilon (-1)^{[a]} \kappa^a \partial^b - (-1)^{[a][b]+[b]}\kappa^b \partial^a \right) + \frac{1}{4}(-1)^{[a]+[b]} \kappa^a \kappa^b \; .
\end{equation}
Similarly we obtain
\begin{align}
& c^{(a} c^{b)} c^{(a_1}\cdots c^{a_k)} = (-1)^{([a]+[b])([a_1]+\cdots+[a_k])} \partial^{a_1}\cdots \partial^{a_k} \ \partial'^a \partial'^b e^{((\kappa+\kappa')\cdot c)-\frac{1}{2}(\kappa'\cdot \kappa)} \Big|_{\kappa,\kappa'=0} &
 \notag \\
& = (-1)^{([a]+[b])([a_1]+\cdots+[a_k])} \partial^{a_1}\cdots \partial^{a_k} \; \,
  [-]^{ab} \, e^{(\kappa \cdot c)} \Big|_{\kappa=0} , &
 \label{Fcc}
\end{align}
where (cf. (\ref{plus}))
\begin{equation}
\label{minus}
[-]^{ab}=\partial^a \partial^b - \frac{1}{2} \left( \epsilon (-1)^{[a]} \kappa^a \partial^b - (-1)^{[a][b]+[b]}\kappa^b \partial^a \right) + \frac{1}{4}(-1)^{[a]+[b]} \kappa^a \kappa^b \; .
\end{equation}
Hence using (\ref{ccF}) and (\ref{Fcc})
 we write the supercommutator
  (\ref{SupComm}) of $c^{(a_1}\cdots c^{a_k)}$ and $c^{(a}c^{b)}$ as
\begin{equation} \label{eq:CommBasis}
[c^{(a_1}\cdots c^{a_k)}, c^{(a}c^{b)}]_\pm = \partial^{a_1}\cdots \partial^{a_k} \left( \epsilon (-1)^{[a]} \kappa^a \partial^b - (-1)^{[a][b]+[b]}\kappa^b \partial^a\right) e^{(\kappa \cdot c)} \Big|_{\kappa=0}.
\end{equation}

\begin{proposition} \label{propFc} The elements
 \begin{equation}
 \label{invars}
 \varepsilon_{a_1b_1}\dots \varepsilon_{a_kb_k} \;
 c_1^{(a_1}\cdots c_1^{ a_k)} c_2^{(b_k} \cdots c_2^{b_1)}
 \in {\cal A}\otimes {\cal A}
  \end{equation}
 are invariant under the action (\ref{transU})
 of the supergroup $OSp$:
 \begin{equation}
 \label{transUc}
 c^a \;\; \to \;\; U^a_{\;\; b} \, c^b \; .
 \end{equation}
 It means that the elements (\ref{invars}) are
 invariant under the action of the Lie superalgebra  $osp$ and
 satisfy the  infinitesimal form (\ref{AFc3}) of the
 invariance condition
\begin{equation}
 \label{invarC}
\left[\varepsilon_{a_1b_1}\dots \varepsilon_{a_kb_k} c_1^{(a_1}\cdots c_1^{ a_k)} c_2^{(b_k} \cdots c_2^{b_1)} ,F_1^{ab} +F_2^{ab}\right] = 0 \; ,
\end{equation}
 where
 \begin{equation}
\label{defF12}
F_1^{ab} \equiv \epsilon (-1)^{b}c_1^{(a}c_1^{b)}
 \; , \;\;\;\  F_2^{ab} \equiv \epsilon (-1)^{b}c_2^{(a}c_2^{b)} \;
\end{equation}
 are generators of $osp$  (see proposition {\bf \ref{propF}}).
\end{proposition}
\noindent
{\bf Proof.}
 According to (\ref{transUc})
 the element $U \in OSp$ acts on the product
 $c^{b_k} \cdots c^{b_1}$ as
 (see (\ref{s-act})):
 \begin{equation}
 \label{s-act01}
c^{k \rangle} \cdots c^{1 \rangle} \;\; \to \;\; U_k  U_{\{k,k-1\}} \cdots U_{\{k,\dots,1\}}
\, c^{k \rangle} \cdots c^{1 \rangle} \; ,
\end{equation}
where $U_{\{k,\dots,j\}} = (-)^{k,k-1} \cdots (-)^{k,j} U_j (-)^{k,j} \cdots (-)^{k,k-1}$.
 Let the oscillators $c^a$ commute as in (\ref{suposc}), where
 in the right hand side we put $\varepsilon^{ab} = 0$. Then we have
 $c^{(b_k} \cdots c^{b_1)} = c^{b_k} \cdots c^{b_1}$ and (\ref{s-act01}) gives
 \begin{equation}
 \label{s-act02}
c^{(\; k \rangle} \cdots c^{1 \rangle \;)} \;\; \to \;\; U_k  U_{\{k,k-1\}} \cdots U_{\{k,\dots,1\}}
\, c^{(\; k \rangle} \cdots c^{1 \rangle \;)} \; ,
\end{equation}
where the parentheses $(\dots)$ denote the supersymmetrization.
 Since the commutation relations of elements $U^a_{\;\; b}$ and $c^d$ are independent
 of the right hand side of (\ref{suposc}), the transformation rule
 (\ref{s-act02}) will be the same for the algebra of super-oscillators (\ref{suposc}).
 From (\ref{transUc}) and in view of the invariance of the
 bilinear form (\ref{bf}) we have the transformation rule
 for new variables $\bar{c}_a \equiv \varepsilon_{da} c^d$:
 \begin{equation}
 \label{transUd}
 \bar{c}_a   \;\; \to \;\;  \bar{c}_b  \, (U^{-1})^b_{\;\; a}
 \;\; \Leftrightarrow  \;\;  \bar{c}_{\langle j }  \;\; \to \;\;
  \bar{c}_{\langle j } \, U_j^{-1} \; ,
 \end{equation}
 where $\bar{c}_a = \epsilon (-1)^{[a]} c_a$ and
  $j$ denotes the label of the  superspace.
 Arguing as above we obtain the transformation rule
 for the supersymmetrized product of the
 super-oscillators $\bar{c}_a$:
 \begin{equation}
 \label{s-act03}
 \bar{c}_{(\, \langle 1 } \cdots \bar{c}_{\langle k \,)} \;\; \to \;\;
  \bar{c}_{(\, \langle 1 } \cdots \bar{c}_{\langle k \,)}
  U_{\{k,\dots,1\}}^{-1}  \cdots U_{\{k,k-1\}}^{-1}  U_k^{-1} \; .
 \end{equation}
 From eqs. (\ref{s-act02}) and (\ref{s-act03}) we immediately see that
 the element
 $$
 \varepsilon_{a_1b_1}\dots \varepsilon_{a_kb_k}
 c_1^{(a_1}\cdots c_1^{ a_k)} c_2^{(b_k} \cdots c_2^{b_1)} =
 \bar{c}_{1_{(\, \langle 1 }} \cdots \bar{c}_{1_{\langle k \,)}} \,
 c_2^{(\; k \rangle} \cdots c_2^{1 \rangle \;)}
 \in {\cal A} \; ,
 $$
 is invariant under the action
 (\ref{transUc}) of the supergroup $OSp$.
  Considering now the infinitesimal form of this action $U = I + A + \dots$
 and taking into account eqs. (\ref{AFc2})
 we deduce the condition (\ref{invarC}).
\hfill \qed

We present the direct proof of (\ref{invarC}) in appendix {\bf
\ref{AppInv}}, giving an alternative of the above proof.

\subsection{The construction of the R-operator}
\label{consR}

Having introduced generating functions as an effective formulation of the
 supersymmetrization of super-oscillators,
we are prepared to solve the conditions \eqref{eq:Rcond1} and \eqref{eq:Rcond2a}
imposed on the R-operator $\hat{\mathcal{R}}_{12}(u)$.

The condition \eqref{eq:Rcond1} says that the R-operator has to be invariant w.r.t.
the Lie superalgebra  $osp$.
Therefore, it has to be a sum of $osp$-invariants (\ref{invars})
\begin{equation}
 \label{anzR}
\check{\mathcal{R}}_{12}(u) = \sum_k \frac{r_k(u)}{k!} \, \varepsilon_{\vec{a},\vec{b}} \
 c_1^{(a_1\dots a_k)} c_2^{(b_k\dots b_1)} \; ,
\end{equation}
where we use the concise notation
$$
\varepsilon_{\vec{a} ,\vec{b}} = \varepsilon_{a_1b_1}\dots \varepsilon_{a_kb_k}  , \;\;\;\; c_1^{(a_1\dots a_k)} := c_1^{(a_1} \cdots c_1^{a_k)}, \;\;\;\;  c_2^{(b_k\dots b_1)} := c_2^{(b_k}\cdots c_2^{b_1)}  .
$$
Inserting this ansatz into the condition  \eqref{eq:Rcond2al}, we obtain
\begin{equation}
\begin{array}{c}
\sum_k \frac{r_k(u)}{k!} \varepsilon_{\vec{a},\vec{b}} \Big\{ \epsilon u (-1)^{[b]} \left[ c_1^{(a_1\dots a_k)} c_2^{(b_k\dots b_1)} c_2^{(a b)} - c_1^{(ab)}
c_1^{(a_1\dots a_k)} c_2^{(b_k\dots b_1)}  \right] - \\ [0.3cm]
-(-1)^{([b]+[c])([a]+[d])+[b]+[d]} \varepsilon_{cd} \left[ c_1^{(a_1 \dots a_k)}
c_2^{(b_k\dots b_1)} c_1^{(cb)} c_2^{(ad)} - c_1^{(cb)} c_2^{(ad)}
c_1^{(a_1\dots a_k)} c_2^{(b_k \dots b_1)}  \right] = 0 \; .
\end{array}
\end{equation}
We use now the advantage of the generating function formulation
developed in the last subsection
 (in particular we apply relations (\ref{ccF}), (\ref{Fcc}))
 and rewrite the equation as
\begin{align}
\sum_k \frac{r_k(u)}{k!} \varepsilon_{\vec{a},\vec{b}} \partial_1^{a_1}\cdots \partial_1^{a_k} \partial_2^{b_k}\cdots \partial_2^{b_1} \Big\{ \epsilon u (-1)^{[b]}\left(  [+]_2 ^{ab} - [-]_1 ^{ab}  \right) - \notag \\
(-1)^{([b]+[c])([a]+[d])+[b]+[d]} \varepsilon_{cd} \left(  [+]_1^{cb} [+]_2^{ad} -  [+]_1^{cb} [+]_2^{ad}  \right) \Big\} e^{(\kappa_1\cdot c_1)} e^{(\kappa_2\cdot c_2)} \Big|_{\kappa_1,\kappa_2=0}=0 , \label{ccFF}
\end{align}
where the notation $[\pm]^{ab}$ was introduced
in (\ref{plus}), (\ref{minus}).
We also see that
$$
\varepsilon_{\vec{a},\vec{b}} \; \partial_1^{a_1}\cdots \partial_1^{a_k}
\; \partial_2^{b_k}\cdots \partial_2^{b_1}=(\partial_1\cdot\partial_2)_1^k=(\partial_\lambda)^k e^{\lambda(\partial_1\cdot\partial_2)_1}|_{\lambda=0} \; ,
$$
 and obtain
\begin{align} \label{eq:Rcond2b}
&\sum_k \frac{r_k(u)}{k!}  (\partial_\lambda)^k e^{\lambda(\partial_1\cdot\partial_2)_1} \Big|_{\lambda=0} \Big\{  \epsilon u (-1)^{[b]} \left(  [+]_2 ^{ab} - [-]_1 ^{ab}  \right) - \\
& \qquad - (-1)^{([b]+[c])([d]+[a])+[b]+[d]} \varepsilon_{cd} \left(  [+]_1^{cb} [+]_2^{ad} - [-]_1^{cb} [-]_2^{ad}  \right) \Big\} e^{(\kappa_1\cdot c_1)} e^{(\kappa_2\cdot c_2)} \Big|_{\kappa_1,\kappa_2=0}=0. \notag
\end{align}
We want to commute all the partial derivatives $\partial_1,\partial_2$
to the right and the variables $\kappa_1,\kappa_2$ to the left and then
apply $\kappa_1=0,\kappa_2=0$.
 For this purpose we need to know how the operator
 $e^{\lambda(\partial_1\cdot\partial_2)_1}$ acts on the variables
$\kappa_1,\kappa_2$
\begin{align}
e^{\lambda(\partial_1\cdot\partial_2)_1} \kappa_1^a = (\kappa_1^a - \lambda \epsilon (-1)^{[a]} \partial_2^a) e^{\lambda(\partial_1\cdot\partial_2)_1}, \\
e^{\lambda(\partial_1\cdot\partial_2)_1} \kappa_2^a = (\kappa_2^a + \lambda \epsilon (-1)^{[a]} \partial_1^a) e^{\lambda(\partial_1\cdot\partial_2)_1}.
\end{align}
First of all
\begin{equation} \label{eq:aux1}
e^{\lambda(\partial_1\cdot \partial_2)_1} \left(  [+]_2 ^{ab} - [-]_1 ^{ab}  \right) \Big|_{\kappa_1,\kappa_2=0} = \left(\frac{\lambda^2}{4}-1\right) \left( \partial_1^a \partial_1^b - \partial_2^a \partial_2^b \right) e^{\lambda(\partial_1\cdot \partial_2)_1}.
\end{equation}
Let us denote
\begin{align}
 Y^{(cb)(ad)}  =&  (-1)^{([b]+[c])([a]+[d])+[b]+[d]} ([+]_1^{cb} [+]_2^{ad} - [-]_1^{cb} [-]_2^{ad})=  (-1)^{([b]+[c])([a]+[d])+[b]+[d]} \notag \\
 & \times \Big\{\left(\partial_1^c \partial_1^b + \frac{1}{4}(-1)^{[c]+[b]}\kappa_1^c \kappa_1^b \right)\left(\epsilon(-1)^{[a]}\kappa_2^a\partial_2^d  - (-1)^{[a][d]+[d]}\kappa_2^d \partial_2^a\right)  \\
& \quad+ \left(\epsilon(-1)^{[c]}\kappa_1^c \partial_1^b -  (-1)^{[c][b]+[b]} \kappa_1^b \partial_1^c\right)\left(\partial_2^a \partial_2^d + \frac{1}{4}(-1)^{[a]+[d]}\kappa_2^a \kappa_2^d\right)\Big\}. \notag
\end{align}
Commuting  $e^{\lambda(\partial_1\cdot\partial_2)_1}$ through the operator
 $Y^{(cb)(ad)}$ and imposing $\kappa_1,\kappa_2=0$ we obtain
\begin{align}
&e^{\lambda(\partial_1\cdot\partial_2)_1} Y^{(cb)(ad)} \Big|_{\kappa_1,\kappa_2=0} = (-1)^{([b]+[c])([a]+[d])+[b]+[d]} \notag\\
&  \times \Big\{\lambda \Big[\left(\partial_1^c \partial_1^b + \frac{\lambda^2}{4}\partial_2^c \partial_2^b \right) \left(\partial_1^a \partial_2^d -\epsilon (-1)^{[a][d]}\partial_1^d\partial_2^a\right) \notag \\
& \qquad\qquad - \left( \partial_2^c \partial_1^b - \epsilon (-1)^{[b][c]} \partial_2^b \partial_1^c \right) \left( \partial_2^a \partial_2^d + \frac{\lambda^2}{4} \partial_1^a \partial_1^d \right) \Big]+ \\
&\qquad +\frac{\epsilon\lambda^2}{4}  \Big[\Big( (-1)^{[a]}\varepsilon^{ab}\partial_2^c \partial_2^d -(-1)^{[a][b]+[a]} \varepsilon^{ac} \partial^b_2 \partial_2^d \notag\\
& \qquad\qquad  - (-1)^{[a][d]+[d]}\varepsilon^{db}\partial_2^c \partial_2^a +(-1)^{([a]+[b])[d]+[d]} \varepsilon^{dc} \partial_2^b\partial_2^a \Big) - \left(2\rightarrow 1 \right)\Big]  \Big\} e^{\lambda(\partial_1\cdot\partial_2)_1}. \notag
\end{align}
Therefore
\begin{align} \label{eq:aux2}
 \varepsilon_{cd} e^{\lambda(\partial_1\cdot\partial_2)_1} Y^{(cb)(ad)} = (-1)^{[a][b]+[b]}\left[\lambda\left(\frac{\lambda^2}{4}+1\right)(\partial_1\cdot \partial_2)_1- \frac{\lambda^2}{4}(\omega-2)\right] \notag \\
 \times \left(\partial_1^b \partial_1^a - \partial_2^b \partial_2^a \right) e^{\lambda(\partial_1\cdot\partial_2)_1}
\end{align}
where $\omega=\varepsilon^{cd}\varepsilon_{cd}$ and we used the skew-symmetry of the product $(\partial_1\cdot \partial_2)_1=-(\partial_2\cdot \partial_1)_1$.
Inserting \eqref{eq:aux1} and \eqref{eq:aux2} into \eqref{eq:Rcond2b} and using the fact that $(\partial_1\cdot \partial_2)_1= \partial_\lambda e^{\lambda(\partial_1\cdot\partial_2)_1}|_{\lambda=0}$, we rewrite \eqref{eq:Rcond2b} as
\begin{align}
\sum_k \frac{r_k(u)}{k!} (\partial_\lambda)^k \left\{ \left(\frac{\lambda^3}{4} +\lambda\right) \partial_\lambda + \frac{\lambda^2}{4}(u-\omega+2) -u \right\} e^{\lambda(\partial_1\cdot\partial_2)_1} \Big|_{\lambda=0} & \notag\\
\times\left( \partial_1^a\partial_1^b - \partial_2^a \partial_2^b \right) e^{(\kappa_1\cdot c_1)+(\kappa_2\cdot c_2)} \Big|_{\kappa_1,\kappa_2=0} &=0.
\end{align}
By means of the general formula
\begin{equation}
(\partial_\lambda)^k \lambda^r = \sum_{i\geq 0} \frac{r!k!}{i!(r-i)!(k-i)!} \lambda^{r-i} \partial_\lambda^{k-i} \; ,
\end{equation}
we commute the derivatives w.r.t. $\lambda$ to the right and obtain
\begin{align}
\sum_k \frac{r_k(u)}{k!}  \left\{ (k-u) (\partial_\lambda)^k + \frac{k(k-1)}{4} (k +u-\omega) (\partial_\lambda)^{k-2} \right\} e^{\lambda(\partial_1\cdot\partial_2)_1} \Big|_{\lambda=0}& \notag\\ \times\left( \partial_1^a\partial_1^b - \partial_2^a \partial_2^b \right) e^{(\kappa_1\cdot c_1)+(\kappa_2\cdot c_2)} \Big|_{\kappa_1,\kappa_2=0}& =0.
\end{align}
Finally we deduce the recurrence relation for $r_k(u)$
\begin{equation} \label{Rrec1}
r_{k+2}(u) = \frac{4(u-k)}{k+2+u-\omega} r_k(u)
\end{equation}
which is solved in terms of the $\Gamma$-functions:
\begin{equation}
 \label{sol-rm}
 \begin{array}{rl}
r_{2m}(u) &= (-4)^{m} \frac{\Gamma(m-\frac{u}{2})}{\Gamma(m+1+\frac{u-\omega}{2})} A(u), \\
r_{2m+1}(u) &= (-4)^{m} \frac{\Gamma(m-\frac{u-1}{2})}{\Gamma(m+1+\frac{u-\omega+1}{2})} B(u)  \; ,
\end{array}
\end{equation}
where $\omega = \epsilon(N-M)$ (see (\ref{oMN})),  and $A(u),B(u)$ are arbitrary functions of $u$.
Substitution of (\ref{sol-rm}) in (\ref{anzR}) gives
the expression for the $osp$-invariant $R$-matrix.

This expression for the $osp$-invariant $R$-matrix generalizes the formulas for the $so$-type $R$-matrices obtained in \cite{Witten}, \cite{CDI2} (see also \cite{KarT},\cite{ZamL},\cite{Resh},\cite{Ogiev},\cite{IsKarKir15}). The $so$- and $sp$-invariant $R$-matrices are obtained easily by restriction to the corresponding Lie subalgebras of $osp$.

The bosonic part of $osp(N|M)$ (in the case $\epsilon=1$) corresponds to the embedded subalgebra $so(N)$. Similarly the fermionic part (in the case $\epsilon=-1$) corresponds to the embedded Lie subalgebra $so(M)$. Hence, restricting ourselves to $so\subset osp$ in \eqref{Rrec1} we obtain the recurrence relations for the coefficients $r_k(u)$ of the $so(d)$-symmetric R-operator
\begin{equation} \label{Rrec2}
r_{k+2}(u) = \frac{4(u-k)}{k+2+u-d} r_k(u)
\end{equation}
with the solution
\begin{equation}
\label{Rrec7}
 \begin{array}{rl}
r_{2m}(u) &= (-4)^{m} \frac{\Gamma(m-\frac{u}{2})}{\Gamma(m+1+\frac{u-d}{2})} A(u), \\
r_{2m+1}(u) &= (-4)^{m} \frac{\Gamma(m-\frac{u-1}{2})}{\Gamma(m+1+\frac{u-d+1}{2})} B(u).
\end{array}
\end{equation}
Moreover, in such a restriction the supersymmetrizers \eqref{GradSym}
appearing in the ansatz \eqref{anzR} transfer to the antisymmetrisers.
This result coincides with the results obtained in \cite{Witten}, \cite{CDI2}.
Indeed, after the rescaling of the spectral parameter  $u\rightarrow -u$ and
of the generators $c^a\rightarrow \sqrt{2} c^a$ one can directly see the coincidence
with \cite{CDI2}. The rescaling $c^a\rightarrow \sqrt{2} c^a$ gives the standard
Clifford algebra $c^ac^b+c^bc^a=2\varepsilon^{ab}$   for $so\subset osp$ which
was used in \cite{Witten}, \cite{CDI2} instead of the algebra
$\mathcal{A}$ \eqref{suposc} used in this text. Moreover, the generators
$F^{ab}$ of $so$ \eqref{defF} used in our text differ by the factor
$-\epsilon(-1)^{[b]}=-1$ from their equivalents in \cite{CDI2}.
This is the reason that here and in the left hand side of \eqref{eq:Rcond2a}
 of the spectral parameter is to be rescaled as $u\rightarrow -u$.

Similar considerations can be done for the Lie subalgebra $sp\subset osp$. The fermionic part of $osp(N|M)$ (for $\epsilon=1$) corresponds to $sp(M)\subset osp(N|M)$. The bosonic part of $osp(N|M)$ (for $\epsilon=-1$) corresponds to $sp(N)\subset osp(N|M)$. Restricting \eqref{Rrec1} to $sp\subset osp$ we obtain the recurrence relation for for the $sp(d)$-symmetric R-operator
\begin{equation}
r_{k+2}(u) = \frac{4(u-k)}{k+2 +u+d} r_k(u)
\end{equation}
with the solution
\begin{equation}
 \begin{array}{rl}
r_{2m}(u) &= (-4)^{m} \frac{\Gamma(m-\frac{u}{2})}{\Gamma(m+1+\frac{u+d}{2})} A(u), \\
r_{2m+1}(u) &= (-4)^{m} \frac{\Gamma(m-\frac{u-1}{2})}{\Gamma(m+1+\frac{u+d+1}{2})} B(u)  .
\end{array}
\end{equation}
The supersymmetrizers \eqref{GradSym} appearing in the ansatz \eqref{anzR} transfer to the symmetrizers.

\subsection{The condition on the generators G}
\label{condG}

We intend to prove here that from the condition \eqref{eq:Rcond2ar}
follows that
$
\{G_{(bc},G_{d)a}\}_{\mp} = 0.
$.

We study $X^{(cb)(ad)}$  defined in \eqref{Xcbad}.
It possess  obviously the following two symmetries:
\begin{equation} \label{ssymX}
X^{(cb)(ad)} = -\epsilon (-1)^{[c][b]+[c]+[b]} X^{(bc)(ad)}, \quad X^{(cb)(ad)} =
 -\epsilon (-1)^{[a][d]+[a]+[d]} X^{(cb)(da)}.
\end{equation}
They are the same symmetries as of the generators $F^{cb},F^{ad}$.

Further, we see that from the properties of the superanticommutator
$\{G_{bc},G_{da}\}_{\mp} $ and \eqref{eq:Rcond2ar} follows
\begin{equation}
\left( X^{(cb)(ad)} + (-1)^{([b]+[c])([a]+[d])} X^{(ad)(cb)} \right) \{G_{bc},G_{da}\}_{\mp} = 0.
\end{equation}
Using the results of section  {\bf \ref{consR}} we  see that this equation can be rewritten as
\begin{equation} \label{condXX}
\sum_{k} \frac{r_{k}(u)}{k!} \partial_\lambda^k  \left( Z^{(cb)(ad)} + (-1)^{([b]+[c])([a]+[d])} Z^{(ad)(cb)} \right) e^{\lambda(\partial_1\cdot \partial_2)_1 }e^{(\kappa_1\cdot c_1)}e^{(\kappa_2\cdot c_2)} \Big|_{\lambda=\kappa_1=\kappa_2=0} =0
\end{equation}
where
\begin{align}
& Z^{(cb)(ad)} = (-1)^{([b]+[c])([a]+[d])+[b]+[d]} \notag \\
&  \times \Big\{\lambda \Big[\left(\partial_1^c \partial_1^b + \frac{\lambda^2}{4}\partial_2^c \partial_2^b \right) \left(\partial_1^a \partial_2^d -\epsilon (-1)^{[a][d]}\partial_1^d\partial_2^a\right) \notag \\
& \qquad\qquad - \left( \partial_2^c \partial_1^b - \epsilon (-1)^{[b][c]} \partial_2^b \partial_1^c \right) \left( \partial_2^a \partial_2^d + \frac{\lambda^2}{4} \partial_1^a \partial_1^d \right) \Big]+ \\
&\qquad +\frac{\epsilon\lambda^2}{4}  \Big[\Big( (-1)^{[a]}\varepsilon^{ab}\partial_2^c \partial_2^d -(-1)^{[a][b]+[a]} \varepsilon^{ac} \partial^b_2 \partial_2^d \notag\\
& \qquad\qquad  - (-1)^{[a][d]+[d]}\varepsilon^{db}\partial_2^c \partial_2^a +(-1)^{([a]+[b])[d]+[d]} \varepsilon^{dc} \partial_2^b\partial_2^a \Big) - \left(2\rightarrow 1 \right)\Big]  \Big\}. \notag
\end{align}

Let us investigate all the terms appearing in
$Z^{(cb)(ad)} + (-1)^{([b]+[c])([a]+[d])} Z^{(ad)(cb)}$.
It is a third order polynomial in $\lambda$:
\begin{align}
Z^{(cb)(ad)} + (-1)^{([b]+[c])([a]+[d])} Z^{(ad)(cb)} = \lambda\cdot A^{(cb)(ad)} + \frac{\epsilon\lambda^2}{4} \cdot B^{(cb)(ad)} + \frac{\lambda^3}{4} \cdot C^{(cb)(ad)}.
\end{align}
The coefficient $A^{(cb)(ad)}$ separates into two parts.
The first part contains terms with the structure $\partial_1^3\partial_2$
whereas the second part contains terms with the structure $\partial_1 \partial_2^3$.
We describe here only the first part,
the second one is analysed in the same way. The first part of $A^{(cb)(ad)}$ is:
\begin{align}
&(-1)^{([b]+[c])([a]+[d])+[b]+[d]}  \left\{  \partial_1^{cba} \partial_2^d -\epsilon (-1)^{[a][d]} \partial_1^{cbd}\partial_2^a + \right. \notag \\
& \qquad +\left. (-1)^{([b]+[c])([a]+[d])} \left( \partial_1^{adc} \partial_2^b -\epsilon (-1)^{[c][b]} \partial_1^{adb}\partial_2^c \right) \right\} = \notag \\
&= (-1)^{([b]+[c])([a]+[d])+[b]+[d]+[a][d]+[b][c]} \left\{  \partial_1^{bcd} \partial_2^a \right. + \notag \\
& \qquad +\left. \left( \partial_1^{bc} \partial_2^{d} + (-1)^{[b]([c]+[d])}\partial_1^{cd} \partial_2^{b} +(-1)^{[d]([b]+[c])} \partial_1^{db} \partial_2^{c} \right) \partial_1^a \right\}.
\end{align}
Hence, we see the following symmetry:
\begin{equation}
A^{(bd)(ac)} = (-1)^{[b][d]+[c][d]+[b]+[c]} A^{(cb)(ad)}
\end{equation}
which can be regarded as the supercyclic symmetry in three indices $bdc\rightarrow cbd$.
 The coefficient $C^{(cb)(ad)}$ is analysed in the same way and possess
the same symmetry as $A^{(cb)(ad)}$:
\begin{equation}
C^{(bd)(ac)} = (-1)^{[b][d]+[c][d]+[b]+[c]} C^{(cb)(ad)}.
\end{equation}
Moreover, it is not difficult to see that $B^{(cb)(ad)}=0$.
This implies that   $X^{(cb)(ad)} + (-1)^{([b]+[c])([a]+[d])} X^{(ad)(cb)}$
possess  the supercyclic symmetry too.
From this fact and the properties \eqref{ssymX}
we conclude that $X^{(cb)(ad)} + (-1)^{([b]+[c])([a]+[d])} X^{(ad)(cb)}$
is supersymmetric w.r.t. the cyclic permutation of any three of its four indices.

It follows from the above considerations and equation \eqref{condXX} that
\begin{equation} \label{GG}
\{G_{(bc},G_{d)a}\}_{\mp}=0
\end{equation}
where $(bcd)$ denotes the supersymmetrization over the indices $b,c,d$.
Let us remark that this supersymmetrization differs from
the supersymmetrization of the super-oscillators \eqref{symcacb}.
The corresponding symmetrizer is defined like in \eqref{GradSym} with the
replacement $\hat \sigma \to \tilde \sigma$, where for the elementary
permutation of adjacent sites $j, j+1$  $ \tilde \sigma_j = \hat \sigma_j +
[a_j] + [a_{j+1}] $.

% Here we have
%\begin{equation}
%G_{(ab)} = \frac{1}{2} \left( G_{bc} - \epsilon (-1)^{[a][b]+[a]+[b]} G_{ba} \right),
%\end{equation}
%whereas for the super-oscillators
%\begin{equation}
%c^{(a} c^{b)} = \frac{1}{2} \left( c^a c^b -\epsilon (-1)^{[a][b]} c^b c^a \right).
%\end{equation}
We summarize the results of this section in the
\begin{proposition} \label{propRosc}
The $L$-operator 
%$L_1, L_2$ with
  $$
L(u) \equiv u\cdot e\otimes\mathbf{1} -
\frac{1}{2} F^{ab} \otimes G_{ba} 
 $$
constructed from the super-oscillator $osp$ generators (\ref{defF})
and $osp$ generators $G_{ba}$, which solve the additional constraint
$$\{G_{(bc},G_{d)a}\}_{\mp}=0 \; , $$
 obeys the spinorial  RLL-relation
$$
\check{\mathcal{R}}_{12}(u) L_1(u+v) L_2(v) =
L_1(v) L_2(u+v) \check{\mathcal{R}}_{12}(u) \; ,
$$
 where  super-spinorial R-operator
 $\check{\mathcal{R}}_{12}(u) \in {\cal A} \otimes {\cal A}$ is
$$ 
 \check{\mathcal{R}}_{12}(u) = \sum_k \frac{r_k(u)}{k!}
  \, \varepsilon_{\vec{a},\vec{b}} \
 c_1^{(a_1\dots a_k)} c_2^{(b_k\dots b_1)},
$$
$$ r_{2m}(u) = (-4)^{m} \frac{\Gamma(m-\frac{u}{2})}{\Gamma(m+1+\frac{u-\omega}{2})}
A(u), \ \ \
r_{2m+1}(u) = (-4)^{m} \frac{\Gamma(m-\frac{u-1}{2})}{\Gamma(m+1+\frac{u-\omega+1}{2})}
 B(u) \; .
$$
Here $\omega = \epsilon(N-M)$ (see (\ref{oMN})),  and $A(u),B(u)$
 are arbitrary functions of $u$. 
\end{proposition}

\section{The fusion of super-spinor $L$ operators}
\setcounter{equation}{0}

 It was shown in \cite{Resh} (Theorem 3)
  %\marginpar{\bf \large Is, \it RK}
 that the $so$-type $L$-operator
 (i.e. the spinor-vector\footnote{Here "spinor-vector"
 (or "spinor-spinor") means that the
$R$-matrix acts in the space of the tensor product
$V_s \otimes V_f$
(or $V_s \otimes V_s$), where $V_s$ and $V_f$ are spinor
 and vector (fundamental) representation spaces.} $so$-type $R$-matrix) can be obtained
 by fusion of two spinor-spinor $so$-type $R$-matrices $\check{\mathcal{R}}$. The main result of Sect. {\bf \ref{RFRF}} about the generalization of
 the matrix $\check{\mathcal{R}}$ to the case
 of the Lie superalgebra $osp$ is given in Proposition {\bf \ref{propRosc}}
 %(\ref{anzR}), (\ref{sol-rm}) 
 (see eq. (\ref{Rrec7})
at the end of Subsection {\bf \ref{consR}}
for the explanation of reducing to the $so$-case).
 The vector-vector $so$-type
 $R$-matrix (analog of the $osp$-type
 $R$-matrix (\ref{eq:Rmatrix})) was
 obtained in \cite{Resh} (Theorem 5) by the fusion of
two spinor-vector $so$-type $R$-matrices. The
standard fusion procedure \cite{KRS}
applied  in \cite{Resh}
requires the use of the projector
operators $V_s \otimes V_s \to V_f$ which are not simple
 objects, so that the fusion procedure
 of \cite{Resh} is  technically non-trivial.
 In  \cite{IsKarKir15}, for the cases
of the $so$ and $sp$ Lie algebras, we found that the vector-vector
$R$-matrix can be constructed as the fusion of
two $so$- and $sp$- type $L$-operators
using instead of those projectors the
  intertwining operators
 $V_s \otimes V_s \to V_f$ which are realized
respectively in terms of gamma-matrices and
generators of the oscillator algebra.
In this Section we generalize the fusion procedure of the paper \cite{IsKarKir15}
to the case of the $osp$ Lie superalgebra.

The $RLL$ relation (\ref{eq:RLL}) has the following component form,
\begin{equation}\label{rl2}
 \begin{array}{c}
(-1)^{([b_2]+[c_2])[c_1]} \;
R^{a_1a_2}_{b_1b_2}(u-v) \; L^{b_1}_{c_1}(u-\lambda)L^{b_2}_{c_2}(v-\lambda) = \\ [0.3cm]
= (-1)^{([a_2]+[b_2])[a_1]} \; L^{a_1}_{b_1}(v-\lambda)L^{a_2}_{b_2}(u-\lambda) \;
R^{b_1b_2}_{c_1c_2}(u-v) \; .
\end{array}
\end{equation}
It is convenient to introduce the supertensor product $\otimes_s$
modifying the graded tensor products as
\begin{equation}\label{tpro1}
(A\otimes_sB)^{a_1a_2}_{b_1b_2}=(-)^{([a_2]+[b_2])[b_1]}
A^{a_1}_{b_1}B^{a_2}_{b_2}.
 \end{equation}
It has the important property of associativity,
\begin{equation}\label{tpro2}
\quad (A\otimes_sB)(C\otimes_sD)=AC\otimes_sBD \; .
 \end{equation}
This is checked by the following calculation
$$
(A\otimes_sB)^{a_1a_2}_{b_1b_2}(C\otimes_sD)^{b_1b_2}_{c_1c_2}=
(-)^{([a_2]+[b_2])[b_1]}A^{a_1}_{b_1}B^{a_2}_{b_2}(-)^{([b_2]+[c_2])[c_1]}
C^{b_1}_{c_1}D^{b_2}_{c_2}=
$$
 %$$=(-)^{([a_1]+[b_1])[a_2]+([b_1]+[c_1])[b_2]
 %+([a_2]+[b_2])([b_1]+[c_1])}A^{a_1}_{b_1}C^{b_1}_{c_1}
 %B^{a_2}_{b_2}D^{b_2}_{c_2}=$$
$$
=(-)^{([a_2]+[c_2])[c_1]}A^{a_1}_{b_1}C^{b_1}_{c_1}B^{a_2}_{b_2}
D^{b_2}_{c_2}=(AC\otimes_sBD)^{a_1a_2}_{c_1c_2}.
$$

Using the supertensor product one can represent the graded RLL relation
(\ref{rl2}) in the form\footnote{One can define another
supertensor product (cf. (\ref{tpro1})) $(A\otimes_s' B)^{a_1a_2}_{b_1b_2}=(-1)^{([a_1]+[b_1])[a_2]}
A^{a_1}_{b_1}B^{a_2}_{b_2}$ which respects the property (\ref{tpro2})
as well. For this supertensor product the $RLL$-relation (\ref{rl1})
is equivalent to (\ref{RLL2}).}:
\begin{equation}\label{rl1}
R(u-v)\Bigl(L(u-\lambda)\otimes_sL(v-\lambda)\Bigr)=
\Bigl(L(v-\lambda)\otimes_sL(u-\lambda)\Bigr)R(u-v) \; .
\end{equation}

Consider
two different $L$-operators $L(v-\mu)$ and $L'(u-\lambda)$
 which commute up to the
 standard sign factor according to grading:
$$
L^{a_2}_{b_2}(u-\lambda)L'^{b_1}_{c_1}(u-\mu)=
(-)^{(b_1+c_1)(a_2+b_2)}L'^{b_1}_{c_1}(u-\mu)
L^{a_2}_{b_2}(u-\lambda).
$$
This means that for the supertensor
 products (\ref{tpro1}) we have
 \begin{equation}\label{rl5}
 (L_1 \otimes_s L_2)(L'_1 \otimes_s L'_2) =
 (L_1 \, L'_1 \; \otimes_s \; L_2 \, L'_2) \; .
 \end{equation}
We assume the $RLL$ relation (\ref{rl2}), (\ref{rl1}) to hold for $L$ replaced by $L'$.
The property (\ref{rl5}) allows to apply the "train argument" \cite{LDF}
in the fusion procedure as in the non-supersymmetric case without extra
signs responsible for the grading.
Thus the $RLL$ relation (\ref{rl2}), (\ref{rl1}) holds for the matrix product
$$
T(u) = L(u+\lambda) L'(u+\mu),
$$
where $\lambda,\mu$ are any shifts of the spectral parameter. In components
this matrix product reads as
\begin{equation}
 \label{rl4}
(T^b_{\;\; d}(u))^{\alpha_1 \alpha_2}_{\;\;\; \beta_1 \beta_2} =
(L^b_{\;\; c}(u+\lambda))^{\alpha_1}_{\;\; \beta_1}
 (L^{\prime \, c}_{\;\;\; d}
 (u+\mu))^{\alpha_2}_{\;\;\beta_2}  \; ,
\end{equation}
where indices $a,b,\dots$ label the coordinates
in the vector (fundamental)
representation space $V_f$, while indices $\alpha,\beta,\dots$ are formal
 indices of the coordinates of the representation space $V_s$
in which the super-oscillator algebra ${\cal A}$
acts. Now the formal matrix
 $(c^a)^{\alpha}_{\;\;\beta}$  of the super-oscillator generators $c^a$
 plays the role of the intertwiner: $\overline{V}_s \otimes V_s \to V_f$,
 where $\overline{V}_s$ denotes the space which is dual to $V_s$.
 The spaces $\overline{V}_s$ and $V_s$ are identified with
 the help of the metric $D_{\alpha \beta}$ and inverse
 metric $D^{\alpha \beta}$ which can be used for
 lowering and rising indices $\alpha,\beta,\dots$.
  Now we define the operator $L^{\prime}(u)$ in (\ref{rl4}) as
 following
 \begin{equation}
 \label{Lprime}
 (L^{\prime \, c}_{\;\;\; d}
 (u))^{\alpha}_{\;\;\beta} =
 u \, \delta^{c}_{d} \delta^{\alpha}_{\beta}
 + (F^{\prime \, c}_{\;\;\; d})^{\alpha}_{\;\;\beta}
 = u \, \delta^{c}_{d} \delta^{\alpha}_{\beta}
 - D^{\alpha \gamma}\,
 (F^{c}_{\;\;\; d})^{\gamma'}_{\;\;\gamma}
 \, D_{\gamma' \beta} \; ,
 \end{equation}
 where $(F^{\prime \, c}_{\;\;\; d})^{\alpha}_{\;\;\beta}
 = - D^{\alpha \gamma}\, (F^{c}_{\;\;\; d})^{\alpha}_{\;\;\gamma}
 \, D_{\gamma' \beta}$ are generators of the $osp$ Lie
 superalgebra. Indeed
  one can check directly that they satisfy the
 graded commutation relations (\ref{osp2}),
 (\ref{t1}). The elements $F'$
 define a representation of $osp$ which is
  contra-gradient to the representation given by
 elements (\ref{defF}).
  %\marginpar{\it RK}
Further,  the  generators
 $(L^{(1)})^{c}_{\;\; d} = F^{c}_{\;\; d}$
 satisfy the conditions
 of Proposition {\bf \ref{prop22}} by construction. This
implies  that the generators
 $F^{\prime \, c}_{\;\; d}$ obey the conditions
 of Proposition {\bf \ref{prop22}} as well.
 Thus, both operators
 $L'$ and $T$ given in (\ref{rl4}) and (\ref{Lprime})
 satisfy the $RLL$ relations (\ref{rl2}).

 The projection
 of the elements $T^b_{\;\; d}(u)$  which are
 operators in the space $V_s \otimes V_s$
 to the operators in the space $V_f$ gives
 us the desirable fusion of two $L$ operators
 to the vector-vector $R$-matrix.
 This projection can be done by the invariant contraction
 of the matrices (\ref{rl4}) with two intertwiners
 $(c_{d_2})^{\beta_1\beta_2} =
 (c_{d_2})^{\beta_1}_{\;\; \beta'} D^{\beta_2 \beta'}$ and
 $(c^{b_2})_{\alpha_2 \alpha_1} =
  (c^{b_2})^{\alpha'}_{\;\; \alpha_1}  D_{\alpha' \alpha_2}$:
 \begin{equation}
 \label{rl6}
 \begin{array}{c}
 (T^{b_1}_{\;\; d_1}(u))^{\alpha_1 \alpha_2}_{\;\;\; \beta_1 \beta_2} \;
 (c_{d_2})^{\beta_1\beta_2}  \; (c^{b_2})_{\alpha_2 \alpha_1} = \\ [0.2cm] =(-1)^{(c_1+d_1)d_2}{\textrm{Tr}}\Bigl(L^{b_1}{}_
{c_1}(u+\lambda)c_{d_2}\tilde L^{c_1}{}_{d_1}(u+\mu)c^{b_2}\Bigr)
\equiv{\mathbb{T}}^{b_1b_2}_{\;\; d_1d_2}(u)  \; ,
\end{array}
 \end{equation}
 where
 $$
(\tilde L^{c_1}_{\;\; d_1}(u+\mu))^{\beta'}_{\;\;\alpha'} :=
 D_{\alpha'\alpha_2}  \; (L^{\prime \, c_1}_{\;\;\; d_1}(u+\mu))^{\alpha_2}_{\;\;\;\beta_2}
\; D^{\beta_2\beta'} =
(u + \mu) \delta^{c_1}_{d_1} \delta^{\beta'}_{\alpha'}
-  (F^{c_1}_{\;\; d_1})^{\beta'}_{\;\;\alpha'} \; .
$$
 We show now that the fusion expression(\ref{rl6}) is related to the
fundamental $R$ matrix
up to  a multiplication by certain sign factors which will be fixed at the end
 of this Section.
Traces ${\textrm{Tr}}$ of products
 of super-oscillators with definite grading
are fixed by the symmetry arguments. In the cases of $so$, $(\epsilon = +1)$
and $sp$, $(\epsilon = -1)$ we had in \cite{IsKarKir15}
$$
{\textrm{Tr}}(c^ac^b)=\varepsilon^{ab}\, {\textrm{Tr}}\mathbf{1} \; , \qquad
{\textrm{Tr}}(c^ac^bc^cc^d)= \frac{1}{2}
(\varepsilon^{ab}\varepsilon^{cd}-\epsilon \varepsilon^{ac}\varepsilon^{bd}+
\varepsilon^{ad}\varepsilon^{bc})\, {\textrm{Tr}}\mathbf{1},
$$
where ${\textrm{Tr}}\mathbf{1}$ is a normalization constant
which is not important here.
In the supersymmetric case of the $osp$ algebra this is modified as follows:
\begin{equation}
\label{rl7}
{\textrm{Tr}}(c^ac^b)=\varepsilon^{ab} \,
{\textrm{Tr}}\mathbf{1} \; ,\;\;\; {\textrm{Tr}}(c^ac^bc^cc^d)=
\frac{1}{2} (\varepsilon^{ab}\varepsilon^{cd}-\epsilon(-1)^{ab}\varepsilon^{ac}
\varepsilon^{bd}+
 %(-1)^{a(b+c)}\varepsilon^{ad}\varepsilon^{bc})
 \varepsilon^{ad}\varepsilon^{bc})
\, {\textrm{Tr}}\mathbf{1}.
\end{equation}
 To simplify formulas here and below in this Section  we write
 the gradings $[a],[b],\dots$ in sign factors as $a,b,\dots$.
Now we calculate the projection (\ref{rl6}):
 %of the product of two super-oscillator
 %$L$ operators (\ref{solF2}):
 \begin{equation}
 \label{first1}
{\mathbb{T}}^{b_1b_2}_{d_1d_2}(u) =(-1)^{(c_1+d_1)d_2}{\textrm{Tr}}\Bigl(L^{b_1}{}_
{c_1}(u+\lambda)c_{d_2}\tilde L^{c_1}{}_{d_1}(u+\mu)c^{b_2}\Bigr)=
 \end{equation}
$$
=(-1)^{(c_1+d_1)d_2}{\textrm{Tr}}
\Bigl([(u+\lambda-\frac12)\delta^{b_1}_{c_1} + \epsilon(-1)^{c_1}
c^{b_1}c_{c_1}]c_{d_2}[(u+\mu+\frac12)\delta^{c_1}_{d_1} -
\epsilon(-1)^{d_1}c^{c_1}c_{d_1}]c^{b_2}\Bigr)=
$$
 \begin{equation}
 \label{first2}
\begin{array}{c}
= (-1)^{(c_1+d_1)d_2} (u+\lambda-\frac12)(u+\mu+\frac12) {\textrm{Tr}}
\Bigl( \delta^{b_1}_{c_1} c_{d_2} \delta^{c_1}_{d_1} c^{b_2}\Bigr) - \\
- (-1)^{(c_1+d_1)d_2 +d_1} \epsilon  (u+\lambda-\frac12)  {\textrm{Tr}}
\Bigl( \delta^{b_1}_{c_1} c_{d_2} c^{c_1}c_{d_1} c^{b_2}\Bigr) + \\
+ (-1)^{(c_1+d_1)d_2 + c_1} \epsilon (u+\mu+\frac12) {\textrm{Tr}}
\Bigl( c^{b_1}c_{c_1} c_{d_2} \delta^{c_1}_{d_1} c^{b_2}\Bigr) + \\
 - (-1)^{(c_1+d_1)d_2 + c_1 + d_1}{\textrm{Tr}}
 \Bigl(c^{b_1}c_{c_1} c_{d_2} c^{c_1}c_{d_1} c^{b_2}\Bigr) \; .
 %+ (-1)^{d_1 d_2}
 %\Bigl(\epsilon (-1)^{d_2} \frac{\omega}{2} - (-1)^{d_1} \Bigr)
 %{\textrm{Tr}} \Bigl(c^{b_1} c_{d_2} c_{d_1} c^{b_2}\Bigr) \; ,
\end{array}
\end{equation}
In the last line we commute $c_{d_2}$ with $c^{c_1}$
and use the identity (\ref{Contract}) which leads to
 $$
 \begin{array}{c}
(-1)^{c_1d_2 + c_1 + d_1}{\textrm{Tr}}
\Bigl(c^{b_1}c_{c_1} c_{d_2} c^{c_1}c_{d_1} c^{b_2}\Bigr)
= (-1)^{d_1}\Bigl(1 -  \frac{\omega}{2} \Bigr)
 {\textrm{Tr}} \Bigl(c^{b_1} c_{d_2} c_{d_1} c^{b_2}\Bigr) \; .
 \end{array}
 $$
Then applying formulas for traces (\ref{rl7}) we write (\ref{first2}) as
\begin{equation}
 \label{first3}
\begin{array}{c}
= (u+\lambda-\frac12)(u+\mu+\frac12) \;
\delta^{b_1}_{d_1}  \delta^{b_2}_{d_2} \; {\textrm{Tr}}({\bf 1}) - \\ [0.3cm]
- \frac{1}{2} (u+\lambda-\frac12)  \Bigl( \delta^{b_1}_{d_1}  \delta^{b_2}_{d_2}
+ \epsilon (-1)^{d_1 + d_2 +d_1d_2}
\delta^{b_1}_{d_2} \delta^{b_2}_{d_1}  -
\varepsilon_{d_1d_2} \varepsilon^{b_1 b_2}\Bigr) \; {\textrm{Tr}}({\bf 1}) + \\ [0.3cm]
+ \frac{1}{2} (u+\mu+\frac12)  \Bigl( \delta^{b_1}_{d_1}  \delta^{b_2}_{d_2}
- \epsilon (-1)^{d_1 + d_2 +d_1d_2}
\delta^{b_1}_{d_2} \delta^{b_2}_{d_1}  +
\varepsilon_{d_1d_2} \varepsilon^{b_1 b_2}\Bigr) \; {\textrm{Tr}}({\bf 1}) + \\ [0.3cm]
 + \frac{1}{2} (1- \frac{\omega}{2}) \Bigl( \delta^{b_1}_{d_1}  \delta^{b_2}_{d_2}
- \epsilon (-1)^{d_1 + d_2 +d_1d_2}
\delta^{b_1}_{d_2} \delta^{b_2}_{d_1}  -
\varepsilon_{d_1d_2} \varepsilon^{b_1 b_2}\Bigr) \; {\textrm{Tr}}({\bf 1}) =
\end{array}
 \end{equation}
\begin{equation}
 \label{first3a}
\begin{array}{c}
= \Bigl( (u+\lambda)(u+\mu) + \frac{3}{4} -  \frac{\omega}{4} \Bigr)
 \;\delta^{b_1}_{d_1}\;\delta^{b_2}_{d_2} \; {\textrm{Tr}}({\bf 1}) - \\ [0.3cm]
- \; \epsilon \; \Bigl(u +\frac12(\lambda +\mu + 1 - \frac{\omega}{2}) \Bigr) \;
 (-1)^{d_1+ d_2 + d_1d_2} \delta^{b_1}_{d_2} \delta^{b_2}_{d_1}
 \; {\textrm{Tr}}({\bf 1}) + \\ [0.3cm]
+ \Bigl(u +\frac12(\lambda +\mu - 1 + \frac{\omega}{2}) \Bigr) \;
 \varepsilon^{b_1b_2}\varepsilon_{d_1d_2} \; {\textrm{Tr}}\mathbf{1} \; .
\end{array}
 \end{equation}
Let arbitrary parameters $\lambda$ and $\mu$ be expressed
via one parameter $\kappa$ as following
$$
\mu=\kappa -\frac12 \; ,\qquad \lambda=\kappa+ \frac{3-\omega}2 \; .
$$
For this choice of the parameters we finally obtain:
$$
{\mathbb{T}}^{b_1b_2}_{d_1d_2}(u') =
\Bigl( u'(u'+\beta)\, \delta^{b_1}_{d_1}\;\delta^{b_2}_{d_2}
-(u'+\beta) (-1)^{d_1+ d_2 + d_1d_2} \delta^{b_1}_{d_2} \delta^{b_2}_{d_1}
  + u' \; \varepsilon^{b_1b_2}\varepsilon_{d_1d_2}  \Bigr)
  \; {\textrm{Tr}}\mathbf{1} \; ,
$$
where $u'=u+\kappa$
and as usual we denote $\beta = 1 - \frac{\omega}{2}$.
So we see that the projection (\ref{rl6}) leads to the result that the fusion of two conjugated
super-oscillator $L$ operators decorated by sign
factors
 %$(-1)^{b_1 d_2}
 %{\mathbb{T}}^{b_1b_2}_{d_1d_2}(u)(-1)^{d_1 b_2}$
coincides with the vector-vector
(fundamental) $osp$ $R$-matrix (\ref{eq:Rmatrix})
 %Note that the matrix $(-1)^{b_1}
 %{\mathbb{T}}^{b_1b_2}_{d_1d_2}(u)(-1)^{d_1}$ coincides
 and with the twisted
$R$-matrix  $(-)^{12}R(u)(-)^{12}$
 \begin{equation}
 \label{intw01}
 \begin{array}{c}
(-1)^{b_1 d_2}
 {\mathbb{T}}^{b_1b_2}_{d_1d_2}(u)(-1)^{d_1 b_2} =
 R^{b_1b_2}_{d_1d_2}(u) \; , \\ [0.3cm]
 (-1)^{b_1} {\mathbb{T}}^{b_1b_2}_{d_1d_2}(u)(-1)^{d_1}
 = (-1)^{b_1 d_1} R^{b_1b_2}_{d_1d_2}(u)
 (-1)^{b_2 d_2} \; .
 \end{array}
 \end{equation}
Recall that the twisted
$R$-matrix  $(-)^{12}R(u)(-)^{12}$
defines the vector (fundamental)
representation of $L$-operator (\ref{Lfunda}).

\vspace{0.2cm}

\noindent
{\bf Remark.} The Yang-Baxter equation (\ref{eq:YBEgraded})
for the vector-vector $R$-matrix follows
 from the $RLL$ relations (\ref{rl2}) for the matrices
 $(T^b_{\;\; d}(u))$ defined in (\ref{rl4}). Indeed, this statement is based on the remarkable identity
 \begin{equation}
 \label{intw02}
 (-1)^{[p]([b]+[c])} \;
 L^a_{\ c}\Bigl(v+\beta+\frac{1}{2}\Bigr)
 \; c_p \; \widetilde{L}^c_{\ b} \Bigl(v-\frac{1}{2}\Bigr) =
 (-1)^{[a][p]} \, R^{ac}_{bp}(v) \, (-1)^{[b][c]} \; c_c \; ,
 \end{equation}
 which generalizes the relations (\ref{intw01})
 and justifies the use of the super-oscillator
 generators as intertwiners.

\section{The quadratic evaluation of the Yangian $\mathcal{Y}(osp)$}
\setcounter{equation}{0}

\subsection{The conditions for the quadratic evaluation}

We derive the conditions on the terms of a quadratic evaluated L-operator following from the RLL-relation. We investigate a particular
solution for the second term.

As above we denote the operator $A$ acting non-trivially only in the first space
of a tensor  product of vector spaces as $A_1$ and the
operator $B$ acting non-trivially only in the second space of the tensor product
as $B_2$. We introduce a new symbol $\tilde{B}_2$ for the following object
\begin{equation}
\tilde{B}_2 \equiv (-)^{12} B_2 (-)^{12}
\end{equation}
 where $(-)^{12}$ is the sign operator \eqref{eq:GradOp}
introduced in the first section. It is a particular case of a sign operator
dressed operator \eqref{signop}.

Let us solve the graded RLL-relation \eqref{eq:RLL}
$$
R_{12}(u-v)L_1(u)(-)^{12} L_2(v) (-)^{12} = (-)^{12} L_2(v) (-)^{12} L_1(u)  R_{12}(u-v)
$$
for a quadratic evaluation of the L-operator
\begin{equation} \label{Lu2}
L(u) = u^2 \cdot\mathbf{1} + u\cdot G + N
\end{equation}
with the $osp$-invariant R-matrix \eqref{eq:Rmatrix}.
Expanding in $u,v$, we obtain the following set of six equations.
The rest of equations is linearly dependent on these six.
\begin{align}
A.\quad &[G_1,\tilde{G}_2] = [\epsilon \mathcal{P}-\mathcal{K}, \tilde{G}_2],& \notag\\
B.\quad &[G_1,\tilde{N}_2] = [\epsilon \mathcal{P}-\mathcal{K}, \tilde{N}_2],& \notag \\
C.\quad &[N_1,\tilde{G}_2] -2 [G_1,\tilde{N}_2] +\beta [G_1,\tilde{G}_2] = [\mathcal{K}-\epsilon \mathcal{P}, \tilde{N}_2] +\beta [\epsilon \mathcal{P},\tilde{G}_2] - \mathcal{K}G_1 \tilde{G}_2 + \tilde{G}_2 G_1 \mathcal{K},& \label{Equations} \\
D.\quad &[N_1,\tilde{N}_2] + \beta [G_1,\tilde{N}_2] = (\epsilon \mathcal{P}-\mathcal{K}) G_1\tilde{N}_2 -  \tilde{N}_2 G_1(\epsilon \mathcal{P}-\mathcal{K}) + \beta [\epsilon\mathcal{P},\tilde{N}_2],& \notag\\
E. \quad &-2[N_1,\tilde{N}_2] - \beta [G_1,\tilde{N}_2] + \beta [N_1,\tilde{G}_2]= \epsilon\mathcal{P}(N_1\tilde{G}_2-G_1\tilde{N}_2) - (\tilde{G}_2 N_1 - \tilde{N}_2 G_1)\epsilon \mathcal{P},& \notag \\
F.\quad &\beta [N_1,\tilde{N}_2] = \beta (\epsilon \mathcal{P} G_1 \tilde{N}_2 - \tilde{N}_2 G_1 \epsilon \mathcal{P}) -\mathcal{K} N_1 \tilde{N}_2 + \tilde{N}_2 N_1 \mathcal{K}.& \notag
\end{align}

Equation $A.$ says that at the first level appear the generators of $osp$. For details, please, see section \ref{sec:LinEval}. We discussed there that the generators can be arranged  supertraceless and that they satisfy
\begin{equation}
\mathcal{K}(G_1 + \tilde{G}_2) = 0 = (G_1 + \tilde{G}_2) \mathcal{K}.
\end{equation}
This can be arranged also here as a consequence of equation $A.$

Equation $B.$ is fulfilled if  the second level operator $N$ is a
linear combination  of powers of $G$,
\begin{equation}
N= \sum_{j=0}^\infty b_j G^j
\end{equation}
where $b_j$ commute with $G^k$ for all $j,k$.

Equation $C.$ can be rearranged in the following way:
\begin{equation}
[\mathcal{K}, N_1 +\tilde{N}_2] = \beta [\mathcal{K},\tilde{G}_2] + \mathcal{K} \tilde{G}_2^2 - \tilde{G}_2^2 \mathcal{K}.
\end{equation}
Multiplying this equation from both sides by the super-permutation $\mathcal{P}$
we obtain an equivalent equation
\begin{eqnarray}
[\mathcal{K}, N_1 +\tilde{N}_2] =
\beta [\mathcal{K},G_1] + \mathcal{K} G_1^2 - G_1^2 \mathcal{K}.
\end{eqnarray}
Adding these two equations, multiplying them by $\mathcal{K}$ and using the
identities \eqref{eq:Identities} we obtain
\begin{align}
N_1+\tilde{N}_2-\frac{1}{2}(\beta G_1 +G_1^2 +\beta \tilde{G}_2 + \tilde{G}_2^2) = \frac{\epsilon}{\omega} \left[ 2\mathrm{str}(N) - \mathrm{str}(G^2) \right]
\end{align}
which is obviously solved by
\begin{equation} \label{eq:N}
N=\frac{\beta}{2}G + \frac{1}{2}G^2.
\end{equation}
One can easily show that
\begin{equation}
\mathcal{K}N_1 = \mathcal{K} \tilde{N}_2, \qquad N_1 \mathcal{K} =  \tilde{N}_2 \mathcal{K}.
\end{equation}

\subsection{Generators of $osp$ in Jordan-Schwinger form}

We introduce a set of graded canonical pairs, variables $x_a$  and the corresponding
 partial derivatives $\partial_a$,
  such that (cf. (\ref{commzw}))
 \begin{equation}
 \label{commxd}
\begin{array}{c}
 x_a x_b = \epsilon (-1)^{[a][b]} x_b x_a, \qquad
\partial_a \partial_b = \epsilon (-1)^{[a][b]}
\partial_b \partial_a, \\
 \partial_a x_b - \epsilon (-1)^{[a][b]} x_b \partial_a = \varepsilon_{ab} \; .
\end{array}
\end{equation}
For $\epsilon = +1$ the variables
$\{ x_a, \partial_a \}$ can be identified with
the superspace coordinates and derivatives
  with the degree $[a]$ as introduced in sect. \ref{supgr}.
 According to (\ref{commxd}),
 for $\epsilon = -1$ and $[a]=0$ (or $[b]=0$), these variables
 anticommute while for $[a]=[b]=1$ they are commutative.
 For  $\epsilon = -1$  these variables
 behave like the graded differential forms.
  (\ref{commxd}) implies that
 the invariant  bilinear form $(x,y) = \varepsilon^{ba} x_a y_b$
 is always symmetric $(x,y) =(y,x)$.
 %With such defined variables,
 One can directly prove that the elements of this graded Heisenberg algebra
\begin{equation}
\label{Mab}
M_{ab} \equiv x_a \partial_b - \epsilon (-1)^{[a][b]+[a]+[b]} x_b \partial_a
\end{equation}
satisfy the supercommutation relations of $osp$ \eqref{eq:OSpComm}
and the symmetry condition \eqref{SymmCond1}.
Therefore, they compose a set of generators of $osp$.
One can easily check that the matrix of generators (\ref{Mab}) is supertraceless
\begin{equation}
\mathrm{str} (M)=(-1)^{[a]} M^a_{\ a} = \epsilon \varepsilon^{ba} M_{ba} = 0.
\end{equation}

\begin{proposition}\label{charM}
The matrix (\ref{Mab}) of the generators of the
Lie superalgebra $osp$ satisfies the following cubic
 characteristic condition
\begin{equation} \label{M3}
M^3= (\omega - 1)M^2 +\left( \frac{\epsilon}{2} \mathrm{str} (M^2) -\omega+2 \right) M - \frac{\epsilon}{2} \mathrm{str} (M^2) \varepsilon.
\end{equation}
\end{proposition}
\noindent
{\bf Proof}.
It is useful to introduce the following operator
\begin{equation}
H \equiv \varepsilon^{bc}x_b \partial_c = x_b \partial^b
\end{equation}
with the properties
\begin{equation}
H x_a = x_a (H+1),\qquad H\partial_a= \partial_a (H-1),\qquad [H, x_b\partial_c] =0.
\end{equation}
The square of $M$ is then expressed as
\begin{equation}
\begin{array}{c}
(M^{2})_{ad} = \varepsilon^{bc} M_{ab} M_{cd}
= (2H+\omega -4) x_a \partial_d + \\ [0.2cm]
 + (1-H)M_{ad} -\epsilon (-1)^{[d]} x_ax_d \partial^2 - \epsilon (-1)^{[a]} x^2 \partial_a \partial_d + H \varepsilon_{ad},
\end{array}
 \end{equation}
where we remind $\omega \equiv \varepsilon^{ef}\varepsilon_{ef}$ and introduce concise notation $x^2=\varepsilon^{ba}x_ax_b=x^b x_b$ and $\partial^2= \partial^b\partial_b$.
The supertrace of $M^2$ is
\begin{align}
\mathrm{str}(M^2) = \epsilon \left[ (2H+2\omega -4)H -2x^2 \partial^2 \right].
\end{align}
It is useful to use also the following identities
\begin{equation}
\partial^2 x_b = x_b \partial^2 +2\epsilon (-1)^{[b]} \partial_b, \qquad \partial_b x^2 = x^2 \partial_b +2\epsilon(-1)^{[b]} x_b, \qquad \partial^2 M_{ac} = M_{ac} \partial^2.
\end{equation}
After a lengthy calculation we obtain (\ref{M3}).
\hfill \qed

The $osp$ representation generated by $M_{ab}$ satisfies the
condition (\ref{GG}) for the linear  L- operators intertwining the
super-oscillator with the Jordan-Schwinger type representation.
\begin{proposition}
The matrix (\ref{Mab}) of the generators of the
Lie superalgebra $osp$ satisfies the condition
$$\{M_{(bc},M_{d)a}\}_{\mp}=0 $$
and the L-operator
$$ L(u) = u I -\frac{1}{2} F^{ab}M_{ba},$$
with $F^{ab}$ (\ref{defF}) generating the super-oscillator representation
obeys the RLL-relation with the super-spinorial R-operator
(\ref{anzR} ) constructed in  section (\ref{RFRF}).
\end{proposition}
\noindent
{\bf Proof}.
The super-anticommutator condition can be checked by direct calculation
using the
expression for $M_{ab}$ (\ref{Mab}).
The super-spinorial RLL-relation is fulfilled by
 the proposition (\ref{propRosc}).
\hfill \qed

The above cubic characteristic condition of (\ref{M3})
follows from the condition
(\ref{GG}) written in terms of $M$. This is the consequence of the
following:

\begin{proposition}
If the generators $G^a_{\ b}$ of $osp$ obey the condition (\ref{GG})
\begin{equation} \label{SASS}
\{G_{(a_1a_2}, G_{c_1)c_2}\}_{\mp} = 0,
\end{equation}
then the matrix $G = ||G^a_{\ b}||$ obeys the cubic characteristic identity
(\ref{M3}) written in terms of $G$ as
\begin{equation}
G^3 = (\omega-1) G^2 +
\left( \frac{\epsilon}{2} \mathrm{str}(G^2) +2-\omega \right) G -\frac{\epsilon}{2} \mathrm{str}(G^2) \varepsilon.
\end{equation}
\end{proposition}

\begin{proof}
The condition \eqref{SASS} can be rewritten as
\begin{align}
& [G_{a_1a_2},G_{c_1c_2}]_\pm +
2(-1)^{([a_1]+[a_2])([c_1]+[c_2])} G_{c_1c_2} G_{a_1 a_2} \notag\\
& + (-1)^{[c_1][a_2]+[c_1][a_1]+[a_1]+[a_2]} [G_{c_1a_1},G_{a_2c_2}]_\pm \notag\\
& + 2(-1)^{[c_1][a_2]+[c_1][a_1]+[a_1]+[a_2]+([a_1]+[a_2])([c_1]+[c_2])}
G_{a_2c_2} G_{c_1a_1} \notag\\
& + (-1)^{[a_1][a_2]+[a_1][c_1]+[a_2]+[c_1]} [G_{a_2c_1},G_{a_1c_2}]_\pm   \notag\\
& +2 (-1)^{[a_1][a_2]+[a_1][c_1]+[a_2]+[c_1]+([a_1]+[a_2])([c_1]+[c_2])} G_{a_1c_2}G_{a_2c_1} =0
\end{align}
and using the super-commutation relations \eqref{osp2} and the symmetry condition \eqref{SymmCond1}, one obtains
\begin{align}
& (-1)^{([a_1]+[a_2])([c_1]+[c_2])} G_{c_1c_2} G_{a_1 a_2} +
(-1)^{[c_1][a_1]+[c_1][c_2]+[a_1][c_2]+[a_1][a_2]+[a_1]+[a_2]} G_{a_2c_2} G_{c_1a_1}  \notag\\
& + (-1)^{[a_2][c_2]+[c_1][c_2]+[a_2]+[c_1]} G_{a_1c_2} G_{a_2c_1} = \epsilon (-1)^{[a_1][c_1]+[a_2][c_1]+[a_2][c_2]} \varepsilon_{a_1c_2} G_{c_1a_2}  \notag\\
& + \epsilon (-1)^{[c_1]+[a_2][c_1]+[a_2][c_2]} \varepsilon_{a_2c_2} G_{a_1c_1} + \epsilon (-1)^{[c_1][a_1]+[a_1][a_2]+[a_1][c_2]+[a_1]+[a_2]} \varepsilon_{c_1c_2} G_{a_2a_1}. \label{gggeg}
\end{align}
Further we use that $G$ is super-traceless $\mathrm{str}(G)=0$. Multiplying \eqref{gggeg} from the right by $(-1)^{([a_1]+[a_2])([c_1]+[c_2])} G^{a_2a_1}$ and summing over $a_1,a_2$, the left hand-side is equal to
\begin{equation}
\mathrm{Left}=\epsilon G_{c_1c_2} \mathrm{str}(G^2) +2(-1)^{[c_1][c_2]+[c_1]+[c_2]+[a_1][a_2]+[c_1]([a_1]+[a_2])} G_{c_2a_1} G_{a_2c_1}G^{a_1a_2}.
\end{equation}
Super-commuting $G_{a_2c_1}$ with $G^{a_1a_2}$, this can be  further rewritten to
\begin{equation}
\mathrm{Left} = (-1)^{[c_1][c_2]+[c_1]+[c_2]} \left\{ -G_{c_2c_1} \mathrm{str}(G^2) +2\epsilon  \left(  (G^{3})_{c_2c_1} + (2-\omega)(G^{2})_{c_2c_1} \right)  \right\}.
\end{equation}
 The right hand-side is after multiplication from the right by $(-1)^{([a_1]+[a_2])([c_1]+[c_2])} G^{a_2a_1}$ and summation over $a_1,a_2$ of the form
\begin{equation}
\mathrm{Right} = 2 (G^2)_{c_1c_2} - \epsilon\cdot \mathrm{str}(G^2) \varepsilon_{c_1c_2}.
\end{equation}
If we use the properties of the super-metric $\varepsilon_{c_1c_2}$ and the transposition rule for $G^2$
\begin{align}
(G^{2})_{c_1c_2} = \epsilon (-1)^{[c_1][c_2]+[c_1]+[c_2]} \left[ (G^2)_{c_2c_1} + (2-\omega) G_{c_2c_1} \right],
\end{align}
the right hand-side can be further rewritten to
\begin{equation}
\mathrm{Right} =  (-1)^{[c_1][c_2]+[c_1]+[c_2]} \left\{ 2\epsilon(G^2)_{c_2c_1} +2\epsilon(2-\omega) G_{c_2c_1}
-  \mathrm{str}(G^2) \varepsilon_{c_2c_1} \right\}.
\end{equation}
Comparing the left and right hand-side, we arrive at the statement of the proposition.
\end{proof}

The fusion of $ L(u) = u I - \frac{1}{2}F^{ab}M_{ba}$ and $ \tilde L(u) = u I + \frac{1}{2} (F^{t})^{ab}M_{ba}$
with respect to the super-spinor representations generated by $F, F^t$
results in an $L$ operator obeying the RLL relation with the vector (fundamental) $R$
matrix (\ref{eq:RLL}). It is qudratic in $u$ and equivalent to the form
\eqref{Lu2} with $N$ of the form (\ref{eq:N}),
$$ L(u) = u^2 \cdot\mathbf{1} + u\cdot M + N, \qquad N = \half (M^2 + \beta M),  $$
 shown above to obey the conditions
A--C. It obeys also the remaining conditions D--F.
The proof  can be done
by direct calculations using the relations (\ref{GG}, \ref{M3}).

\section{Discussion}

Yang-Baxter relations with orthosymplectic supersymmetry, in particular the
ones involving the fundamental $R$ matrix, can be written in a similar form
like the ones with orthogonal or symplectic symmetry. The formulation
presented in this paper provides a systematic treatment and displays
explicitly the features distinguishing the $osp$ case from the $so$ and $sp$
cases.

We have pointed out that the invariant tensors appearing in the fundamental
$R$ matrix represent the Brauer algebra.

$L$ operators can have a simple form in distinguished representations. We
have identified the superspinor representation resulting in an $L$ operator
linear in the spectral parameter being the generalization of the spinor
representation of the $so$ case and the metaplectic representation in the
$sp$ case.

The super spinorial $R$ operator 
 which intertwines super-spinor representations 
 (see Proposition {\bf \ref{propRosc}})
has been constructed by
the generalizing the methods developed
in \cite{CDI}, \cite{CDI2}, \cite{IsKarKir15} for the $so$ and $sp$ cases.

 %\marginpar{\it RK, added}
 The  superspinorial $RLL$ relation holds for  $L$ operators  linear
 in the
spectral parameter and acting in the spinor and the vector (fundamental)
representations. It also holds for generalized $L$ operators where the
vector (fundamental) representation is replaced by another one obeying a constraint
 (\ref{GG}) represented 
 in the form of a super anticommutator of the elements $G_{ab}$
 of the matrix of generators.
All these results were summarized in the Proposition {\bf \ref{propRosc}}.

We have investigated the case of the second order Yangian evaluation, in
particular the
solution for the $L$ operators with all terms expressed as function of the
Lie algebra generator matrix $G$. Its second non-trivial term is
proportional to the supertraceless part of $G^2$.  The Lie algebra
representation generated by the matrix elements of $G$ is constraint in such
a way that $G$ obeys a condition in terms of a cubic characteristic
polynomial. The latter condition is related to the super anticommutator
condition appearing in connection with the spinorial Yang Baxter relation.
The class of Lie algebra representations constructed by the Jordan-Schwinger
ansatz based on graded Heisenberg pairs obeys these constraints.

\vspace{0.4cm}

\noindent
{\bf Acknowledgment.} We thank S.Derkachov for valuable discussions.

The work of J.F. was supported by the Grant Agency of the Czech Technical University in Prague, grant No. SGS15/215/OHK4/3T/14
and by the Grant of the Plenipotentiary of the Czech Republic at JINR, Dubna.

The work of A.P.I. was supported by Russian Science Foundation
grant 14-11-00598 (Sections 1-5)
and by RFBR grants 16-01-00562-a, 15-52-05022 Arm-a (Sections 6-8).

The work of D.K. was partially supported by the Armenian
State Committee of Science grant SCS 15RF-039. It was done
within programs of the ICTP Network NET68 and of the
Regional Training Network on Theoretical Physics
sponsored by Volkswagenstiftung Contract nr. 86 260.

 Our collaboration was also supported by JINR (Dubna) via the
programs Heisenberg-Landau (J.F. and R.K.) and
Smorodinski-Ter-Antonyan (D.K.).

\newpage
\appendix

\section{The graded tensor product and  Yang-Baxter relations}
\label{sec:GTP}
\setcounter{equation}{0}

There appear different conventions in the literature regarding
the R-matrices and Yang-Baxter equations.
In this section we intend to relate the Yang-Baxter equation used,
e.g., in \cite{Ragoucy} to the Yang-Baster equation \eqref{eq:YBEgraded}.

Let $\cal V$ be a  vector superspace $\cal V$ with the basis $\{\ket{e_1},\dots,
\ket{e_{N+M}}\}$, where $\{\ket{e_1}\!,\dots ,\!\ket{e_{N}} \! \}$ are the basis vectors of the even part
of $\cal V$ and $\{\ket{e_{N+1}},\dots,\ket{e_{N+M}}\}$ are the basis vectors of the odd part
of $\cal V$. The basis of the dual superspace $\overline{\cal{ V}}$ is
$\{\langle e^1|,\dots, \langle e^{N+M}|\}$ with the even part $\{\langle e^1|,\dots ,\langle e^{N}|\}$ and the
odd part $\{\langle e^{N+1}|,\dots,\langle e^{N+M}|\}$. We demand that these two bases are dual in
the following sense:
\begin{equation}
\langle e^a|e_b\rangle=\delta^a_b, \qquad \forall\, a,b=1,\dots,N+M.
\end{equation}
Unlike the formulation used in the main part of this paper
 the gradation is now carried by the basis vectors, i.e., $\mathrm{grad}(\ket{e_a}) =\mathrm{grad}(\bra{e^a})=[a]$, whereas the coordinates are ordinary numbers from the field $\mathbb{F}$ over which the superspace $\cal V$ is constructed (compare with the approach introduced in this article, especially in section {\bf \ref{supgr}}).
The matrix units and the identity operator on $\cal V$ can be expressed as:
\begin{equation}
E_a^{\ b} =\ket{e_a} \! \langle e^b|, \qquad I=\sum_{a=1}^{N+M} \ket{e_a} \! \langle e^{a}|.
\end{equation}
The gradation of $E_a^{\ b}$ is $[a]+[b]$ and $I$ is the even operator.

The matrix elements of the operator $A:{\cal V} \rightarrow {\cal V}$ are
\begin{equation}
A^a_{\ b} = \langle e^a| A \ket{ e_b}
\end{equation}
and one can immediately check that
\begin{equation}
A = \sum_{a,b} \ket{e_a} \! A^a_{\ b} \langle e^b| = \sum_{a,b} A^a_{\ b} E_a^{\ b}.
\end{equation}

We introduce the graded tensor product of the superspaces ${\cal V}\otimes {\cal V}$
\eqref{gradten}.
The basis of ${\cal V}\otimes {\cal V}$ is
$\{\ket{ e_{b_1}}\otimes \ket{e_{b_2}}\}_{b_1,b_2=1}^{N+M}$.
Its dual basis in  $\overline{\cal V}\otimes \overline{\cal V}$
is $\{ (-1)^{[a_1][a_2]} \langle e^{a_1}| \otimes \langle e^{a_2}| \}_{a_1,a_2=1}^{N+M}$ as
can be easily seen,
\begin{equation}
 (-1)^{[a_1][a_2]} (\langle e^{a_1}| \otimes \langle e^{a_2}|) (\ket{e_{b_1}}\otimes \ket{e_{b_2}}) = (-1)^{[a_1][a_2]+[a_2][b_1]} \langle e^{a_1} \! \ket{e_{b_1}}\otimes \langle  e^{a_2}\!\ket{e_{b_2}} = \delta^{a_1}_{b_1} \delta^{a_2}_{b_2}.
\end{equation}
The operator $R$ acting in ${\cal V}\otimes {\cal V}$ has the components
w.r.t. the above basis of the form
\begin{equation}
R^{a_1 a_2}_{b_1 b_2} = (-1)^{[a_1][a_2]} (\langle e^{a_1}|\otimes \langle e^{a_2}|) R\, (\ket{e_{b_1}}\otimes \ket{e_{b_2}})
\end{equation}
and satisfies
\begin{equation}
R = \sum R^{a_1 a_2}_{b_1 b_2}  (\ket{ e_{a_1}} \otimes \ket{ e_{a_2}})(-1)^{[b_1][b_2]}(\langle e^{b_1}|\otimes \langle e^{b_2}|) = \sum  R^{a_1 a_2}_{b_1 b_2} (-1)^{[b_1][b_2]+[a_2][b_1]} E_{a_1}^{\ b_1}\otimes E_{a_2}^{\ b_2}.
\end{equation}

The graded permutation is defined as
\begin{equation}
P \ket{e_a} \otimes \ket{e_b} = (-1)^{[a][b]} \ket{e_b} \otimes \ket{e_a}.
\end{equation}
In view of the above considerations, it has the following components
\begin{equation}
P^{a_1 a_2}_{b_1 b_2} = (-1)^{[a_1][a_2]} \delta^{a_1}_{b_2} \delta^{a_2}_{b_1}
\end{equation}
as expected (compare with \eqref{PP12}). It can be expressed using the matrix units as
\begin{equation}
P = \sum_{a,b} (-1)^{[b]} E_{a}^{\ b}\otimes E_{b}^{\ a}.
\end{equation}
 The identity operator on ${\cal V}\otimes {\cal V}$ is
\begin{equation}
I = \sum_{a,b} E_{a}^{\ a}\otimes E_b^{\ b} = \sum_{a,b} (-1)^{[a][b]} (\ket{e_a}\otimes \ket{e_b})(\langle e^a| \otimes \langle e^b|).
\end{equation}

This formalism can be obviously extended to ${\cal V}^{\otimes n}$  for arbitrary $n$.
Due to the Yang-Baxter relation
we need to discuss the situation ${\cal V}^{\otimes 3}$.
Its basis is $\{\ket{e_{b_1}}\otimes \ket{e_{b_2}}\otimes  \ket{e_{b_3}}\}_{b_1\!,b_2\!,b_3=1}^{N+M}$
and the corresponding basis of the dual superspace $\overline{\cal V}^{\otimes 3}$ is
$$\left\{ (-1)^{[a_1][a_2]+[a_1][a_3]+[a_2][a_3]} \langle e^{a_1}|
\otimes \langle e^{a_2}|\otimes  \langle e^{a_3}|\right\}_{a_1,a_2,a_3=1}^{N+M}.$$
Let us remark that the identity operator in ${\cal V}^{\otimes 3}$
is
\begin{equation}\label{Iaux3}
I = \sum_{a,b,c} E_{a}^{a} \otimes E_b^{\ b} \otimes E_{c}^{\ c} = \sum_{a,b,c} (-1)^{[a][b]+[a][c]+[b][c]} (\ket{e_a}\otimes \ket{e_b}\otimes \ket{e_c})(\langle e^a|\otimes \langle e^b| \otimes \langle e^c|).
\end{equation}
It is useful to use the shorthand notation
\begin{equation}\label{eshort}
\langle e^{a_1a_2a_3}| = \langle e^{a_1}|\otimes \langle e^{a_2}|\otimes \langle e^{a_3}|, \qquad \ket{e_{b_1b_2b_3}} = \ket{e_{b_1}}\otimes \ket{e_{b_2}}\otimes  \ket{ e_{b_3}}.
\end{equation}

The YB equation appearing, e.g., in \cite{Ragoucy} is of the form
\begin{equation} \label{YBeRag}
R_{12}R_{13}R_{23} = R_{23}R_{13}R_{12}.
\end{equation}
We show here that if we write it in components,
we obtain our form of the Yang-Baxter relation \eqref{eq:YBEgraded}
provided that the R-matrix is even
(see the definition of the even R-matrix \eqref{Reven1}).
The left hand side of \eqref{YBeRag} has the component form
\begin{align}
&(R_{12}R_{13}R_{23})^{a_1a_2a_3}_{b_1b_2b_3} = (-1)^{[a_1][a_2]+[a_1][a_3]+[a_2][a_3]}  \langle e^{a_1a_2a_3}| R_{12}R_{13}R_{23} \ket{ e_{b_1b_2b_3}}
\intertext{and we embed the identity operator \eqref{Iaux3}
in ${\cal V}^{\otimes 3}$ between $R_{12}$ and $R_{13}$}
&=(-1)^{[a_1][a_2]+[a_1][a_3]+[a_2][a_3]+[c_1][c_2]+[c_1][c_3]+[c_2][c_3]}
\langle e^{a_1a_2a_3}| R_{12} \ket{e_{c_1c_2c_3}} \notag\\
& \hspace{10cm} \times \langle e^{c_1c_2c_3}| R_{13} R_{23}  \ket{e_{b_1b_2b_3}}
\notag \\
&= (-1)^{[a_1][a_3]+[a_2][a_3]+[c_1][c_2]}
R^{a_1a_2}_{c_1c_2} \langle e^{c_1c_2a_3}| R_{13}R_{23} \ket{e_{b_1b_2b_3}} \notag
\intertext{where we used the properties of the graded tensor product
\eqref{gradten}. We embed the identity operator
\eqref{Iaux3} between $R_{13}$ and $R_{23}$ and obtain }
&(R_{12}R_{13}R_{23})^{a_1a_2a_3}_{b_1b_2b_3} =(-1)^{[a_1][a_3]+[a_2][a_3]+[c_1][c_2]+[d_1][d_2]+[d_1][d_3]+[d_2][d_3]}
R^{a_1a_2}_{c_1c_2} \notag\\
& \hspace{7cm} \times \langle e^{c_1c_2a_3}|R_{13} \ket{e_{d_1d_2d_3}} \langle e^{d_1d_2d_3}| R_{23}
\ket{ e_{b_1b_2b_3}} \notag \\
&=(-1)^{[a_1][a_3]+[a_2][a_3]+[c_1][c_2]+[c_1][a_3]+[c_2][a_3]+[d_1][d_3]+[c_2][d_3]}
R^{a_1a_2}_{c_1c_2} R^{c_1a_3}_{d_1d_3}\langle e^{d_1c_2d_3}|R_{23} \ket{e_{b_1b_2b_3}}. \notag \\
\intertext{For the even R-matrix we obtain}
&(R_{12}R_{13}R_{23})^{a_1a_2a_3}_{b_1b_2b_3} = (-1)^{[c_1][c_2]+[d_1][d_3]+[c_2][d_3]} R^{a_1a_2}_{c_1c_2}
R^{c_1a_3}_{d_1d_3} \langle e^{d_1c_2d_3}|R_{23} \ket{e_{b_1b_2b_3}} \notag \\
&\hspace{3.1cm}= R^{a_1a_2}_{c_1c_2}\, (-1)^{[c_1][c_2]}\,  R^{c_1a_3}_{b_1d_3} \,
(-1)^{[b_1][c_2]}\, R^{c_2d_3}_{b_2b_3}
\end{align}
which is exactly the left hand side of \eqref{eq:YBEgraded}.
Similarly, the right hand side of \eqref{YBeRag} has the component form
\begin{equation}
(R_{23}R_{13}R_{12})^{a_1a_2a_3}_{b_1b_2b_3} =
R^{a_2a_3}_{c_2c_3}(-1)^{[a_1][c_2]} R^{a_1c_3}_{d_1b_3}(-1)^{[d_1][c_2]}
R^{d_1c_2}_{b_1b_2}
\end{equation}
which coincides with the right hand side of \eqref{eq:YBEgraded}.
Thus, the equivalence of \eqref{YBeRag} and \eqref{eq:YBEgraded} is established.
We recall that this equivalence holds due to  the R-matrix being even.

\section{Properties of operators ${\cal P}$, ${\cal K}$ \label{PKprop}}
\setcounter{equation}{0}

We use here the concise matrix notation introduced
 in sections {\bf \ref{supgr}, \ref{sec2}}
  (in bosonic case this notation was
 proposed in \cite{FRT}).
 Matrices (\ref{osp07}) satisfy identities
$$
 {\mathcal{P}}_{12} = {\mathcal{P}}_{21} \; , \;\;\;
 {\mathcal{K}}_{12} =
 (-)^{12}{\mathcal{K}}_{21}  (-)^{12} \; , \;\;\;
 (-)^{1} {\mathcal{K}}_{12}  =
 (-)^{2} {\mathcal{K}}_{12} \; , \;\;\;
 {\mathcal{K}}_{12} (-)^{1}  =
  {\mathcal{K}}_{12} (-)^{2} \; ,
$$
\begin{equation}
 \label{ident00}
 \begin{array}{c}
{\mathcal{P}}_{12}{\mathcal{P}}_{12}= {\bf 1} \; , \;\;\;
{\mathcal{K}}_{12}{\mathcal{K}}_{12}=
\omega {\mathcal{K}}_{12} \; , \;\;\;
{\mathcal{K}}_{12}{\mathcal{P}}_{12} =
{\mathcal{P}}_{12}{\mathcal{K}}_{12} = \epsilon {\mathcal{K}}_{12},
\end{array}
 \end{equation}
 where $\omega = \epsilon(N-M)$ and
 $(-)^{i} = (-1)^{[i]} \delta^{a_i}_{b_i}$ is the matrix of super-trace in the $i$-th super-space ${\cal V}_{(N|M)}$.
 Then we have
   \begin{equation}
 \label{ident16}
 \begin{array}{c}
  (-)^1 {\mathcal{P}}_{12} = {\mathcal{P}}_{12} (-)^2
   \; , \;\;\;
{\mathcal{P}}_{12}{\mathcal{P}}_{23}=
(-)^{12}(-)^{23} {\mathcal{P}}_{13} {\mathcal{P}}_{12}=
 {\mathcal{P}}_{23} {\mathcal{P}}_{13} (-)^{12}(-)^{23}\; ,
  \\ [0.2cm]
 {\mathcal{P}}_{12}{\mathcal{K}}_{13}=
(-)^{12} {\mathcal{K}}_{23}(-)^{12} {\mathcal{P}}_{12} \; , \;\;\; {\mathcal{P}}_{12}(-)^{12} {\mathcal{K}}_{13} (-)^{12}=
{\mathcal{K}}_{23} {\mathcal{P}}_{12} \; ,
\end{array}
 \end{equation}
 \begin{equation}
 \label{ident15}
 \begin{array}{c}
 \epsilon {\mathcal{K}}_{12}{\mathcal{P}}_{31}=
{\mathcal{K}}_{12}(-)^{12} {\mathcal{K}}_{32} (-)^{12} \; , \;\;\;
\epsilon {\mathcal{P}}_{31}{\mathcal{K}}_{12}=
 (-)^{12} {\mathcal{K}}_{32} (-)^{12} {\mathcal{K}}_{12}\; ,
 \\ [0.2cm]
{\mathcal{K}}_{12}{\mathcal{K}}_{31} = \epsilon
{\mathcal{K}}_{12}(-)^{12}{\mathcal{P}}_{32}(-)^{12} \; , \;\;\;
{\mathcal{K}}_{31}{\mathcal{K}}_{12} = \epsilon
 (-)^{12}{\mathcal{P}}_{32}(-)^{12}{\mathcal{K}}_{12}.
\end{array}
 \end{equation}
 Identities (\ref{ident16}) follow from the representation
 (\ref{PP12}): ${\mathcal{P}}_{12}= (-)^{12} P_{12}=
 P_{12}(-)^{12}$, where $P_{12}$ is the usual permutation operator. Identities (\ref{ident15}) follow
 from the definitions (\ref{osp07}),
 (\ref{PP12}), (\ref{KK12}) of the operators ${\mathcal{P}}$
 and ${\mathcal{K}}$. We prove only the last equality
 in (\ref{ident15}) since the other identities
 in (\ref{ident15}) can be proved in the same way.
 We denote incoming matrix indices by $a_1,a_2,a_3$ and outcoming indices by
$c_1,c_2,c_3$ while dummy indices are $b_i$ and $d_i$.
Then we have
 $$
 \begin{array}{c}
 ({\mathcal{K}}_{31}
 {\mathcal{K}}_{12})^{a_1 a_2 a_3}_{c_1 c_2 c_3} =
 \varepsilon^{a_3 a_1} \varepsilon_{c_3 b_1}
 \varepsilon^{b_1 a_2} \varepsilon_{c_1 c_2} =
 \varepsilon^{a_3 a_1} \delta_{c_3}^{a_2} \varepsilon_{c_1 c_2}
  = \delta_{c_3}^{a_2} \delta^{a_3}_{b_2}
  \epsilon (-)^{[a_1][b_2]} \varepsilon^{a_1b_2}
  \varepsilon_{c_1 c_2} = \\ [0.2cm]
  = \epsilon (-1)^{[a_2][a_3]}
  ({\mathcal{P}}_{23})^{a_2a_3}_{b_2c_3}
  %(-1)^{[c_3][a_2]}
  %\delta_{c_3}^{a_2} \delta^{a_3}_{b_2}
  %\epsilon (-1)^{[a_1][b_2]} \varepsilon^{a_1b_2}
  %\varepsilon_{c_1 c_2}
  (-1)^{[a_1][b_2]}
  ({\mathcal{K}}_{12})^{a_1b_2}_{\;\; c_1 c_2} =
  ((-)^{23} {\mathcal{P}}_{23} (-)^{12}
  {\mathcal{K}}_{12})^{a_1 a_2 a_3}_{\;\; c_1 c_2 c_3} \; ,
  \end{array}
 $$
 and in view of the relation $(-)^{23} {\mathcal{K}}_{31}=
 (-)^{12} {\mathcal{K}}_{31}$
 which follows from (\ref{SuperMet2})
 we obtain the last formula  in (\ref{ident15}).

By means of the relations (\ref{ident16}), (\ref{ident15})
 one can immediately check eqs. (\ref{osp08}),
 (\ref{osp09}) and also deduce
 \begin{equation}
 \label{ident01}
{\mathcal{P}}_{12}{\mathcal{P}}_{23}{\mathcal{P}}_{12}=
{\mathcal{P}}_{23}{\mathcal{P}}_{12}{\mathcal{P}}_{23}.
 \end{equation}
 \begin{equation}
 \label{ident02}
{\mathcal{K}}_{12}{\mathcal{K}}_{23}{\mathcal{K}}_{12}=
{\mathcal{K}}_{12}, \;\;\;\;
{\mathcal{K}}_{23}{\mathcal{K}}_{12}{\mathcal{K}}_{23}=
{\mathcal{K}}_{23},
\end{equation}
 \begin{equation}
 \label{ident05}
{\mathcal{P}}_{12}{\mathcal{K}}_{23}{\mathcal{K}}_{12}=
{\mathcal{P}}_{23}{\mathcal{K}}_{12} \; ,  \;\;\;
{\mathcal{K}}_{12}{\mathcal{K}}_{23}{\mathcal{P}}_{12}=
{\mathcal{K}}_{12}{\mathcal{P}}_{23},
\end{equation}
\begin{equation}
 \label{ident03}
{\mathcal{P}}_{23}{\mathcal{K}}_{12}{\mathcal{K}}_{23}=
{\mathcal{P}}_{12}{\mathcal{K}}_{23} \; , \;\;\;
{\mathcal{K}}_{23}{\mathcal{K}}_{12}{\mathcal{P}}_{23}=
{\mathcal{K}}_{23}  {\mathcal{P}}_{12} \; .
\end{equation}
Identity (\ref{ident01}) follows from
the relations in the first line of (\ref{ident16}).
We consider few examples in (\ref{ident02})-(\ref{ident03})
 in details.
We start to prove the first relation in (\ref{ident02}):
 %For the left hand side
 %${\mathcal{K}}_{12}{\mathcal{K}}_{23}{\mathcal{K}}_{12}$
  %we have
$$
({\mathcal{K}}_{12}{\mathcal{K}}_{23}
{\mathcal{K}}_{12})^{a_1a_2a_3}_{c_1c_2c_3} =
\varepsilon^{a_1a_2}\varepsilon_{b_1b_2}
\varepsilon^{b_2a_3}\varepsilon_{d_2c_3}
\varepsilon^{b_1d_2}\varepsilon_{c_1c_2} =
\varepsilon^{a_1a_2}\delta_{b_1}^{a_3}\delta_{c_3}^{b_1}
\varepsilon_{c_1c_2} =
{\cal K}^{a_1a_2}_{\;\; c_1c_2}\delta_{c_3}^{a_3}\; .
$$
The second relation in (\ref{ident02})
can be proved in the same way. Then
we prove the first
 equation in (\ref{ident05}). For the left hand side
  of (\ref{ident05}) one has:
$$
\begin{array}{c}
({\mathcal{P}}_{12}{\mathcal{K}}_{23}
{\mathcal{K}}_{12})^{a_1a_2a_3}_{\;\; c_1c_2c_3} =
(-1)^{[a_1][a_2]}\delta^{a_1}_{b_2}\delta^{a_2}_{b_1}
\varepsilon^{b_2a_3}\varepsilon_{d_2c_3}
\varepsilon^{b_1d_2} \varepsilon_{c_1c_2}=
(-1)^{[a_1][a_2]}\varepsilon^{a_1a_3}
\delta^{a_2}_{c_3} \varepsilon_{c_1c_2} = \\ [0.2cm]
 = \delta^{a_2}_{c_3}  \delta^{a_3}_{b_2}
 (-1)^{[a_1][c_3]} \varepsilon^{a_1b_2} \varepsilon_{c_1c_2}
 = \delta^{a_2}_{c_3}  \delta^{a_3}_{b_2}
 (-1)^{[b_2][c_3]} \varepsilon^{a_1b_2} \varepsilon_{c_1c_2} =
 ({\mathcal{P}}_{23}
 {\mathcal{K}}_{12})^{a_1a_2a_3}_{\;\; c_1c_2c_3} \; ,
\end{array}
$$
and similarly one deduces other relations
in (\ref{ident05}) and (\ref{ident03}).
From the identities (\ref{ident01}) -- (\ref{ident03})
 we also deduce the following relations
 \begin{equation}
 \label{ident12}
 {\mathcal{K}}_{12}{\mathcal{P}}_{23}{\mathcal{K}}_{12}=
\epsilon{\mathcal{K}}_{12}, \;\;\;
{\mathcal{K}}_{23}{\mathcal{P}}_{12}{\mathcal{K}}_{23}=
\epsilon{\mathcal{K}}_{23}.
\end{equation}
 \begin{equation}
 \label{ident11}
{\mathcal{P}}_{12}{\mathcal{K}}_{23}{\mathcal{P}}_{12}=
{\mathcal{P}}_{23}{\mathcal{K}}_{12}{\mathcal{P}}_{23}.
 \end{equation}
 \begin{equation}
 \label{ident04}
{\mathcal{P}}_{12}{\mathcal{P}}_{23}{\mathcal{K}}_{12}=
{\mathcal{K}}_{23}{\mathcal{P}}_{12}{\mathcal{P}}_{23},
\;\;\;\;
{\mathcal{K}}_{12}{\mathcal{P}}_{23}{\mathcal{P}}_{12}=
{\mathcal{P}}_{23}{\mathcal{P}}_{12}{\mathcal{K}}_{23}.
\end{equation}
Indeed, if we act from the left
on both sides of the first relation
in (\ref{ident05}) by ${\mathcal{K}}_{12}$
and use (\ref{ident00}), (\ref{ident02}) we obtain
the first relation in (\ref{ident12}). In the same way
one can deduce from the first relation in (\ref{ident03})
the second relation in (\ref{ident12}). Now we act
on both sides of (\ref{ident05}) by ${\mathcal{P}}_{23}$
 from the right and use the last equation in (\ref{ident03}).
  As a result we arrive at the identity (\ref{ident11}).
  Finally the relations (\ref{ident04})  trivially follow from eq. (\ref{ident11}).

At the end of this appendix we stress that
identities (\ref{ident00}),
(\ref{ident01}) -- (\ref{ident03}) are images
of the defining relations (\ref{defBrauer2})
for the Brauer algebra in the
representation (\ref{Brauer}). The $R$-matrix
(\ref{RBrauer01}) is the image of the element (\ref{RBrauer})
 and the Yang-Baxter equation (\ref{eq:YBEbraid})
 is the image of the identity (\ref{YBEbr}).
Thus,  it follows from proposition
 {\bf \ref{Prop11}} that
 the $R$-matrix (\ref{RBrauer01}) is a solution
 of the braided version of the Yang-Baxter equation (\ref{eq:YBEbraid}).

\section{Direct proof of proposition \ref{propFc}}
\label{AppInv}
\setcounter{equation}{0}

We shall use the advantage of the generating functions language developed in subsection
{\bf \ref{sec:auxvar}}.
and shall  work with two sets of auxiliary variables $\kappa^a,\kappa'^b$
with the corresponding
derivatives $\partial^a,\partial'^b$. Then
\begin{align}
& \left[\varepsilon_{a_1b_1}\dots \varepsilon_{a_k b_k} c_1^{(a_1}\cdots c_1^{ a_k)}
c_2^{(b_k} \cdots c_2^{b_1)} ,c_1^{(a}c_1^{b)}+c_2^{(a}c_2^{b)}\right] = & \notag \\
& =\varepsilon_{a_1b_1}\dots \varepsilon_{a_kb_k}\left\{ (-1)^{([a]+[b])
([b_1]+\cdots+[b_k])} \left[c_1^{(a_1}\cdots c_1^{a_k)}, c_1^{(a}c_1^{b)}\right]_\pm
c_2^{(b_k}\cdots c_2^{b_1)} +
 \right. & \notag \\
& \quad \left. + c_1^{(a_1}\cdots c_1^{a_k)}
\left[ c_2^{(b_k}\cdots c_2^{b_1)}, c_2^{(a}c_2^{b)} \right]_\pm \right\} =
\varepsilon_{a_1b_1}\dots \varepsilon_{a_kb_k} \times \\
& \quad \times \Big\{ \underline{ (-1)^{([a]+[b])([b_1]+\cdots+[b_k])}
\partial^{a_1}\cdots\partial^{a_k}} \left( \underline{\epsilon (-1)^{[a]}
\kappa^a\partial^b} -(-1)^{[a][b]+[b]}\kappa^b\partial^a\right)
\partial'^{b_k}\cdots\partial'^{b_1} & \notag\\
& \quad+ \underline{ \partial^{a_1}\cdots\partial^{a_k}
\partial'^{b_k}\cdots\partial'^{b_1}} \left( \epsilon (-1)^{[a]}
\kappa'^a\partial'^b -\underline{ (-1)^{[a][b]+[b]}\kappa'^b\partial'^a} \right) \Big\}
e^{(\kappa\cdot c_1)} e^{(\kappa'\cdot c_2)} \Big|_{\kappa,\kappa'=0}, & \notag
\end{align}
where we used the supercommutation relations \eqref{eq:CommBasis}.
We  need two identities:
\begin{align}
 \partial^{a_1}\cdots\partial^{a_k} \kappa^a &=
\sum_{j=1}^k  (-\epsilon)^{k-j} (-1)^{\sum_{i=j+1}^k [a_i] [a]}  \
\varepsilon^{a a_i} \  \partial^{a_1}\cdots\partial^{a_{j-1}}\partial^{a_{j+1}}
\cdots\partial^{a_k} + & \notag\\
& \quad + (-\epsilon)^k (-1)^{[a]\sum_{i=1}^k [a_i]} \kappa^a \partial^{a_1}
\cdots\partial^{a_k}, & \\
 \partial^{b_k}\cdots\partial^{b_1} \kappa^b &= \sum_{j=1}^k  (-\epsilon)^{j-1}
(-1)^{\sum_{i=1}^{j-1} [b_i] [b]} \  \varepsilon^{b b_j} \
\partial^{b_k}\cdots\partial^{b_{j+1}}\partial^{b_{j-1}}\cdots\partial^{b_1}+ &
\notag \\
& \quad+ (-\epsilon)^k (-1)^{[b]\sum_{i=1}^k [b_i]} \kappa^b \partial^{b_k}
\cdots\partial^{b_1}. &
\end{align}
We show now that the two underlined terms cancel. We write them here without the
factor $e^{(\kappa\cdot c_1)} e^{(\kappa'\cdot c_2)} |_{\kappa,\kappa'=0}.$
The first underlined term is
\begin{align}
& \epsilon (-1)^{[a]} \varepsilon_{a_1b_1}\dots \varepsilon_{a_kb_k}
(-1)^{([a]+[b])([b_1]+\cdots+[b_k])} \partial^{a_1}\cdots\partial^{a_k}
\kappa^a\partial^b \partial'^{b_k}\cdots\partial'^{b_1} = & \notag\\
&=\epsilon (-1)^{[a]} \varepsilon_{a_1b_1}\dots \varepsilon_{a_kb_k}
(-1)^{([a]+[b])([b_1]+\cdots+[b_k])} \;\; \times & \notag \\
& \quad \times \sum_{j=1}^k (-\epsilon)^{k-j} (-1)^{\sum_{i=1}^k [a][a_i]}
\varepsilon^{a a_j} \partial^{a_1}\cdots \partial^{a_{j-1}} \partial^{a_{j+1}}
\cdots \partial^{a_k} \partial^b \partial'^{b_k}\cdots\partial'^{b_1} +\mathcal{Z}_1
= & \notag\\
& = \epsilon (-1)^{[a][b]} \sum_{j=1}^k (-1)^{([a]+[b])([a_1]+\cdots+[a_{j-1}])}
\varepsilon_{a_{1}b_{1}} \dots \varepsilon_{a_{j-1}b_{j-1}} \varepsilon_{a_{j+1}b_{j+1}}
\dots \varepsilon_{a_{k}b_{k}} & \notag\\
& \qquad\qquad \qquad \times\partial^{a_1}\cdots \partial^{a_{j-1}}
\partial^b \partial^{a_{j+1}} \cdots \partial^{a_k} \partial'^{b_k}\cdots
\partial'^{b_{j+1}}\partial'^a \partial'^{b_{j-1}}\cdots\partial'^{b_1} +\mathcal{Z}_1.
&
\end{align}
The second underlined term is
\begin{align}
& -(-1)^{[a][b]+[b]}\varepsilon_{a_1b_1}\dots \varepsilon_{a_kb_k}
\partial^{a_1}\cdots \partial^{a_k} \partial'^{b_k}\cdots\partial'^{b_1}
\kappa'^b\partial'^a = & \notag\\
& = -(-1)^{[a][b]+[b]} \varepsilon_{a_1b_1}\dots \varepsilon_{a_kb_k} \partial^{a_1}
\cdots \partial^{a_k} \;\; \times & \notag \\
& \quad \times
\sum_{j=1}^k  (-\epsilon)^{j-1}  (-1)^{\sum_{i=1}^{j-1} [b_i] [b]} \
\varepsilon^{b b_j} \ \partial'^{b_k}\cdots\partial'^{b_{j+1}}\partial'^{b_{j-1}}
\cdots\partial'^{b_1} \partial'^a +\mathcal{Z}_2 = & \notag\\
& = -\epsilon (-1)^{[a][b]} \sum_{j=1}^k (-1)^{([a]+[b])([b_1]+\cdots+[b_{j-1}])}
\varepsilon_{a_{1}b_{1}} \dots \varepsilon_{a_{j-1}b_{j-1}} \varepsilon_{a_{j+1}b_{j+1}}
\dots \varepsilon_{a_{k}b_{k}} & \notag\\
& \qquad\qquad\qquad\times \partial^{a_1}\cdots \partial^{a_{j-1}}
\partial^b \partial^{a_{j+1}} \cdots \partial^{a_k} \partial'^{b_k}\cdots
\partial'^{b_{j+1}}\partial'^a \partial'^{b_{j-1}}\cdots\partial'^{b_1} +
\mathcal{Z}_2.&
\end{align}
Let us remark that $\mathcal{Z}_1,\mathcal{Z}_2$ are proportional to
$\kappa,\kappa'$ respectively and, therefore, vanish.
As we see, the two underlined terms really cancel. The two non-underlined terms
cancel, too.

%%%%%%%%%%%%%%%%%%%%%%%%%%%%%%%%%%%%%%%%%%%%%%%%%%%%%%%%%%%%%%
%%%%%%%%%%%%%%%%%%%%%%%%%%%%%%%% THE BIBLIOGRAPHY %%%%%%%%%%%%%%%%%%%%%%%%%%%%%%%%%%%%%%%%%%%%%%%%%%%%%%%%%%%%%%%%%%%%%%%%%%

\end{document}